\documentclass[aps,prx,twocolumn,superscriptaddress,nofootinbib]{revtex4}

\usepackage{graphicx} 
\graphicspath{{figures/}} 

\usepackage{mathpazo}
\usepackage{amsmath}
\usepackage{mathrsfs}
\usepackage{amsthm}
\usepackage{amsfonts,amssymb,amsthm, amscd, bbm,braket}
\usepackage{graphicx}   
\usepackage{subcaption}  
\usepackage{amsbsy} 
\usepackage{bm} 
\usepackage[bold]{hhtensor} 
\usepackage{algorithm}
\usepackage[noend]{algpseudocode}
\usepackage{algpseudocode}

\usepackage{multirow}
\usepackage{tcolorbox}

\newcommand{\Tr}{\mbox{Tr}}

\newcommand{\rtht}{\rho_\theta}
\newcommand{\pstht}{\psi_\theta}
\newcommand{\partht}{\partial_\theta}
\newcommand{\Ltht}{\mathcal{L}_\theta}
\newcommand{\half}{\frac{1}{2}}
\newcommand{\comm}[1]{}

\renewcommand\bra[1]{{\langle{#1}|}}
\renewcommand\ket[1]{{|{#1}\rangle}}

\newtheorem{theorem}{Theorem}
\newtheorem{lemma}{Lemma}
\newtheorem{proposition}{Proposition}

\definecolor{burntorange}{rgb}{0.8, 0.33, 0.0}\usepackage[colorlinks,citecolor=blue,linkcolor=red,urlcolor=blue]{hyperref}

\makeatletter

\definecolor{DYblue}{RGB}{1,100,0}

\begin{document}

\title{Optimal estimation of pure states with displaced-null measurements}

\author{Federico Girotti}
\affiliation{School of Mathematical Sciences, University of Nottingham, United Kingdom}
\affiliation{Centre for the Mathematics and Theoretical Physics of Quantum Non-Equilibrium Systems,
University of Nottingham, Nottingham, NG7 2RD, UK}
\affiliation{Department of Mathematics, Polytechnic University of Milan, Milan, Piazza L. da Vinci 32, 20133, Italy}

\author{Alfred Godley}
\affiliation{School of Mathematical Sciences, University of Nottingham, United Kingdom}
\affiliation{Centre for the Mathematics and Theoretical Physics of Quantum Non-Equilibrium Systems,
University of Nottingham, Nottingham, NG7 2RD, UK}

\author{M\u{a}d\u{a}lin Gu\c{t}\u{a}}
\affiliation{School of Mathematical Sciences, University of Nottingham, United Kingdom}
\affiliation{Centre for the Mathematics and Theoretical Physics of Quantum Non-Equilibrium Systems,
University of Nottingham, Nottingham, NG7 2RD, UK}

-----------------------------------------------------------------------------------------------------------

\begin{abstract}
We revisit the problem of estimating an unknown parameter of a pure quantum state, and investigate `null-measurement' strategies in which the experimenter aims to measure in a basis that contains a vector close to the true system state. Such strategies are known to approach the quantum Fisher information for models where the quantum Cram\'{e}r-Rao bound is achievable but a detailed adaptive strategy for achieving the bound in the multi-copy setting has been lacking.

We first show that the following naive null-measurement implementation fails to attain even the standard estimation scaling: estimate the parameter on a small sub-sample, and apply the null-measurement corresponding to the estimated value on the rest of the systems. This is due to non-identifiability issues specific to null-measurements, which arise when the true and reference parameters are close to each other. To avoid this, we propose the alternative \textit{displaced-null} measurement strategy in which the reference parameter is altered by a small amount which is sufficient to ensure parameter identifiability.

We use this strategy to devise asymptotically optimal measurements for models where the quantum Cram\'{e}r-Rao bound is achievable. More generally, we extend the method to arbitrary multi-parameter models and prove the asymptotic achievability of the the Holevo bound. An important tool in our analysis is the theory of quantum local asymptotic normality  which provides a clear intuition about the design of the proposed estimators, and shows that they have asymptotically normal distributions.
\end{abstract}

\maketitle

\section{Introduction and main results}
\label{sec.intro}












The estimation of unknown parameters from measurement data is the central task of quantum statistical inference \cite{Hayashi2005, Paris2008,TothReview,RafalReview,Tomo2,Albarelli2020,Sidhu_2020}. In recent decades, the area has witnessed an explosive growth covering a wealth 
of topics such as quantum state tomography \cite{Gross10,Cramer2010,Tomo1,Haah16,Lanyon2017,Guta_2017,Yang19,lahiry2021,Yuen23}, multi-parameter estimation \cite{Szczykulska16,Nichols2018,RafalReview,Albarelli19,Albarelli2020,Liu20}, sufficiency \cite{Mosonyi2003, Jencova2006}, local asymptotic normality \cite{LAN1,LAN2,LAN3,LAN4,LAN6,Yamagata13,Fujiwara20,Fujiwara22}, shadow tomography \cite{Aaronson2018, Huang2020}, Bayesian methods  \cite{Personick1971,Tsang20,Rubio2020,Rubio2021}, quantum metrology \cite{Metrology1,Fujiwara2008,Giovannetti2011, Metrology2,Girolami14,Smirne16, Seveso2017,Haase18,Metrology4,Rossi20,Sisi21}, error correction methods \cite{Gorecki2019, Zhou2018_2}, hamiltonian learning \cite{Yuan2015, Yuan2016},  thermometry \cite{Correa2015, Mehboudi2019},  gravitational waves detection \cite{GW4,GW5}, magnetometry \cite{Jones2009, Jan2021, Brask2015,ARPG17}, quantum sensing \cite{Degen2017, Marciniak2022, Zwick2023}, imaging \cite{Tsang16,Tsang21,Lupo20,Fiderer21,Oh21}, semi-parametric estimation \cite{Tsang19,Tsang2020} estimation of open systems \cite{GW01,Guta2011,Molmer14,Guta_2015,GutaCB15,Guta_2017,Ilias22,Fallani22}, waveform \cite{TWC11,Berry2015} and noise \cite{Ng16,Norris16,Shi23,Sung19,Tsang23} estimation.

A common feature of many quantum estimation problems is that `optimal' measurements depend on the unknown parameter, so they can only be implemented approximately, and the  optimality is at best achieved in the limit of large `sample size'. This raises the question of how to interpret theoretical results such as the quantum Cram\'{e}r-Rao bound (QCRB) \cite{Holevo2011, YuenLax76,Belavkin76, Helstrom1976,QCR1,Nagaoka} and how to design adaptive measurement strategies which attain the optimal statistical errors in the asymptotic limit. When multiple copies of the state are available, the standard strategy is to use a sub-sample to compute a rough estimator and then apply the optimal measurement corresponding to the estimated value. Indeed this works well for the case of the symmetric logarithmic derivative \cite{GillMassar}, an operator which saturates the quantum Cram\'{e}r-Rao bound for one-dimensional parameters. However, the QCRB fails to predict the correct \emph{attainable} error for quantum metrology models which consist of correlated states and exhibit Heisenberg (quadratic) scaling for the mean square error \cite{RDDWiseman}. This is due to the fact that in order to saturate the QCRB one needs to know the parameter to a precision comparable to what one ultimately hopes to achieve. 

In this paper we uncover a somewhat complementary phenomenon, where the usual adaptive strategy fails \emph{precisely} because it is applied to a `good' guess of the true parameter value. This happens in the standard multi-copy setting when estimating a pure state by means of `null measurements', where the experimenter aims to measure in a basis that contains the unknown state. While this can only be implemented approximately, the technique is known to exhibit certain Fisher-optimality properties \cite{NullQFI1,NullQFI2,NullQFI3} and has the intuitive appeal of `locking' onto the correct value as outcomes corresponding to other measurement vectors become more and more unlikely. 

In Theorem \ref{prop.null}, which is our first main result, we show that the standard adaptive strategy in which the parameter is first estimated on a sub-sample and then the null-measurement for this rough value is applied to the rest of the ensemble, fails to saturate the QCRB, and indeed does not attain the standard rate of precision.
Our result shows the importance of accompanying mathematical properties with clear operational procedures that allow us to draw  statistical conclusions; this provides another example of the limitations of the `local' estimation approach based on the quantum Cram\'{e}r-Rao bound \cite{Tsang_blog}. Indeed the reason behind the failure of the standard adaptive strategy is the fact that null-measurements suffer from non-identifiability issues when the true parameter and the rough preliminary estimator are too close to each other, i.e. when the latter is a reasonable estimator of the former.

Fortunately, it turns out that the issue can be resolved by deliberately shifting the measurement reference parameter  away from the estimated value by a vanishingly small but sufficiently large amount to resolve the non-identifiabilty issue. Using this insight we devise a novel adaptive measurement strategy which achieves the Holevo bound for arbitrary multi-parameter models, asymptotically with the sample size. This second main result is described in Theorem \ref{thm:dnmgeneral}. In particular our method can be used to achieve the quantum 
Cram\'{e}r-Rao bound for models where this is achievable, which was the original theme of \cite{NullQFI1,NullQFI2,NullQFI3}. The validity of the displaced-null strategy goes beyond the setting of the estimation with independent copies and has already been employed for optimal estimation of dynamical parameters of open quantum systems by counting measurements \cite{DayouCounting}. The extension of our present results to the setting of quantum Markov chains will be presented in a forthcoming publication \cite{GGG}. 
In the rest of this section we give a brief review of the main results of the paper.

\subsubsection*{The quantum Cram\'{e}r-Rao bound and the symmetric logarithmic derivative}

The quantum estimation problem is formulated as follows: given a quantum system prepared in a state $\rho_\theta$ which depends on an unknown (finite dimensional)
parameter $\theta\in \Theta$, one would like to estimate $\theta$ by performing a measurement $M$
and constructing an estimator $\hat{\theta} = \hat{\theta}(X)$ based on the (stochastic) outcome $X$. The Cram\'{e}r-Rao bound \cite{Lehmann1998, Vaart1998} shows that for a given measurement $M$, the covariance of any unbiased estimator 
is lower bounded as 
$
{\rm Cov}(\hat{\theta})\geq 
I^{-1}_M(\theta)
$
where $I_M(\theta)$ is the classical Fisher information (CFI) of the measurement 
outcome. 


Since the right side depends on the measurement, this prompts a fundamental and distinctive question in quantum statistics: what are the ultimate bounds on estimation accuracy and what measurement designs achieve these limits?
The cornerstone result in this area is that, irrespective of the measurement $M$, the CFI $I_M(\theta)$ is upper bounded by the quantum Fisher information $F(\theta)$, the latter being an intrinsic property of the quantum statistical model $\{\rho_\theta\}_{\theta\in \Theta}$. By combining the two bounds we obtain the celebrated quantum Cram\'{e}r-Rao bound (QCRB) \cite{Holevo2011, YuenLax76,Belavkin76, Helstrom1976,QCR1,Nagaoka}
$
{\rm Cov}(\hat{\theta})\geq F^{-1}(\theta).
$
For one dimensional parameters the QFI can be (formally) achieved by measuring an observable $\mathcal{L}_\theta$ called the symmetric logarithmic derivative (SLD), defined as the solution of the Lyapunov equation $\frac{d\rho_\theta}{d\theta}= \frac{1}{2} (\rho_\theta\mathcal{L}_\theta+ \mathcal{L}_\theta \rho_\theta)$. However, since the SLD depends on the unknown parameter $\theta$, this measurement cannot be performed without its prior knowledge, and the formal achievability is unclear without  further operational specifications.

Fortunately, this apparent circularity issue can be solved in the context of asymptotic estimation \cite{Hayashi2005}. In most practical applications one does not measure a single system but deals with (large) ensembles of identically prepared systems, or multi-partite correlated states as in quantum enhanced metrology \cite{Metrology4, Zhou2018} and continuous time estimation of Markov dynamics \cite{DayouCounting, Guta_2015, Guta_2017, Molmer}. Here one considers issues such as the scaling of errors with sample size, collective versus  separable measurements, and whether one needs fixed or adaptive measurements.  In particular, in the case of \emph{one-dimensional} models, the QCRB can be achieved \emph{asymptotically} with respect to the size $n$ of an ensemble of independent identically prepared systems, by using a two steps \emph{adaptive} measurement strategy \cite{GillMassar}. In the first step, a preliminary `rough' estimator $\tilde{\theta}_n$ is computed by measuring a sub-ensemble of $\tilde{n} =o(n)$ systems, after which  the SLD for parameter value $\tilde{\theta}_n$ (our best guess at the optimal observable $\mathcal{L}_\theta$) is measured on each of the remaining systems.
In the limit of large sample size $n$, the preliminary estimator $\tilde{\theta}_n $ approaches $\theta$ and the two step procedure achieves the QCRB in the sense that the mean square error (MSE) of the final estimator scales as $(nF(\theta))^{-1}$.

By implicitly invoking the above adaptive measurement argument, the quantum estimation literature has largely focused on computing or estimating the QFI of specific models, or designing input states which maximise the QFI in quantum metrology settings. However, as  shown in \cite{RDDWiseman}, the adaptive argument breaks down for models exhibiting quadratic (or Heisenberg) scaling of the QFI where the \emph{achievable} MSE is larger by a constant factor compared to the QCRB prediction, even asymptotically. In this work we show that similar care needs to be taken even when considering standard estimation problems involving ensembles of independent quantum systems and standard error scaling.
\subsubsection*{Null measurements and their standard adaptive implementation}

\begin{figure*}
\includegraphics[width=\linewidth]{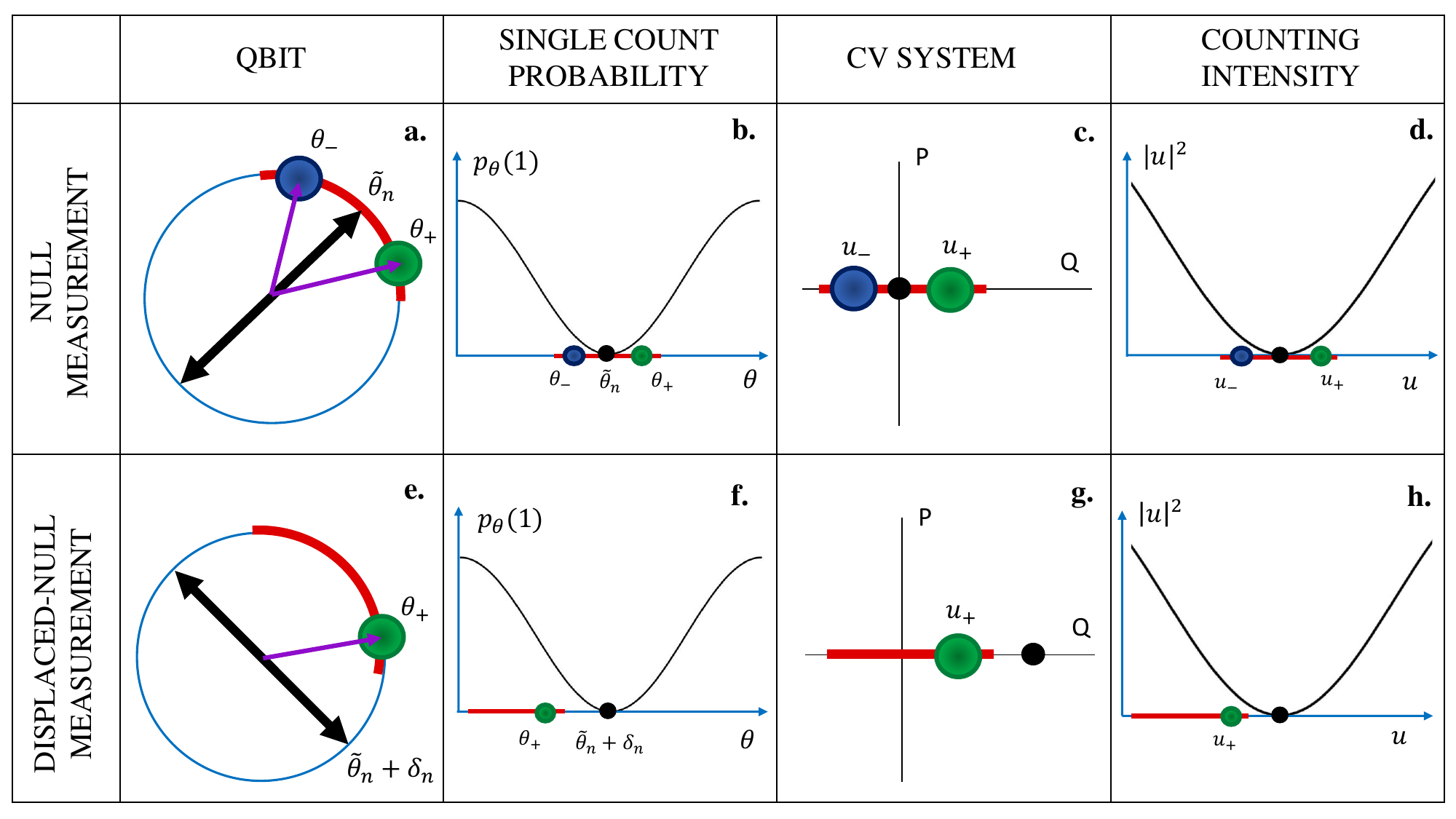}
\caption{The figure illustrates the non-identifiability problem occurring with null measurement (first row) and how it is fixed by displaced-null measurement (second row). In the first column the red arc on the xz Bloch sphere circle (in blue) represents the set of parameters after localisation (confidence interval), the green disk represents the true parameter value $\theta=\theta_+$ and the blue disk (panel a) is the parameter $\theta_-$ which is indistinguishable from the true one, in the null basis. The black arrow represents the chosen measurement basis. The second column displays a plot of the single count probability as a function of the parameter: in the null measurement case such a function is not injective on the set of parameters determined after the localisation (panel b). The third column shows the phase space of a Gaussian model consisting of coherent states with unknown displacement along the Q axis: the red interval is the parameter space, the black dot corresponds to the number operator measured, the green disk to the true coherent state and the blue disk (panel c) is the coherent state which is indistinguishable from the true one in the null measurement case. The last column plots the intensity of the number operator as a function of the coherent state amplitude.} \label{Fig:nonidentifiability}
\end{figure*}

Specifically, we revisit the problem of estimating a parameter of a \emph{pure state model} $\{|\psi_\theta\rangle\}_{\theta\in \Theta}$ and analyse a measurement strategy  \cite{NullQFI1,NullQFI2,NullQFI3}, which we broadly refer to as \emph{null measurement}.  The premise of the null measurement is the observation that if one measures $|\psi_\theta\rangle$ in an orthonormal basis $\mathcal{B}(\theta) :=\{|v_1\rangle , \dots , |v_d\rangle\}$ such that $|v_1\rangle = |\psi_\theta\rangle$  then the only possible outcome is $X=1$ and all other outcomes have probability zero. Since $\theta$ is unknown, in practice one would measure in a basis $\mathcal{B}(\tilde{\theta})$ corresponding to an approximate value $\tilde{\theta}$ of the true parameter $\theta$, and exploit the occurrence of low probability outcomes $X\neq 1$ in order to estimate the deviation of $\theta$ from $\tilde{\theta}$.
This intuition is supported by the following property which is a specialisation to one-dimensional parameters of a more general result derived in \cite{NullQFI1,NullQFI2,NullQFI3}: as $\tilde{\theta}$ approaches $\theta$, the classical Fisher information $I_{\tilde{\theta}}(\theta)$ associated with $\mathcal{B}(\tilde{\theta})$ converges to the QFI $F(\theta)$.
This implies that null measurements can achieve MSE rates scaling  as $n^{-1}$ with constants that are \emph{arbitrarily close} to $F^{-1}(\theta)$, by simply measuring all $n$ systems of an ensemble in a basis $\mathcal{B}(\tilde{\theta})$ with a \emph{fixed} $\tilde{\theta}$ that is close to $\theta$:
$$
n\mathbb{E}_\theta[(\hat{\theta}_n-\theta)^2] \to 
I^{-1}_{\tilde{\theta}}(\theta)\approx F^{-1}(\theta).
$$
Do null measurements actually achieve the QCRB (asymptotically) or just 'come close' to it? In absence of a detailed multi-copy  operational interpretation in \cite{NullQFI1,NullQFI2,NullQFI3}, the most natural strategy is to apply the same two step adaptive procedure which worked well in the case of the SLD measurement. 
A preliminary estimator $\tilde{\theta}_n$ is first computed by measuring $\tilde{n}$ systems and the rest of the ensemble is subsequently measured in the basis $\mathcal{B}(\tilde{\theta}_n)$. Since $I_{\tilde{\theta}_n}(\theta)$ converges to $ F(\theta)$ as $\tilde{\theta}_n$ approaches $\theta$, it would appear that the QCRB is achieved asymptotically. One of our main results is to show that this adaptive procedure actually \emph{fails} to achieve the QCRB even in the simple qubit model 
\begin{equation}\label{eq:simple.qubit.model}
|\psi_\theta\rangle= \cos\theta|0\rangle+ \sin \theta|1\rangle,
\end{equation}
thus providing another example where caution is needed when using arguments based on Fisher information, see \cite{Tsang_blog} for other examples.

More precisely, we show that if the preliminary estimator $\tilde{\theta}_n$ is reasonably good (cf. section \ref{sec.null} for precise formulation), any final estimator $\hat{\theta}_n$ computed from the outcomes of the null measurement $\mathcal{B}(\tilde{\theta}_n)$ is not only suboptimal but does not even achieve the standard $n^{-1}$ estimation MSE rate. The reason for the radically different behaviors of the SLD and null meaurement settings is that the latter suffers from a \emph{non-identifiability} problem when the parameter $\tilde{\theta}$ (which determines the null basis) is close to $\theta$. Indeed, 
since at $\tilde{\theta}=\theta$ the null measurement has a deterministic outcome, 
for $\tilde{\theta}\approx\theta$ the outcome probabilities are quadratic in 
$\epsilon= \theta-\tilde{\theta}$ and therefore, the parameters 
$\theta_\pm = \tilde{\theta}\pm \epsilon$ cannot be distinguished (at least in second order). If $\tilde{\theta}_n$ is a reasonably good estimator, then $\epsilon_n= |\theta-\tilde{\theta}_n|$ is of the order $\tilde{n}^{-1/2}$, so the error in estimating $\theta$ is at least of the order of the distance $|\theta_+ -\theta_-|$ between the two undistinguishable candidate parameters $\theta_{\pm}= \tilde{\theta}_n \pm \epsilon_n$, which scales as  $\tilde{n}^{-1/2}$ instead of $n^{-1/2}$. Since $\tilde{n}=o(n)$ the mean square error decreases slower that the standard rate $n^{-1}$. This argument is illustrated in Figure \ref{Fig:nonidentifiability}a. for the simple case of the qubit rotation model \eqref{eq:simple.qubit.model} which is discussed in detail in section \ref{sec.null}.
\subsubsection*{Asymptotic optimality of displaced-null measurements} 
Fortunately, the above explanation offers an intuitive solution to the non-identifiability problem. Assuming that the preliminary estimator $\tilde{\theta}_n$ satisfies standard concentration properties (e.g. asymptotic normality), one finds that $\theta$ belongs (with high probability) to a confidence interval $I_n$ centered at $\tilde{\theta}_n$, whose length is slightly larger than the estimation uncertainty $\tilde{n}^{-1/2}$. Therefore by displacing $\tilde{\theta}_n$ by a (vanishingly small) amount $\delta_n>0$ that is larger than this uncertainty, we can make sure that $I_n$ lies at the left side of $\theta^{\prime}_n := \tilde{\theta}_n+\delta_n$ and therefore measuring in the basis $\mathcal{B}(\theta^{\prime}_n)$ circumvents the non-identifiability issue. This is illustrated in panels e. and f. of Figure \ref{Fig:nonidentifiability}.

The main aim of the paper is to investigate this method which we call a \emph{displaced-null} measurement strategy and derive asymptotic optimality results for the resulting estimators. In section \ref{sec.one.parameter} we show that the displaced-null measurement achieves the QCRB in the one-parameter qubit model for which the standard adaptive procedure failed; the corresponding second stage estimator is a simple average of measurement outcomes and satisfies asymptotic normality, thus allowing practitioners to define asymptotic confidence intervals.

In section 
\ref{sec:qdits} we extend the null-measurement strategy to \emph{multi-parameter} models of pure qudit states. In this case, the QCRB is typically not attainable even asymptotically due to the incompatibility of optimal measurements corresponding to different parameter components. However, we show that the Holevo bound \cite{Holevo2011} \emph{can} be achieved asymptotically. We first consider the task of estimating a completely unknown pure state with respect to the the Bures (fidelity) distance. In this case we show that the Holevo bound can be achieved 
by using two separate displaced-null measurements, for the real and imaginary parts of the state coefficients with respect to a basis containing $|\psi_{\theta^\prime_n}\rangle$ as a vector. 
The second task is to estimate a general $m$-dimensional model with respect to an arbitrary locally quadratic distance on the parameter space. Here we show that the Holevo bound is achievable by applying displaced-null measurements on copies of the systems coupled with an ancilla in a fixed state. The proof relies on the intuition gained from quantum local asymptotic normality theory and its use in establishing the achievability of the Holevo bound \cite{LAN4,RafalReview} by mapping the ensemble onto a continuous variables system. However, unlike the latter, the displaced-null technique only involves separate projective measurements on system-ancilla pairs.

Finally, in section \ref{sec:achievingQCRB-displacednull} we show that for multiparameter models where the QCRB is achievable, this can be done using displaced-null measurements. This puts related results of \cite{NullQFI1,NullQFI2,NullQFI3} on a firm operational basis.
\subsubsection*{Local asymptotic normality perspective} 
The theory of quantum local asymptotic normality (QLAN) \cite{LAN1, LAN2, LAN3, LAN4} offers an alternative perspective on the displaced-null measurements strategy outlined above.
In broad terms, QLAN is a statistical tool that allows us to approximate the i.i.d. model describing the joint state of an ensemble of systems, by a single continuous variables Gaussian state whose mean encodes information about the unknown parameter (cf. sections \ref{sec:LAN} and \ref{sec:LAN-qudits} for more details). By applying this approximation, the null measurement problem discussed earlier can be cast into a Gaussian version formulated as follows. Suppose we are given a one-mode continuous variables system prepared in a coherent state $|u\rangle$ with unknown displacement $u\in \mathbb{R}$ along the $Q$ axis, and assume that $|u|\leq a_n$ for some bound $a_n$ which diverges with $n$. At $u=0$, the system state is the vacuum, and the measurement of the number operator $N$ is a null measurement (see Figure \ref{Fig:nonidentifiability}c.). However, for a given $u\neq 0$ the number operator has  Poisson distribution with intensity $|u|^2$, and therefore cannot distinguish between parameters $u_{\pm}:=\pm u$, cf. Figure \ref{Fig:nonidentifiability}d. This means that any estimator will have large MSEs of order $a_n^2$ for large values of $u$. In contrast, measuring the quadrature $Q$ produces (optimal) estimators with fixed MSE given by the vacuum fluctuations. 
However, the non-identifiability problem of the counting measurement can be lifted by displacing the coherent state along the $Q$ axis by an amount $\Delta_n>a_n$  and then measuring $N$. Equivalently, one can measure the corresponding displaced number operator on the original coherent state as illustrated in panels g. and h. of Figure \ref{Fig:nonidentifiability}. In this case the intensity $(u-\Delta_n)^2$ is in one-to-one correspondence with $u$ so the parameter is identifiable. Moreover, for large $n$, the counting measurement can be linearised and becomes equivalent to measuring the quadrature $Q$, a well known fact from homodyne detection \cite{Leonhardt}. 

QLAN shows that the Gaussian problem discussed above is the asymptotic version of the one-parameter qubit rotation model \eqref{eq:simple.qubit.model}
which we used earlier to illustrate the concept of approximate and displaced null measurements. The coherent state $|u\rangle$ corresponds to all qubits in the state $|\psi_{u/\sqrt{n}}\rangle$ (assuming for simplicity that $\tilde{\theta}_n =0$ and writing $\theta=u/\sqrt{n}$). The number operator corresponds to measuring in the standard basis, which is an exact null measurement at $u=0$. On the other hand, the displaced number operator corresponds to measuring in the rotated basis with angle $\delta_n = n^{-1/2}\Delta_n$.

The same Gaussian correspondence is used in section \ref{sec:qdits} for more general problems involving multiparameter estimation for pure qudit state models and establishing the achievability of the Holevo bound, cf. Theorem \ref{thm:dnmgeneral}. The general strategy is to translate the i.i.d. problem into a Gaussian one, solve the latter by using displaced number operators in a specific mode decomposition and then translate this into qudit measurement with respect to specific rotated bases.

This paper is organised as follows. Section \ref{sec:QCRB} reviews the QCRB and the conditions for its achievability. 
In section \ref{sec.null} we show that null measurements based at reasonable preliminary estimators fail to achieve the standard error scaling. In section \ref{sec.one.parameter} we introduce the idea of displaced-null measurement and prove its optimality in the paradigmatic case of a one-parameter qubit model. 
In section \ref{sec:qdits} we treat the general case of $d$ dimensional systems and show how the Holevo bound is achieved on general models, and deal with the case where the multi-parameter QCRB is achievable.



\section{Achievability of the quantum Cram\'{e}r-Rao bound for pure states}
\label{sec:QCRB}

In this section we review the quantum Cram\'{e}r-Rao bound (QCRB) and the conditions for its achievability in the case of models with \emph{one-dimensional} parameters, which will be relevant for the first part of the paper. 

The estimation of multidimensional models and the corresponding Holevo bound is discussed in section \ref{sec:qdits}.


Consider a quantum statistical model given by a family of d-dimensional density matrices $\rho_\theta $ which depend smoothly on an unknown  parameter 
$\theta\in \mathbb{R}$.
Let $\mathcal{M}$ be a measurement on $\mathbb{C}^d$ with positive operator valued measure (POVM) elements 
$\{M_0,\dots, M_p\}$. By measuring $\rho_\theta$ we obtain an outcome $X\in \{0,\dots , p\}$ with probabilities
$$
p_\theta(X=i) = p_\theta(i) = {\rm Tr}(M_i \rho_\theta),\qquad i = 0,\dots, p.
$$
The classical Cram\'{e}r-Rao bound states that the variance of any \emph{unbiased} estimator $\hat{\theta} = \hat{\theta} (X)$ of $\theta$ is lower bounded as
\begin{equation}\label{eq:CCR}
{\rm Var} (\hat{\theta}) :=\mathbb{E}_\theta[(\hat{\theta}- \theta)^2] \geq I_{\mathcal{M}}(\theta)^{-1}    
\end{equation}
where $I_{\mathcal{M}}(\theta)$ is the classical Fisher information (CFI)
\begin{equation}\label{eq.CFI}
I_{\mathcal{M}}(\theta) = \mathbb{E}_\theta\left [\left(\frac{d\log p_\theta }{d\theta}\right)^2\right ] =\sum_{i:p_\theta(i)>0}
p_\theta^{-1}(i) \left
(\frac{dp_\theta(i)}{d\theta}\right)^2.
\end{equation}
The CFI associated to any measurement is upper bounded by the quantum Fisher information (QFI) \cite{QCR1, QCR2}
\begin{equation}\label{eq:FIQFI}
I_{\mathcal{M}}(\theta)\leq F(\theta) 
\end{equation}
where $F(\theta) = {\rm Tr}(\rho_\theta \mathcal{L}^2_\theta)$ and $\mathcal{L}_\theta$ is the symmetric logarithmic derivatives (SLD) defined by the Lyapunov equation  
$$
\frac{d{\rho}_\theta}{d\theta} = \frac{1}{2}(\mathcal{L}_\theta \rho_\theta + \rho_\theta\mathcal{L}_\theta).
$$
By putting together \eqref{eq:CCR} and \eqref{eq:FIQFI} we obtain the quantum Cram\'{e}r-Rao bound (QCRB) \cite{Holevo2011,Helstrom1976}
\begin{equation}
\label{eq:QCRB}
{\rm Var} (\hat{\theta}) :=\mathbb{E}_\theta[(\hat{\theta}- \theta)^2] \geq F(\theta)^{-1}.  
\end{equation}
which sets a fundamental limit to the estimation precision. A similar bound on the covariance matrix of an unbiased estimator holds for  multidimensional models, cf. section \ref{sec:qdits}. 

An important question is which measurements saturate the bound \eqref{eq:FIQFI}, and what is the statistical interpretation of the corresponding QCRB \eqref{eq:QCRB}. For completeness, we state the exact conditions in the following Proposition whose formulation is  adapted from \cite{Zhou2018}. The  proof is included in appendix \ref{sec:QCRBsat}.

\begin{proposition}
\label{th:QCRB-acievability}
Let $\rho_\theta$ be a one-dimensional quantum statistical model and let 
$\mathcal{M}:=\{M_0,\dots,  M_p\}$ be a measurement with probabilities $p_\theta (i) := {\rm Tr}(\rho_\theta M_i)$. Then $\mathcal{M}$ achieves the bound \eqref{eq:FIQFI} if and only if the following conditions hold: 

1) if $p_\theta(i)>0$ there exists $\lambda_i \in \mathbb{R}$ such that
\begin{equation}\label{eq:cond1}
M^{1/2}_i\rho^{1/2}_\theta = \lambda_i M_i^{1/2}\mathcal{L}_\theta
\rho^{1/2}_\theta
\end{equation}

2) if $p_\theta(i)=0$ for some $i$ then ${\rm Tr} (M_i \mathcal{L}_\theta\rho_\theta \mathcal{L}_\theta)=0$.
\end{proposition}

One can check that the conditions in Proposition \ref{th:QCRB-acievability} are satisfied, and hence the bound \eqref{eq:FIQFI} is saturated, if $\mathcal{M}$ is the measurement of the observable $\mathcal{L}_\theta$. However, in general this observable depends on the unknown parameter, so achieving the QFI does not have an immediate statistical  interpretation. Nevertheless, one can provide a meaningful operational interpretation in the scenario in which a large number $n$ of copies of $\rho_\theta$ is available. In this case one can apply the adaptive scheme presented in the introduction: using a (small) sub-sample to obtain a `rough' preliminary estimator $\tilde{\theta}$ of $\theta$ and then measuring $\mathcal{L}_{\tilde{\theta}}$ on the remaining copies. This adaptive procedure provides estimators $\hat{\theta}_n$ which achieve the Cram\'{e}r-Rao bound asymptotically in the sense that (see e.g. \cite{GillMassar,LAN4})
$$
n \mathbb{E}_\theta[(\hat{\theta}_n- \theta)^2] \to F^{-1}(\theta).
$$ 




\emph{Pure state models.} While for full rank states $(\rho_\theta>0)$ the second condition in Proposition \ref{th:QCRB-acievability} is irrelevant, 
this is not the case for rank deficient states, and in particular for pure state models.

Indeed let us assume that the model consists of pure states $\rho_\theta= |\psi_\theta\rangle \langle \psi_\theta|$
and let us choose the phase dependence of the vector state such that $\langle \dot{\psi}_\theta|\psi_\theta\rangle=0$ (alternatively, one can use 
$|\psi^\perp_\theta\rangle  := |\dot{\psi}_\theta \rangle-  \langle \psi_\theta | \dot{\psi}_\theta\rangle |\psi_\theta \rangle$ instead of $|\dot{\psi}_\theta \rangle$ in the equations below). Then 
$$
\mathcal{L}_\theta = 2(|\dot{\psi}_\theta\rangle\langle \psi_\theta| + |\psi_\theta\rangle\langle \dot{\psi}_\theta|),
\quad {\rm and} \quad 
F(\theta)  = 
4\| \dot\psi_{\theta} \|^2.
$$ 
Let $\mathcal{M}$ to be a projective measurement with $M_i= |v_{i}\rangle \langle v_{i}|$ where $\mathcal{B}:=\{|v_0\rangle, \dots , |v_{d-1}\rangle\} $ is an orthonormal basis (ONB). Without loss of generality we can choose the phase factors such that $\langle v_i|\psi_\theta\rangle \in \mathbb{R}$ at the particular value of interest $\theta$. Equation 
\eqref{eq:cond1} in Proposition \ref{th:QCRB-acievability} becomes 
$
\langle v_i|\dot{\psi}_\theta\rangle \in \mathbb{R},
$
i.e. in the first order, the statistical model is in the real span of the basis vectors. Condition 2 requires that if $\langle v_i|\psi_\theta\rangle =0$ then $\langle v_i|\dot{\psi}_\theta\rangle =0$. Intuitively, this implies that, in the first order, the model is restricted to the real subspace spanned by the basis vectors with positive probabilities. For example if
\begin{equation} \label{eq:example}
    |\psi_\theta\rangle: =\cos\theta |0\rangle + \sin\theta |1\rangle\in \mathbb{C}^2,
\end{equation}
then any measurement with respect to an ONB consisting of superpositions of $|0\rangle$ and $|1\rangle$ with \emph{nonzero} \emph{real} coefficients achieves the QCRB at $\theta=0$, and no other measurement does so. This model will be discussed in detail in sections \ref{sec.null}  and \ref{sec:qubits}.


\emph{Null measurements.} We now formally introduce the concept of  a \emph{null measurement} which will be the focus of our investigation. The general idea is to choose a measurement basis such that one of its vectors is equal or close to the unknown state. In this case, the corresponding outcome has probability close to one while the occurrence of other outcomes can serve as a `signal' about the deviation from the true state. Let us consider first an \emph{exact null measurement}, i.e. one in which the measurement basis $\mathcal{B} = \mathcal{B}(\theta)$ is chosen such that $|v_0\rangle = |\psi_\theta\rangle$, e.g in the example in equation \eqref{eq:example} the null measurement at $\theta=0$ is determined by the standard basis. Such a measurement does not satisfy the conditions for achieving the QCRB. Indeed, we have $p_\theta(i) =\delta_{0,i}$ and condition 2 implies $\langle v_i|\dot{\psi}_\theta\rangle=0$ for all $i=1,\dots ,d-1$. However this is impossible given that $|v_0\rangle = |\psi_\theta\rangle$ and $\langle \dot{\psi}_\theta |\psi_\theta\rangle =0$. 
In fact, the exact null measurement has zero CFI, which implies that there exists no (locally) unbiased estimator. Indeed, since probabilities belong to $[0,1]$, and $p_\theta (i)$ is either $0$ or $1$ for a null measurement, all first derivatives at $\theta$ are zero 
so the CFI \eqref{eq.CFI} is equal to zero, i.e. $I_{\mathcal{B}(\theta)}(\theta) = 0$. 

One can rightly argue that the exact null measurement as defined above is not an operationally useful concept and cannot be implemented  experimentally as it requires the exact knowledge of the unknown parameter. However, in a multi-copy setting the measurement  \emph{can} incorporate information about the parameter, as this can be obtained by measuring a sub-ensemble of systems in a preliminary estimation step, similarly to the SLD case.
It is therefore meaningful to consider \emph{approximate null} measurements, which satisfy the null property at $\tilde{\theta}\approx\theta$, i.e. we measure in a basis $\mathcal{B}(\tilde{\theta}) = \{ |v_0^{\tilde{\theta}}\rangle, \dots , |v_{d-1}^{\tilde{\theta}}\rangle \}$ with $|v_0^{\tilde{\theta}}\rangle = |\psi_{\tilde{\theta}}\rangle.$
Interestingly, while the exact null measurement has zero CFI, 
an approximate null measurement $\mathcal{B}(\tilde{\theta})$ `almost achieve' the QCRB in the sense that the corresponding classical Fisher information $I_{\mathcal{B}(\tilde{\theta})}(\theta)$ converges to $F(\theta)$ as $\tilde{\theta}$ approaches $\theta$ \cite{NullQFI1,NullQFI2,NullQFI3}. This means that by using an approximate null measurement we can achieve asymptotic error rates arbitrarily close (but not equal) to the QCRB, by measuring in a basis $\mathcal{B}(\tilde{\theta})$ with a fixed $\tilde{\theta}$ close to $\theta$. 


The question is then, is it possible to achieve the QCRB asymptotically with respect to the sample size by employing null measurements determined by an \emph{estimated} parameter value, as in the case of the SLD measurement? References \cite{NullQFI1,NullQFI2,NullQFI3} do not address this question, aside from the above Fisher information convergence argument. 

To answer the question we allow for measurements which have the null property at parameter values determined by \emph{reasonable} preliminary estimators based on measuring a sub-sample of a large ensemble of identically prepared systems (cf. section \ref{sec.null} for precise definition). We investigate such measurement strategies 
and show that the natural two step implementation -- use the rough estimator as a vector in the second step measurement basis -- \emph{fails} to achieve the standard rate $n^{-1/2}$ on simple qubit models. We will see that this is closely related to the fact that the CFI of the exact null measurement is zero, unlike the SLD case.

Nevertheless, in section 
\ref{sec.one.parameter} we show that a modified strategy which we call a \emph{displaced-null measurement} does achieve asymptotic optimality in the simple qubit model discussed above.  This scheme is then extended to general multidimensional qudit models in section \ref{sec:qdits} and shown to achieve the Holevo bound for general multiparameter models.

\section{Why the naive implementation of a null measurement does not work}
\label{sec.null}
In this section we analyse the \emph{null measurement} scheme described in section \ref{sec:QCRB}, for the case of a simple one-parameter qubit rotation model. The main result is Theorem \ref{prop.null} which shows that the naive/natural implementation of the null-fails to achieve the QCRB.

Let
\begin{equation}\label{eq:qubit.model}
|\psi_\theta\rangle  = e^{-i \theta\sigma_y}|0\rangle= 
\cos(\theta)|0\rangle  + \sin(\theta)|1\rangle 
\end{equation}
be a one-parameter family of pure states which describes a circle in the $xz$ plane of the Bloch sphere. To simplify some of the arguments below we will assume that $\theta$ is known to be  in the open interval $\Theta = (-\pi/8, \pi/8)$, but the analysis can be extended to completely unknown $\theta$. The quantum Fisher information is 
$$
F(\theta) = 4{\rm Var}(\sigma_y) =  
4 \langle \psi_\theta |\sigma_y^2 |\psi_\theta\rangle - 4\langle \psi_\theta |\sigma_y |\psi_\theta\rangle^2 =4.
$$ 
We now consider the specific value $\theta=0$, so $|\psi_0\rangle =|0\rangle$ and $|\dot{\psi}_0\rangle =|1\rangle$. According to Proposition \ref{th:QCRB-acievability} 
any measurement with respect to a basis consisting of real combinations of $|0\rangle$ and $|1\rangle$ achieves the QCRB, with the exception of the basis $\{|0\rangle, |1\rangle\}$ itself. Indeed, let
\begin{equation}\label{eq:v0v1}
|v_0^\tau\rangle = \exp(-i\tau\sigma_y) |0\rangle ,
\quad
|v_1^\tau\rangle= \exp(-i\tau\sigma_y) |1\rangle
\end{equation} 
be such a basis ($\tau\neq 0$) , then the probability distribution is 
$$
p_\theta(0) = \cos^2(\theta-\tau), \qquad
p_\theta(1) = 
\sin^2(\theta-\tau)
$$
and the classical Fisher information
is
$$
I_\tau(\theta=0) = \mathbb{E}_{\theta=0}\left [\left(\frac{d\log p_\theta }{d\theta}\right)^2\right ] = 4.
$$
However, at $\tau=0$ we have 
$I_0(\theta=0)=0$ in agreement with the general fact that exact null measurements have zero CFI. This reveals a curious singularity in the space of optimal measurements, and our goal is to understand to what extent this is mathematical artefact or it has a deeper statistical significance.  

To start, we note that the 
failure of the standard basis measurement can also be understood as a consequence of parameter \emph{non-identifiability} around the parameter value $0$. Indeed, for $\tau=0$ we have $p_\theta(i)= p_{-\theta}(i)$ so this measurement cannot distinguish  $\theta$ from $-\theta$. A similar issue exists for $\tau\neq 0$, if $\theta$ is assumed to be completely unknown, or in an interval containing $\tau$, cf. Figure \ref{Fig:nonidentifiability}. On the other hand, if $\theta$ is \emph{known} to belong to an interval $I$ and $\tau$ is outside this interval, then the parameter \emph{is} identifiable and the standard asymptotic theory applies. For instance, measuring $\sigma_x$ leads to an identifiable statistical model for our quantum qubit model.

Consider now the following two step procedure, which arguably is the most natural way of implementing approximate-null measurements. A sub-ensemble of $\tilde{n}$ systems is used to compute a preliminary estimator $\tilde{\theta}_n$, and subsequently the remaining samples are measured in the null-basis at angle $\tau=\tilde{\theta}_n$. 
For concreteness we assume that $\tilde{n} = n^{1-\epsilon}$ for some small constant $\epsilon>0$, but our results hold more generally for $\tilde{n} =o(n)$ and 
$\tilde{n}\to\infty$ with $n$.

To formulate our theoretical result, we use the language of Bayesian statistics which we temporarily adopt for this purpose. We consider that the unknown parameter $\theta$ is random and is drawn from the uniform \emph{prior distribution}  $\pi(d\theta)= \frac{4}{\pi} d\theta$ over the parameter space $\Theta$. Adopting a Bayesian notation we let 
$p(d\tilde{\theta}_n|\theta):= p_\theta(d\tilde{\theta}_n)$ be the distribution of $\tilde{\theta}_n$ given $\theta$. The joint distribution of $(\theta,\tilde{\theta}_n)$ is then 
$$
p(d\theta, d\tilde{\theta}_n)=\pi(d\theta)p(d\tilde{\theta}_n|\theta) = 
p(d\tilde{\theta}) \pi(d\theta|\tilde{\theta}_n)
$$
where $\pi(d\theta|\tilde{\theta}_n)$ is the \emph{posterior distribution} of $\theta$ given $\tilde{\theta}_n$. 

{\bf Reasonable estimator hypothesis:} we assume that $\tilde{\theta}_n$ is a \emph{reasonable estimator} in the sense that the following conditions are satisfied for every $n \geq 1$:
\begin{enumerate}
\item 
$\pi(d\theta|\tilde{\theta}_n)$ has a density $\pi(\theta|\tilde{\theta}_n)$ with respect to the Lebesgue measure;
\item \label{item.2}
For each $n$ there exist a set 
$A_n\subseteq \Theta$ such that $\mathbb{P}(\tilde{\theta}_n \in A_n)>c$ for some constant $c>0$, and the following condition holds: 
for each $\tilde{\theta}_n\in A_n$, the positive symmetric function
$$
g_{n,\tilde{\theta}_n}(r) := 
{\rm min}
\{ \pi(\tilde{\theta}_n+r|\tilde{\theta}_n), \pi(\tilde{\theta}_n-r|\tilde{\theta}_n)
\}
$$
satisfies 
\begin{equation} \label{eq:constantmass}
\int_{r \geq \tau_n}\,g_{n,\tilde{\theta}_n}(r)dr \geq C
\end{equation}
where $ \tau_n := n^{-1/2+\epsilon/4}$ and  $C>0$ is a constant independent on $n$ and $\tilde{\theta}_n$.
\end{enumerate}
Condition \ref{item.2}. means that the posterior distribution has significant mass on \emph{both} sides of the preliminary estimator $\tilde{\theta}_n$, at a distance which is larger than $n^{-1/2+\epsilon/4}$, as illustrated in Figure \ref{fig:posterior}. Since standard  estimators such as maximum likelihood have asymptotically normal posterior distribution with standard deviation 
$\tilde{n}^{-1/2} = n^{-1/2+\epsilon/2}\gg n^{-(1-\epsilon/2)/2}$,  condition \ref{item.2}. is expected to hold quite generally, hence the name reasonable estimator. The following lemma shows that the natural estimator in our model is indeed reasonable.
\begin{lemma} \label{lem:reas}
Consider the measurement of $\sigma_x$ on a sub-ensemble of $\tilde{n}=n^{1-\epsilon}$ systems, and let $\tilde{\theta}_n$ be the maximum likelihood estimator. Then $\tilde{\theta}_n$ is a reasonable estimator.
\end{lemma}

The proof of Lemma \ref{lem:reas} can be found in Appendix \ref{sec:proof.lemma.reasonable.estimator}. The method can be extended to a wide class of estimators, since it essentially relies on assumptions which are quite standard in usual statistical problems.

The next Theorem is the main result of this section and shows that if a reasonable (preliminary) estimator is used as reference for a null measurement on the remaining samples, the MSE of the final estimator cannot achieve the QCRB, indeed it cannot even achieve standard scaling.  
\begin{center}
\begin{figure}
    \centering
\includegraphics[width=7cm]{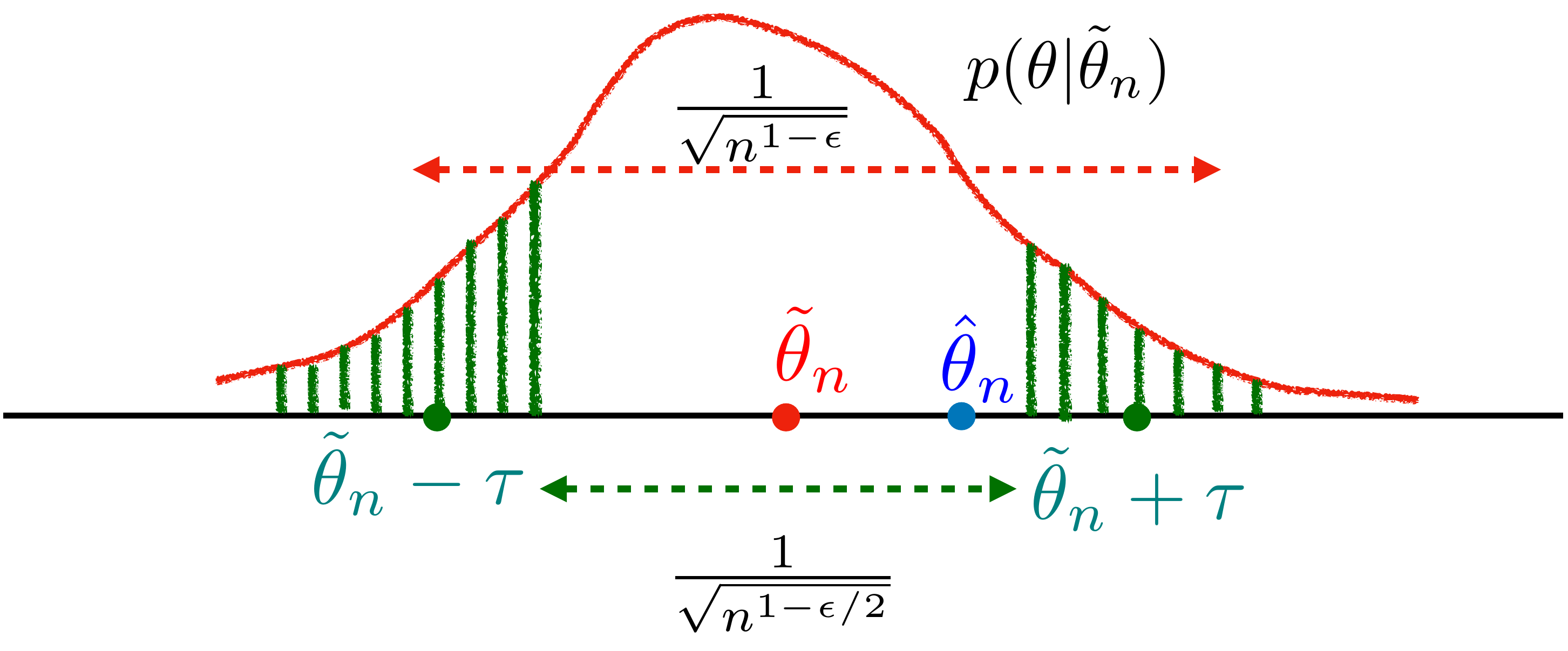}
    \caption{For a reasonable estimator $\tilde{\theta}_n$, the posterior distribution of $\theta$ is centred around $\tilde{\theta}_n$, and has width of order $n^{-(1-\epsilon)/2}$. 
    The assumption amounts to the fact that the posterior has non-vanishing mass on either side of $\tilde{\theta}_n$ at distance larger than $n^{-(1-\epsilon/2)/2}$ which is much smaller that the typical standard deviation.}
    \label{fig:posterior}
\end{figure}
\end{center}
\begin{theorem}
\label{prop.null}
Assume that $\tilde{\theta}_n$ is a reasonable estimator as defined above, obtained by measuring a sub-ensemble of size $\tilde{n}:= n^{1-\epsilon}$. Let $\hat{\theta}_n$ be an estimator of $\theta$ based on measuring the remaining $n-n^{1-\epsilon}$ sub-ensemble in the basis corresponding to angle $\tilde{\theta}_n$. Then   
$$
\lim_{n\to\infty} 
n R_{\pi} (\hat{\theta}_n) = \infty
$$
where 
$$
R_\pi (\hat{\theta}_n )=
\int_\Theta \pi(d\theta) \mathbb{E}_\theta [(\hat{\theta}_n -\theta)^2]
$$
is the average mean square error risk.
\end{theorem}

The proof of Theorem \ref{prop.null} can be found in Appendix \ref{app:proof.theorem.null.measurement}.

The fact that a reasonable estimator has a `balanced' posterior was key in obtaining the negative result in Theorem \ref{prop.null}. This encodes the fact that the null measurement cannot distinguish between possible parameter values $\theta = \tilde{\theta}_n+ \tau_n$ and $\theta = \tilde{\theta}_n- \tau_n$ leading to errors that are larger than $n^{-1/2}$. In section \ref{sec:qubits} we show how we can go around this problem by deliberately choosing the reference parameter of the null measurement to be displaced away from a reasonable estimator 
$\tilde{\theta}_n$ by an amount  $\delta_n$ that is large enough to insure identifiability, but small enough to still be in a shrinking neighbourhood of $\theta$.

In the proof of Theorem \ref{prop.null} we made use of the fact that, for the statistical model defined in equation \eqref{eq:qubit.model}, the law of the measurement in the basis containing $\ket{\psi_{\tilde{\theta}_n}}$ could not distinguish between $\tilde{\theta}_n \pm r$. Although for general pure state models this might not be the case, in the appendix  \ref{sec:local} we show that under some mild additional assumptions, the result of Theorem \ref{prop.null} extends to weaker notions of non-identifiability.

\section{Displaced-null estimation scheme for optimal estimation of pure qubit states}
\label{sec:qubits}
In section \ref{sec.null} we showed that a null measurement that uses a reasonable preliminary estimator as reference parameter is 
sub-optimal. We will now show that one can achieve the asymptotic version of the QCRB \eqref{eq:QCRB} by employing a null measurement at a reference parameter that is \emph{deliberately shifted} away from the reasonable estimator by a certain amount. We will call these \emph{displaced-null} measurements.
\subsection{The displaced-null measurement for one parameter qubit models}
\label{sec.one.parameter}

We consider the one parameter model $|\psi_\theta\rangle$ defined in equation \eqref{eq:qubit.model} and assume that we are given $n$ identical copies of $|\psi_\theta\rangle$. We apply the usual two step adaptive procedure: in the first step we use a vanishingly small proportion of the samples containing $\tilde{n}= n^{1-\epsilon}$ copies (where $\epsilon>0$ is a small parameter) to perform a preliminary (non-optimal) estimation producing a reasonable estimator $\tilde{\theta}_n$. For concreteness we assume that $\tilde{\theta}_n$ is the estimator described in Lemma \ref{lem:reas}. Using Hoeffding's bound we find that $\tilde{\theta}_n$ satisfies the concentration bound 
\begin{equation}\label{eq:concentration}
\mathbb{P}_\theta(| \tilde{\theta}_n-\theta |>n^{-1/2 +\epsilon})
\leq C e^{-n^{\epsilon} r}
\end{equation}
for some constants $C,r>0$.
This means that with high probability, $\theta$ belongs to the confidence interval $I_n = (\tilde{\theta}_n-n^{-1/2 +\epsilon}, ~ \tilde{\theta}_n+n^{-1/2 +\epsilon})$ whose size shrinks at a slightly slower rate than $n^{-1/2}$.

In the second step we would like to measure all remaining qubits in a basis which contains a vector that is close to $|\psi_\theta\rangle$. However, as argued in section \ref{sec.null}, the null measurement basis  $\{|v^{\tilde{\theta}_n}_0\rangle,|v^{\tilde{\theta}_n}_1\rangle \}$ satisfying  $|v^{\tilde{\theta}_n}_0\rangle = |\psi_{\tilde{\theta}_n}\rangle$ is suboptimal. More generally, for any angle $\tau\in I_n$, the basis defined by equation \eqref{eq:v0v1} suffers an identifiability problem as illustrated in panels a. and b. of Figure \ref{Fig:nonidentifiability}. For this reason, in the second step we choose the reference value
$$
\theta^\prime_n:= \tilde{\theta}_n + \delta_n ,\quad \delta_n:= n^{-1/2+3\epsilon},
$$
such that $\theta^\prime_n$ is well outside $I_n$ but nevertheless,  $\theta^\prime_n\to \theta$ for large $n$ (assuming $\epsilon<1/6$). The $3\epsilon$ factor in the exponent is chosen such that the result of Proposition \ref{prop:displaced.null} below holds, but any factor larger than 
$2\epsilon$ suffices. We measure all 
remaining samples in the basis $\{|v_0^{\theta_n^\prime}\rangle, |v_1^{\theta_n^\prime}\rangle \}$ (cf eq. \eqref{eq:v0v1}) to obtain outcomes $X_1,\dots , X_{n}\in \{0,1\}$ 
with probability distribution 
$$
P^{(n)}_\theta = (1-p^{(n)}_\theta, p^{(n)}_\theta), \quad
p^{(n)}_\theta= \sin^2(\theta- \theta^\prime_n).
$$ 

\begin{proposition}
\label{prop:displaced.null}
Assume that $\Theta$ is bounded and $\epsilon<1/10$ is fixed, and let $\tilde{\theta}_n$ be the preliminary estimator based on $\tilde{n}= n^{1-\epsilon}$ samples.

Let $\hat{\theta}_n$ be the estimator 
$$
\hat{\theta}_n := \tilde{\theta}_n+ 
\frac{n^{-1/2+3\epsilon}}{2} -
\frac{n^{1/2-3\epsilon}}{2}\hat{p}_n 
$$
where $\hat{p}_n$ is the empirical estimator of $p^{(n)}_\theta$, i.e.
\begin{equation}\label{eq:p.hat}
\hat{p}_n= \frac{| \{ i: X_i=1 ,~i=1, \dots ,n\}|}{n}.
\end{equation}
Then $\hat{\theta}_n$ is asymptotically optimal in the sense that 
$$
\lim_{n\to \infty} n\mathbb{E}_\theta [(\hat{\theta}_n -\theta)^2] = F^{-1}(\theta) = \frac{1}{4}.
$$
Moreover, $\hat{\theta}_n$ is asymptotically normal, i.e.
$$
 \sqrt{n}(\hat{\theta}_n -\theta) \to N \left (0, \frac{1}{4}\right )
$$
where the convergence holds in distribution.
\end{proposition}

The proof of Proposition \ref{prop:displaced.null} can be found in Appendix \ref{app:proof.prop:displaced.null}. Note that we chose to identify $n$ and $n^\prime=n-n^{1-\epsilon}$ in order to simplify the notation and the proofs, but it is immediate to adapt the reasoning in order to deal with this technicality. We also remark that the assumption $\epsilon <1/10$ is not essential and could be removed at the price of using more involved analysis of the concentration properties of $\tilde{\theta}_n$ and the definition of the displacement parameter $\delta_n$.

\section{Displaced-null measurements in the asymptotic Gaussian picture} \label{sec:qubitLAN}

In this section we cast the null-measurement problem into a companion Gaussian estimation problem which arises in the limit of large sample sizes. The Gaussian approximation is described by the theory of quantum local asymptotic normality (QLAN) developed in \cite{LAN1,LAN3,LAN2,LAN4}. For reader's convenience we review the special case of pure qubit states in section \ref{sec:LAN}.

\subsection{Brief review of local asymptotic normality for pure qubit states}
\label{sec:LAN}

The QLAN theory is closely related to the quantum Central Limit Theorem (QCLT) and  shows that for large $n$ the statistical model describing ensembles of $n$ identically prepared qubits can be approximated (locally in the parameter space) by a single coherent state of a one-mode continuous variables (cv) system, whose mean encodes the unknown qubit rotation angle. We refer to \cite{LAN1,LAN5} for mathematical details and focus here on the intuitive correspondence between qubit ensembles and the cv mode. 

We start with a completely unknown pure qubit state described by a one-dimensional projection $P=|\psi\rangle\langle \psi|$. In the first step we  measure a sub-sample of $\tilde{n}=n^{1-\epsilon}$ systems and obtain a preliminary estimator $\tilde{P}_n = |\tilde{\psi}_n\rangle\langle \tilde{\psi}_n|$. We assume that $\tilde{P}_n$ satisfies a concentration bound similar to the one in equation \eqref{eq:concentration} so that $P$ lies within a ball of size $n^{-1/2+\epsilon}$ around $\tilde{P}_n$ with high probability. For more about the localisation procedure we refer to Appendix \ref{sec:adaptive.argument}.

We now choose the ONB $\{|0\rangle , |1\rangle\}$ such that $|0\rangle:= |\tilde{\psi}_n\rangle$.

Thanks to parameter localisation we can focus our attention on `small' rotations around $|0\rangle$ whose magnitude is of the order $n^{-1/2+\epsilon}$ where $n$ is the sample size and $\epsilon>0$ is small. We parametrise such states as
$$|\psi_{{\bm u}/\sqrt{n}}\rangle := U\left(\frac{\bm u}{\sqrt{n}}\right)|0\rangle = e^{-i (u_1\sigma_y -u_2\sigma_x)/\sqrt{n}} |0\rangle,
$$
where ${\bm u}=(u_1, u_2)$ is a two-dimensional local parameter of magnitude $|{\bm u}|< n^{\epsilon}$.
The joint state of the ensemble of $n$ identically prepared qubits is then 
$$
|\Psi_{\bm u}^n\rangle = |\psi_{{\bm u}/\sqrt{n}}\rangle^{\otimes n}.
$$

We now describe the \emph{Gaussian shift model} which approximates the i.i.d. qubit model in the large sample size limit.
A one mode cv system is specified by canonical coordinates
 $Q,P$  satisfying $[Q,P] =i\mathbb{1}$. These  act on a Hilbert space $\mathcal{H}$ with a orthonormal Fock basis 
 $\{|k\rangle :k\geq 0\}$, such that $a|k\rangle = \sqrt{k}|k-1\rangle$, where $a$ is the annihilation operator $a= (Q+iP)/\sqrt{2}$. The coherent states are defined as 
 $$
 |z\rangle:=e^{-|z|^2/2}\sum_{k=0}^\infty
 \frac{z^k}{\sqrt{k!}} |k\rangle ,\quad z\in \mathbb{C}
 $$
and satisfy
$\langle z|a|z\rangle = z$. In the coherent state $\ket{z}$, the canonical coordinates $Q,P$ have normal distributions $N\left(\sqrt{2}{\rm Re}z,\, \frac{1}{2}\right)$ and $N\left(\sqrt{2}{\rm Im}z, \,\frac{1}{2}\right)$, respectively. In addition, the number operator $N :=a^*a$ has Poisson distribution with intensity $|z|^2$.

We now outline two approaches to QLAN embodying different ways to express the closeness of the multiqubits model  $\{ |\Psi_{\bm u}^n\rangle : |{\bm u}|\leq n^{\epsilon} \} $ to the quantum Gaussian shift model 
 $\{|u_1+ iu_2\rangle : |{\bm u}|\leq n^{\epsilon} \}$. 
By applying the QCLT \cite{Petz2008}, one shows that the collective spin in the `transverse' directions x and y have asymptotically normal distributions
\begin{eqnarray*}
\frac{1}{\sqrt{2n}}S_x(n):= \frac{1}{\sqrt{2n}} \sum_{i=1}^n \sigma_x^{(i)}  & \to & N\left(\sqrt{2}u_1, \frac{1}{2}\right) \\ 
\frac{1}{\sqrt{2n}}S_y(n):= \frac{1}{\sqrt{2n}} \sum_{i=1}^n \sigma_y^{(i)}  & \to & N\left(\sqrt{2}u_2, \frac{1}{2}\right) 
 \end{eqnarray*}
where the arrows represent convergence in distribution with respect to $|\Psi^n_u\rangle$. In fact the convergence holds for the whole `joint distribution' which we write symbolically as
$$
\left(\frac{1}{\sqrt{2n}}S_x(n), \frac{1}{\sqrt{2n}}S_y(n)\,:\,|\Psi_{\bm{u}}^n\rangle \right)\to 
\left(Q,P \,:\, |u_1+ iu_2\rangle  \right).
$$
So, in what concerns the collective spin observables, the joint qubit state converges to a coherent state whose displacement is linear with respect to the local rotation parameters.

An alternative way to formulate the convergence to the Gaussian model is to show that the two models can be mapped into each other by means of physical operations (quantum channels) with asymptotically vanishing error, uniformly over all local parameters $|{\bm u}|\leq n^\epsilon$. Consider the isometric embedding of the symmetric subspace 
$\mathcal{S}_n =(\mathbb{C}^2)^{\otimes_s ^n} $  of  the tensor product 
$(\mathbb{C}^2)^{\otimes n}$ into the Fock space 
 \begin{eqnarray*}
 V_n := 
 \mathcal{S}_n &\to& \mathcal{H}\\
  |k,n\rangle &\mapsto &|k\rangle
 \end{eqnarray*}
where $|k,n\rangle$ is the normalised projection of the vector
$|1\rangle^{\otimes k} \otimes |0\rangle^{\otimes n-k} $ onto $\mathcal{S}_n$. 
The following limits hold 
\cite{LAN1}
\begin{eqnarray*}
&&\lim_{n\to \infty}
\sup_{|{\bm u}|\leq n^{1/2-\eta}}
\left\|V_n |\Psi_{\bm u}^n\rangle -|u_1+iu_2\rangle\right\|=0,\\
&&\lim_{n\to \infty}
\sup_{|{\bm u}|\leq n^{1/2-\eta}}
\left\||\Psi_{\bm u}^n\rangle -V_n^*|u_1+iu_2\rangle\right\| =0.
\end{eqnarray*}
where $\eta>0$ is an arbitrary fixed parameter. In particular, for $\eta<1/2-\epsilon$ the supremum is taken over regions that contain all $|{\bm u}|
<n^{\epsilon}$, which means that the Gaussian approximation holds uniformly over all values of the local parameter arising from the preliminary estimation step.

We now move to describe the relationship between qubit rotations and Gaussian displacements in the QLAN approximation.
Let $U^{n}(\bm{\Delta}):= U(n^{-1/2}\bm{\Delta})^{\otimes n}$ be a qubit rotation by small angles ${\bm \delta}:= n^{-1/2}\bm{\Delta}$ and let $D(\bm{\Delta}) = \exp(-i\sqrt{2}( \Delta_1 P - \Delta_2 Q))$ be the corresponding displacement operator. Then the following commutative diagram shows how QLAN translates (small) rotations into displacements (asymptotically with $n$ and uniformly over local parameters)
\begin{center}
$$
\begin{CD}
|\Psi_{\bm{u}}^n\rangle 
@> 
V_n
>>
|u_1+iu_2\rangle 
\\
@VV
U^n(-\bm{\Delta})
V   
@VV
D(-\bm{\Delta})
V\\
|\Psi_{\bm{u}-\bm{\Delta}}^n \rangle
@>
V_n
>> 
|u_1-\Delta_1+i(u_2-\Delta_2)\rangle
\end{CD}
$$
\end{center}

Notice that also the vertical arrow on the left of the diagram is true asymptotically with $n$ and has to be intended as $\lim_{n\rightarrow +\infty}\|U^n(-\bm{\Delta})|\Psi_{\bm{u}}^n\rangle-|\Psi_{\bm{u}-\bm{\Delta}}^n \rangle\|=0$.

Finally, we note that while the transverse spin components $S_x, S_y$ converge to the canonical coordinates of the cv mode, the collective operator related to the total spin in direction $z$ becomes the number operator $N$. 
Indeed if
$
E_n:= (n\mathbb{1} - S_z)/2
$ 
then
$
E_n |k,n\rangle = k |k,n\rangle
$
so $E_n = V^*_nNV_n$. This correspondence can be extended to small rotations of such operators. Consider the collective operator 
$$
N^n_{\bm{\Delta}}:= 
U^n(\bm{\Delta})(n\mathbb{1}-S_z)U^n(-\bm{\Delta})
$$ 
which corresponds to measuring individual qubits in the basis
$$
|v^{\bm \delta}_0\rangle = U(\bm{\delta})|0\rangle , \quad
|v^{\bm \delta}_1\rangle = U(\bm{\delta})|1\rangle
$$
and adding the resulting $\{0,1\}$ outcomes. In the limit Gaussian model, this corresponds to measuring the displaced number operator 
$N_{\bm{\Delta}}= D(\bm{\Delta})N D(-\bm{\Delta})$. More precisely, the binomial distribution $p^{(n)}_{{\bm u},\bm{\Delta}}$ of $N^n_{\bm{\Delta}}$ computed in the state $|\Psi^n_{\bm{\Delta}}\rangle$ converges to the Poisson distribution of $N_{\bm{\Delta}}$ with respect to the state $|u_1+iu_2\rangle$
$$
\lim_{n\to \infty}
p^{(n)}_{\bm{u},\bm{\Delta}}(k) =
e^{-\|\bm{u}-\bm{\Delta}\|^2}\frac{\|\bm{u}-\bm{\Delta}\|^{2k}}{k!},\qquad k\geq 0.
$$

\subsection{Asymptotic perspective on displaced-null measurements via local asymptotic normality}

We now offer a complementary picture  of the displaced-null measurement schemes outlined in section \ref{sec.one.parameter}, using the QLAN theory of section \ref{sec:LAN}. In the Gaussian limit, the qubits ensemble is replaced by a single coherent state while the qubit null measurement becomes the number operator measurement. The Gaussian picture will illustrate why the null measurement does not work and how this problem can be overcome by using the displaced null strategy.

Consider first the one dimensional model given by equation \eqref{eq:qubit.model}, and let us assume for simplicity that the preliminary estimator takes the value $\tilde{\theta}_n=0$. The general case can be reduced to this by a rotation of the block sphere.

We write $\theta$ in terms of the local parameter $u$ as $\theta=\tilde{\theta}_n+ u/\sqrt{n} = u/\sqrt{n}$ with $|u|\leq n^{\epsilon}$. By employing QLAN we map the i.i.d. model $|\Psi_u^n\rangle$ (approximately) into the limit coherent state model $|u\rangle$. At $\tilde{\theta}_n =0$ the null measurement for an individual qubit is that of $\sigma_z$ (standard basis). On the ensemble level this translates into measuring the collective spin observable $S_z$, which converges to the number operator $N$ in the limit model, cf. section \ref{sec:LAN}. Indeed, at $u=0$ the coherent state is the vacuum which is an eigenstate of $N$.   

As in the qubit case, the number measurement suffers from the non-identifiabilty issue since both $|\pm u\rangle$ states produce the same Poisson distribution (see panels c. and d. in Figure \ref{Fig:nonidentifiability}). 
 
We now interpret the displaced-null measurement in the QLAN picture. Recall that if we measure each qubit in the rotated basis 
$$
|v^{\delta_n}_0\rangle = U((\delta_n,0))|0\rangle , \quad
|v^{\delta_n}_1\rangle = U((\delta_n,0))|1\rangle,
$$
then the non-identifiability is lifted and the parameter can be estimated optimally. The collective spin in this rotated basis is
$$
N^n_{(\Delta_n,0)}:= 
U^n((\Delta_n,0))(\mathbb{1}-S_z)U^n((-\Delta_n,0)).
$$ 
where $\Delta_n = n^{1/2}\delta_n= n^{3\epsilon}$ and by the QLAN correspondence it maps to the displaced number operator 
$$
N_{(\Delta_n,0)} = D((\Delta_n,0))ND((-\Delta_n,0)).
$$ 
In this case the distribution with respect to the state $|u\rangle$ is ${\rm Poisson}(|\Delta_n-u|^2)$, and since $\Delta_n = n^{3\epsilon}\gg |u|$, the model is identifiable, i.e. the correspondence the intensity $|\Delta_n-u|^2$ and $u$ is one-to-one (see panels g. and h. in Figure \ref{Fig:nonidentifiability}). Moreover, for large $n$ the measurement provides an optimal estimator of $u$. Indeed by writing 
\begin{equation}
\label{eq:homodyne.approx}
N_{(\Delta_n,0)} = (a-\Delta_n\mathbf{1})^*(a-\Delta_n\mathbf{1} ) = a^*a 
-\Delta_n (a+a^*) + \Delta_n^2\mathbf{1}
\end{equation}
and noting that the term $a^*a$ is $O(n^{2\epsilon})$ (for $|u|\leq n^\epsilon$) we get
\begin{equation}
\label{eq.N.displaced}
\frac{1}{2}\Delta_n-\frac{1}{2\Delta_n} N_{(\Delta_n,0)}  = \frac{Q}{\sqrt{2} }+ o(1)
\end{equation}
where we recover the well known fact that quadrature (homodyne) measurement can be implemented by displacement and counting. 
By measuring the operator on the lefthand side of \eqref{eq.N.displaced} we obtain an asymptotically optimal estimator of $u$, which corresponds to the qubit estimator constructed in section \ref{sec.one.parameter}.

\section{Multiparameter estimation for pure qudit states} \label{sec:qdits}

In this section we discuss the general case of a multidimensional statistical model for a $d$-dimensional quantum system (qudit).

The first two subsections review the theory of multiparameter estimation and how QLAN is used to establish the asymptotic achievability of the Holevo bound. This circle of ideas will be helpful in understanding the results in the following sections which deal with displaced-null estimation of qudit models. In particular we show that displaced-null measurements achieve the following:\begin{itemize}
    
     \item[1.] The Holevo bound for completely unknown pure state models where the figure of merit is given by the Bures distance (Proposition \ref{prop:optlqd});\item[2.] The quantum Cram\'er-Rao bound in statistical models where the parameters can be estimated simultaneously (Proposition \ref{th:QCRB-achievability-null}), providing an operational implementation for the results in \cite{NullQFI1,NullQFI2,NullQFI3};
    \item[3.] The Holevo bound for completely general pure state models and figures of merit (Theorem \ref{thm:dnmgeneral}).
\end{itemize}
Since the two stage strategy is discussed in detail in Appendix \ref{sec:adaptive.argument}, we do not give a detailed account of the preliminary stage and assume that the parameter has been localised in a neighbourhood of size $n^{-1/2+\epsilon}$ around a preliminary estimator with probability that converges to $1$ exponentially fast in $n$.

\subsection{Multiparameter estimation}

Let us consider the problem of estimating the parameter $\bm{\theta}$ belonging to an open set $\Theta \subseteq \mathbb{R}^m$ given the corresponding family of states $\rho_{\bm{\theta}}$ of a $d$-dimensional quantum system. Given a measurement with POVM ${\cal M}:=\{M_0,\dots,M_p\}$, the CFI matrix is given by 
$$
I_{\cal M}(\bm{\theta})_{ij}=
\mathbb{E}_\theta\left [ \frac{\partial \log p_{\bm \theta}}{\partial \theta_i} \frac{\partial \log p_{\bm \theta}}{\partial \theta_j} \right ]. 
$$
The QFI matrix is 
$F(\theta)_{ij} = {\rm Tr} (\rho_{\bm \theta} 
\mathcal{L}^i_{\bm \theta} \circ \mathcal{L}^j_{\bm \theta})$ where $\mathcal{L}^j_{\bm \theta}$ are the SLDs satisfying $\partial_i\rho_{\bm \theta}= \mathcal{L}^j_{\bm \theta}\circ \rho_{\bm \theta}$ and $\circ$ denotes the symmetric product $A\circ B = (AB+BA)/2$. If $\hat{\bm \theta}$ is an unbiased estimator then the multidimensional QCRB states that its covariance matrix is lower bounded as
\begin{equation}\label{eq:multidimQCRB}
{\rm Cov}_{{\bm \theta}}(\hat{\bm{\theta}}):=\mathbb{E}_{\bm \theta}[(\hat{\bm{\theta}}-{\bm \theta})(\hat{\bm{\theta}}-{\bm \theta})^T]\geq I_{\cal M}(\bm{\theta})^{-1}\geq F(\bm \theta)^{-1}.
\end{equation}
In general, the second lower bound is not achievable even asymptotically. Roughly, this is due to the fact that the optimal measurements for estimating the different components of ${\bm \theta}$ are incompatible with each other. The precise condition for the achievability of the QCRB is \cite{Ragy16,RafalReview}
\begin{equation} \label{eq:achiev}
 \Tr(\rho_{\bm \theta}[{\cal L}^{i}_{\bm \theta}, {\cal L}^{j}_{\bm \theta}])=0, \quad i,j=1,\dots, m.   
\end{equation}
which in the case of a pure statistical model $\ket{\psi_{\bm \theta}}$ becomes
\begin{equation} \label{eq:pureach}
{\rm Im}(\langle \partial_{\theta_i}\psi|\partial_{\theta_j}\psi \rangle)=0, \quad i,j=1,\dots, m.
\end{equation}
When the QCRB is not achievable, one may look for measurements that optimise a specific figure of merit. The simplest example is that of a quadratic form with positive weight matrix $W$ 
\[
R_W(\hat{\bm \theta}, {\bm \theta}) = 
\mathbb{E}_\theta[(\hat{\bm{\theta}}-\bm{\theta})^TW(\hat{\bm{\theta}}-\bm{\theta})]
\]
This choice is not as restrictive as it may seem since many interesting loss functions have a local quadratic approximation which determines the leading term of the asymptotic risk. A straightforward lower bound on $R_W$ can be obtained by taking the trace with $W$ in \eqref{eq:multidimQCRB} but this bound is not achievable either. A better one is the Holevo bound \cite{Holevo2011}
\begin{eqnarray}\label{eq.Holevo}
&&\Tr(W{\rm Cov}_{{\bm \theta}}(\hat{\bm{\theta}})) \geq {\cal H}^W(\bm{\theta})\\
&:=&\min_{{\bf X}} {\rm Tr} ({\rm Re} (Z({\bf X})) W) +
{\rm Tr} \left| W^{1/2} {\rm Im} (Z({\bf X}) )W^{1/2}\right|
\nonumber
\end{eqnarray}
where the minimum runs over m-tuples of selfadjoint operators $\bm{X} =(X_1,\dots, X_m)^T$ acting on the system, which satisfy the constraints $\Tr(\nabla_{\bm \theta}\rho_{\bm \theta} \bm{X}^T)=\bm{1}$, and 
${Z}({\bf X})$ is the $m\times m$ complex matrix with 
entries ${Z}({\bf X})_{ij} = {\rm Tr}(\rho_{\bm \theta}X_i X_j)$.
Unlike the multidimensional QCRB, the Holevo bound is achievable asymptotically in the i.i.d. scenario \cite{LAN4,RafalReview}. In the next two section we will give an intuitive explanation based on the QLAN theory.

\subsection{Gaussian shift models and QLAN}
\label{sec:LAN-qudits}

Quantum Gaussian shift models play a fundamental role in quantum estimation theory \cite{Holevo2011}. Such models are fairly tractable in that the Holevo bound is achievable with simple linear measurements. More importantly, Gaussian shift models arise as limits of  i.i.d. models in the QLAN theory, which  offers a recipe for constructing estimators which achieve the Holevo bound asymptotically in the i.i.d. setting. 

For the purposes of this work, the asymptotic Gaussian limit offers a clean intuition about the working of the proposed estimators, but is not explicitly used in deriving the mathematical results. We therefore keep the presentation on a intuitive level and refer to the papers \cite{LAN4,RafalReview,LAN6} for more details. 

In this subsection we recall the essentials of multiparameter estimation in a pure quantum Gaussian shift model and of QLAN theory for pure states of finite dimensional quantum systems, extending what we already presented in the case of qubits in Section \ref{sec:qubitLAN}.

\subsubsection{Achieving the Holevo bound in a pure Gaussian shift model} 
\label{sec:Holevo.Gaussian.shift}

Consider a cv system consisting of $(d-1)$ modes. The corresponding Hilbert space ${\cal H}$ is the multimode Fock space which will be identified with the tensor product of $d-1$ copies of the single mode spaces, with ONB given by the Fock vectors $|{\bf k}\rangle := |k_1\rangle\otimes \dots \otimes |k_{d-1}\rangle$, with ${\bf k}=(k_1,\dots, k_{d-1}) \in \mathbb{N}^{d-1}$. The creation/annihilation operators, canonical coordinates and number operator of the individual modes are denoted $a_i^*$, $a_i$, $Q_i=(a_i+a_i^*)/\sqrt{2}$, $P_i=(a_i-a_i^*)/(\sqrt{2}i)$ and $N_i=a^*_ia_i$ for $i=1,\dots, {d-1}$.

We denote by  $\ket{\bm{z}}=\ket{z_1} \otimes \dots \otimes \ket{z_{d-1}}$ the multimode coherent states with $\bm{z}=(z_1,\dots, z_{d-1}) \in \mathbb{C}^{d-1}$, so that $Q_i$ and $P_i$ have normal distribution with variance $1/2$ and mean $\sqrt{2}{\rm Re}(z_i)$ and $\sqrt{2}{\rm Im}(z_i)$, respectively, while $N_i$ have Poisson distributions with intensities $|z_i|^2$. We denote by ${\bf R}:= (Q_1,\dots, Q_{d-1}, P_1, \dots, P_{d-1})^T$ the vector of canonical coordinates which satisfy commutation relations $[R_i,R_j]=i\Omega_{ij}$ where $\Omega$ is the $2(d-1)\times 2(d-1)$ symplectic matrix 
$$
\Omega=\begin{pmatrix} \bf{0} & \bf{1} \\
-\bf{1} & \bf{0} \end{pmatrix}.$$

Let ${\bm u}\in \mathbb{R}^m$ be an unknown parameter and let $\mathcal{G}$ be the \emph{quantum Gaussian shift model}
\[
\mathcal{G}:=\{\ket{C\bm{u}} : {\bm u}\in \mathbb{R}^m\}
\]
where $C:\mathbb{R}^m \rightarrow \mathbb{C}^{d-1}$ is a linear map. The goal is to estimate ${\bm u}$ optimally for a given figure of merit.

Denoting the entries of $C$ as $C_{k,j}=c^q_{kj}+ic^p_{kj}$ for $k=1,\dots, d-1$ and $j=1,\dots, m$, we call $D$ the real $2(d-1)\times m$ matrix with elements $D_{k,j}=\sqrt{2} c^q_{kj}, D_{k+(d-1),j}=\sqrt{2} c^p_{kj} $ with $k=1,\dots d-1$; notice that $\mathbb{E}_{\bm u}[{\bf R}]=D{\bm u}$. We remark that $\bm{u}$ is identifiable if and only if $D$ has rank equal to $m$. The quantum Fisher information matrix is independent of ${\bm u}$ and is given by $F=2 D^T D>0$. 

Let us first consider the case when the QCRB is achievable (in which case it leads to the Holevo bound by tracing with $W$). Condition \eqref{eq:pureach} amounts to $C^*C$ being a \emph{real} matrix which is equivalent to $D^T \Omega D=0$ and the fact that the generators of the Gaussian shift model $\mathcal{G}$
$$
S_j=\sum_{k=1}^{d-1}c_{kj}^q P_k - c_{kj}^p Q_k = (D^T \Omega{\bf R})_j,\quad j=1,\dots m
$$
commute with each other.

An optimal unbiased measurement consists of simultaneously measuring the commuting operators $\bm{Z}=\Sigma^{-1} D^T {\bf R}$, where $\Sigma:=D^TD= F/2$. Indeed
\begin{itemize}
\item $[{\bf Z},{\bf Z}^T]=\Sigma^{-1} D^T\Omega D \Sigma^{-1}=0$ (commutativity),
\item $\mathbb{E}_{\bm{u}}[{\bf Z}]=\Sigma^{-1} D^TD{\bm u}={\bm u}$ (unbiasedeness),
\item ${\rm Cov}_{\bm u}({\bf Z})=\Sigma^{-1}/2$ (achieves the QCRB).
\end{itemize}
Consider now the case when the QCRB is not achievable. For a given positive weight matrix $W$, the corresponding Holevo bound is given by
\begin{eqnarray} \label{eq:linHB}
&&{\rm Tr} ({\rm Cov}_{\bm u}(\hat{\bm u}) W)
\geq
{\cal H}^W(\mathcal{G})\\
&&:=\min_B \frac{1}{2} \left ( \Tr(WBB^T) + \Tr(|\sqrt{W} B \Omega B^T \sqrt{W}|)\right) ,
\nonumber
\end{eqnarray}
where ${\hat {\bm u}}$ is an unbiased estimator and the minimum is taken over real $m\times 2(d-1)$ matrices $B$ such that $BD=\bf{1}$. The Holevo bound can be saturated by coupling the system with another ancillary $(d-1)$-dimensional cv system in the vacuum state and with position and momentum vector that we denote by ${\bf R}^\prime=(Q^\prime_1, \dots, Q^\prime_{d-1}, P^\prime_1,\dots, P_{d-1}^\prime)^T$. In order to estimate ${\bm u}$, we consider a vector of quadratures of the form $\bm{Z}=B{\bf R}+B^\prime {\bf R}^\prime$ for $B,B^\prime$ real $m\times (d-1)$ matrices and we require that $\bm{Z}$ is unbiased and belongs to a commutative family:
\begin{itemize}
\item $B^\prime \Omega B^{\prime T}=-B \Omega B^T$ (commutativity of the $Z_i$'s),
\item $\langle
C{\bm u} \otimes {\bf 0}|{\bf Z}|C{\bm u} \otimes {\bf 0}\rangle ={\bm u} \Leftrightarrow BD= \mathbf{1}$   (unbiasedeness).
\end{itemize}
The corresponding risk is
\[
R_{\bf Z} =  
\frac{1}{2} \left ( \Tr(WBB^T) + \Tr(WB^\prime B^{\prime T})\right)
\]
and by minimizing over $B$ and $B^\prime$ one obtains the expression of the Holevo bound in Equation \eqref{eq:linHB}. Therefore, given a minimiser $(B^\star,B^{\prime \star})$, the corresponding vector of quadratures $\bm{Z}^\star$ is an optimal estimator for any $\bm{u}$. 

To summarise, in the pure Gaussian shift model there always exists a set of commuting quadratures ${\bf Z}^\star$ of a doubled up system that achieves the Holevo bound; in the case when the QCRB is achievable, one does not need an ancilla.

For the discussion in section \ref{subsec:HOdispl} it is useful to consider the following implementation of the optimal measurement. Let $(\tilde{Q}_1,\dots, \tilde{Q}_{2(d-1)}, \tilde{P}_1,\dots ,\tilde{P}_{2(d-1)})$ be a choice of vacuum modes of the doubled-up cv system such that 
$
{\bf Z}^\star=T\tilde{\bf Q}
$
where $\tilde{\bf Q}=(\tilde{Q}_1, \dots, \tilde{Q}_m)^T
$ 
for some $m\times m$ invertible matrix $T$ with real entries. Up to classical post-processing, measuring $Z_1,\dots , Z_m$ is equivalent to measuring $\tilde{Q}_1, \dots, \tilde{Q}_m$. 
If we denote the outcomes of the latter by 
$\tilde{\bf q}:=(\tilde{q}_1,\dots \tilde{q}_m)$ then an optimal unbiased estimator of ${\bf u}$ is given by
$\hat{\bf u} = T\tilde{\bf q}$.

\subsubsection{QLAN for i.i.d. pure qudit models} \label{sec.QLAN.qudits}

The idea of QLAN is that the states in a shrinking neigbourhood of a fixed state can be approximated by a Gaussian shift model. In the next section we will show how this can be used as an estimation tool, but here we describe the general structure of QLAN for \emph{pure} qudit states.

We choose 
the centre of the neighbourhood to be the first vector of an ONB $\{|0\rangle, \dots, |d-1\rangle\}$, and parametrise the local neighborhood of states around $\ket{0}$ as
\begin{equation}\label{eqn:psiu}
    \ket{\psi_{{\bm u}/\sqrt{n}}}=\exp \left ( -i\sum_{k=1}^{d-1} (u_1^k \sigma_y^k - u_2^k \sigma_x^k)/\sqrt{n} \right ) \ket {0}
\end{equation}
for $\bm{u}=({\bm u}_1, \bm{u}_2) \in \mathbb{R}^{2(d-1)}$, $\|{\bm u}\| \leq n^{\epsilon}$, $\sigma_y^k=i \ket{k}\bra{0} -i \ket{0}\bra{k}$ and $\sigma_x^k=\ket{k}\bra{0} + \ket{0}\bra{k}$. As in the qubit case, the appropriately rescaled collective variables converge to position, momentum and number operators in `joint distribution' with respect to $\ket{\Psi^n_{\bm u}}:= \ket{\psi_{{\bm u}/\sqrt{n}}}^{\otimes n}$
\begin{align*}
\left(\frac{1}{\sqrt{2n}}S^k_x(n), \frac{1}{\sqrt{2n}}S^k_y(n),n\mathbb{1}-S^k_z(n) \,:\,|\Psi_{\bm u}^n\rangle \right )\\
\to \left(Q_k,P_k,N_k \,:\, \ket{{\bf z}=\bm{u}_1+i\bm{u}_2}  \right),
\end{align*}
where $S^k_\alpha(n)=\sum_{l=1}^n(\sigma_\alpha^k)^{(l)} $ for $\alpha \in \{x,y,z\}$. More generally, we have a (real) linear map between the orthogonal complement of $|0\rangle$ and Gaussian quadratures: for every vector $\ket{v}=\sum_{k=1}^{d-1}(v_x^k+iv_y^k) \ket{k}$ we construct the corresponding Pauli operator $\sigma(v)= \ket{v}\bra{0}+\ket{0}\bra{v}$ and the following CLT  holds
\begin{equation} \label{eq:QCLTcorr}
\left( 
\frac{1}{\sqrt{2n}}S_v(n) :\ket{\Psi^n_u } \right) \to 
\left( X(v): \ket{{\bf z}=\bm{u}_1+i\bm{u}_2} \right)
\end{equation}
where $S_v(n):= \sum_{l=1}^n \sigma(v)^{(l)}$ and 
$ X(v):= \sum_{k=1}^n v_x^k Q_k + v_y^k P_k$.

In addition to the QCLT, the following strong QLAN statement holds: the statistical model $\{\ket{\Psi_{\bm u}^n}\}$ can be approximated by a pure Gaussian shift model in the sense that
\begin{eqnarray}
\label{eq:LeCam.qudit.1}
&&\lim_{n\to \infty}
\sup_{|{\bm u}|\leq n^{1/2-\eta}}
\left\|V_n |\Psi_{\bm u}^n\rangle -|{\bm u}_1 +i{\bm u}_2 \rangle\right\|=0,\\
\label{eq:LeCam.qudit.2}
&&\lim_{n\to \infty}
\sup_{|{\bm u}|\leq n^{1/2-\eta}}
\left\||\Psi_{\bm u}^n\rangle -V_n^*|{\bm u}_1+i{\bm u}_2\rangle\right\| =0
\end{eqnarray}
for any fixed $0<\eta<1/2$. $V_n$ is the isometric embedding of the symmetric subspace $\mathcal{S}^{(n)}_d := \left(\mathbb{C}^d \right)^{\otimes_s n}$ into a $(d-1)$-mode Fock space ${\cal H}$ (cf. previous section) characterised by  
\begin{eqnarray}
 V_n :
 \mathcal{S}^{(n)}_d &\to& \mathcal{H}\nonumber\\
  |{\bf k};n\rangle &\mapsto & |{\bf k} \rangle
  \label{eq:isometry}
\end{eqnarray}
where $\ket{{\bf k};n}$ denotes the normalised vector obtained by 
symmetrising 
$$
\ket{1}^{\otimes k_1} \otimes \dots \otimes\ket{{d-1}}^{\otimes k_{d-1}} \otimes \ket{0}^{\otimes (n-(k_1+\dots+ k_{d-1}))}.
$$
As in the qubits case, the Gaussian approximation maps small rotations into displacements of the coherent states. Consider collective qubit rotations by small angles ${\bm \delta}:  = n^{-1/2}\bm{\Delta}$
$$U^{n}(\bm{\Delta}):= \left( \exp\left(-i\sum_{k=1}^{d-1} ( n^{-1/2}\Delta_1^k \sigma^k_{y} -n^{-1/2}\Delta_2^k \sigma^k_{x} )\right)\right )^{\otimes n}$$
and the corresponding displacement operators
$$D(\bm{\Delta}) = \exp\left(-i\sum_{k=1}^{d-1} ( \Delta_1^k P_k -\Delta_2^k Q_k )
\right).$$
The diagram below conveys the asymptotic covariance between rotations and displacements, where the arrows should be interpreted in the same way as the strong convergence equations \eqref{eq:LeCam.qudit.1} and \eqref{eq:LeCam.qudit.2}
$$
\begin{CD}
|\Psi_{\bm u}^n\rangle 
@> 
V_n
>>
|{\bm u}_1+i{\bm u}_2\rangle 
\\
@VV
U^n(-{\bm \Delta})
V   
@VV
D(-{\bm \Delta})
V\\
|\Psi_{{\bm u}-{\bm \Delta}}^n \rangle
@>
V_n
>> 
|{\bm u}_1-{\bm \Delta}_1+i({\bm u}_2-{\bm \Delta}_2)  \rangle
\end{CD}
$$

A similar correspondence holds for measurements with respect to rotated bases and displaced number operators
$$
\begin{CD}
|\Psi_{\bm u}^n\rangle 
@> 
V_n
>>
|{\bm u}_1+i{\bm u}_2\rangle 
\\
@VV
N^i_{\bm{\Delta}}(n)
V   
@VV
N^i_{\bm{\Delta}}
V\\
p^n({\bm u}, {\bm \Delta})
@>
V_n
>> 
{\rm Poisson}(\| {\bm u}_1-{\bm \Delta}_1+i({\bm u}_2-{\bm \Delta}_2) \|^2)
\end{CD}
$$

 More precisely, suppose we measure the commuting family of operators $\{N^i_{\bm{\Delta}}(n),  i=1,\dots ,d-1\}$ given by 
 $$
 N^i_{\bm{\Delta}}(n):= 
 U^n(-\bm{\Delta})(n\mathbb{1}-S^i_z(n))U^n(\bm{\Delta}) 
\quad i=1,\dots, d-1,
 $$ 
 which amounts to measuring individual qudits in the basis
 $$
 |v^{\bm \delta}_i\rangle = U(\bm{\delta})|i\rangle \quad i=0,\dots,d-1
 $$
 and collecting the total counts for individual outcomes in $\{0,\dots, d-1\}$. In the Gaussian model this corresponds to measuring the displaced number operators 
 $N^i_{\bm{\Delta}}= D(-\bm{\Delta}) N^i D(\bm{\Delta})$, and by QLAN, the multinomial distribution $p^n(u,\bm{\Delta})$ of $N^i_{\bm{\Delta}}(n)$  converges to the law of the vector of Poisson random variables obtained by measuring $N^i_{\bm{\Delta}}$ with respect to the state $|\bm{u}_1+i\bm{u}_2\rangle$.

 \subsection{Achieving the Holevo bound for pure qudit states via QLAN} \label{sec:LANHolevo}

We will now treat a general pure states statistical model and show how one can use QLAN to achieve the Holevo bound \eqref{eq.Holevo} asymptotically with the sample size. Let $\ket{\psi_{\bm{\theta}}}$ be a statistical model where $\bm{\theta}=(\theta^j)_{j=1}^m$ belongs to some open set $ \Theta \subset \mathbb{R}^m$ with $m \leq 2 (d-1)$ and the parameter is assumed to be identifiable. Given an ensemble of $n$ copies of the unknown state, we would like to devise a measurement strategy and estimation procedure which attains the smallest average error (risk), asymptotically with $n$. For mixed states, a general solution has been discussed in \cite{RafalReview} where it is shown how the Holevo bound can be achieved asymptotically using the QLAN machinery. Here we adapt this method to the case of pure state models. 

In brief, the procedure involves three steps. We first use $\tilde{n}=n^{1-\epsilon}$ samples to produce a preliminary estimator $\tilde{\bm{\theta}}_n$ and write ${\bm \theta} =\tilde{\bm \theta}_n + {\bm u}/\sqrt{n}$ where ${\bm u}$ is the local parameter satisfying $\|{\bm u}\|\leq n^{\epsilon}$ (with high probability). We chooose an ONB 
$\{|0\rangle, \dots |d-1\rangle\}$ such that 
$|\psi_{\tilde{\bm{\theta}}_n}\rangle = |0\rangle$ and use the QLAN isometry $V_n$ (cf. equation \ref{eq:isometry}) to map the remaining qubits $|\Psi^n_{\bm u}\rangle := |\psi_{\tilde{\bm{\theta}}_n+{\bm u}/\sqrt{n}}\rangle^{\otimes n} $ approximately into the Gaussian state $|C{\bm u}\rangle$. We then use the method described in section \ref{sec:Holevo.Gaussian.shift} to estimate the unknown parameter and achieve the Holevo bound.

We start by expressing the local states as small rotations around $|0\rangle$
\begin{eqnarray} \label{eq:genmod}
&&|\psi_{\bm{\tilde{\theta}}_n + {\bm u}/\sqrt{n}}\rangle \\
&&= 
\exp\left(-i 
\sum_{k=1}^{d-1}
\left(
f^q_{k}\left(
\frac{{\bm u}}{\sqrt{n}}
\right)
\sigma^k_{y} - 
f^p_{k}
\left(
\frac{{\bm u}}{\sqrt{n}}
\right) \sigma^k_{x}\right)\right)|0\rangle
\nonumber
\end{eqnarray}
where $f^q_k$ and $f^p_k$ are real functions and $\sigma_y^k$ and $\sigma_x^k$ are the Pauli matrices of equation \eqref{eqn:psiu}. We now `linearise' the generators of the rotations  and define
\begin{equation} \label{eq:linmod}
|\tilde{\psi}_{{\bm u}/\sqrt{n}}\rangle : =
\exp\left(-i\sum_{j=1}^m  u_j S_j/\sqrt{n}
\right) |0\rangle 
\end{equation}
where 
$$
S_j = \sum_{k=1}^{d-1}(c^q_{kj}\sigma^k_{y} - c^p_{kj} \sigma^k_{x}), \quad 
c^{q,p}_{kj}= \left.\partial_j f^{q,p}_k({\bm u})\right|_{{\bm u} ={\bf 0}}.
$$
We denote the ensemble state of the linearised model $|\widetilde{\Psi}^n_{\bm u}\rangle := |\tilde{\psi}_{{\bm u}/\sqrt{n}}\rangle^{\otimes n}$.
The following lemma shows that the original and the `linearised' models are locally undistinguishable in the asymptotic limit. 
\begin{lemma} \label{lem:linear}
With the above notations if $\epsilon< 1/6$ one has
$$
\lim_{n\to \infty}\sup_{\|{\bm u}\|\leq n^\epsilon} 
\| |\Psi^n_{\bm u}\rangle\langle \Psi^n_{\bm u}|-  |\widetilde{\Psi}^n_{\bm u}\rangle\langle \widetilde{\Psi}^n_{\bm u}| \|_1 =0
$$
where $\|\cdot \|_1$ denotes the trace distance.
\end{lemma}
The proof of Lemma \ref{lem:linear} can be found in Appendix \ref{sec:lemproof}. Thanks to such uniform approximation  results, one can replace the original model with the linearised one without affecting the asymptotic estimation analysis. We denote the latter by
$$
\mathcal{Q}_n := 
\{ |\widetilde{\Psi}^n_{\bm u} \rangle: \, {\bm u}\ \in \mathbb{R}^m, \|{\bm u}\| \leq n^{\epsilon}\}.
$$

Let us now consider the second ingredient of the estimation problem, the risk (figure of merit). We fix a loss function $L:\Theta\times \Theta\to \mathbb{R}_+$, so that the risk of an estimator $\hat{\bm \theta}_n$ at ${\bm \theta}$ is $R(\hat{\bm{\theta}}_n, {\bm{\theta}}) = \mathbb{E}_{\bm{\theta}} [L( \hat{\bm{\theta}}_n, {\bm{\theta}})] $.
We assume that the loss function is locally quadratic around any point and in particular
$$
L(\tilde{\bm{\theta}}_n +{\bm u}, \tilde{\bm{\theta}}_n +{\bm v} ) \approx \sum_{i,j=1}^m w_{ij}(\tilde{\bm{\theta}}_n)(u_i-v_i)(u_j-v_j)
$$
for a strictly positive weight matrix function ${\bm \theta}^\prime \mapsto W({\bm \theta}^\prime) = (w_{ij}({\bm \theta}^\prime))$ (which we assume to be continuous in ${\bm \theta}^\prime$). In asymptotics, $\tilde{\bm \theta}_n \to {\bm \theta}$ and the loss function can be replaced by its quadratic approximation at the true parameter ${\bm \theta}$ without affecting the leading contribution to the estimation risk. We denote $W:=W({\bm \theta})$.

Returning to the original estimation problem, we now show how QLAN can be used to construct an estimator which achieves the Holevo bound asymptotically. 

We couple each system with a $d$-dimensional ancillary system in state $|0^\prime\rangle$ and fix an ONB for the ancilla ${\cal B}^\prime=\{\ket{0^\prime}, \dots, \ket{d-1^\prime}\}$. The extended i.i.d. statistical model is $\ket{\Psi^n_{\bm u}} \otimes \ket{0^\prime}^{\otimes n}$. By quantum LAN, the joint ensemble can be approximated by a pure Gaussian shift model coupled with an ancillary $(d-1)$-modes cv system prepared in the vacuum: $\ket{C\bm{u}} \otimes \ket{\bf 0}$ where $C$ is the $(d-1)\times m$ complex matrix with entries $C_{kj} = c_{kj}^q+ic_{kj}^p$; more precisely we map the two qudit ensembles into their Fock spaces by means of a tensor of isometries as in equation \eqref{eq:isometry} and we consider the $2(d-1)$ modes which correspond to the linear space $ \mathcal{L}: ={\rm Lin} \{|0\rangle \otimes |i^\prime\rangle, |i\rangle\otimes |0^\prime\rangle : i=1,\dots d-1\} $ (which contains $\{\ket{\psi_{\bm \theta}} \otimes \ket{0^\prime}\}_{{\bm \theta}\in \Theta}$). Alternatively, one can map the original ensemble to the the cv space and \emph{then} add a second cv system in the vacuum state. The reason we chose to add an ancillary ensemble at the beginning is because this same setup will be used in the next section in the context of displaced-null measurements.

We now apply the optimal measurement for the Gaussian shift model $|C{\bm u}\rangle$ with weight matrix $W$, as described in section \ref{sec:Holevo.Gaussian.shift}. This involves measuring commuting quadratures of the doubled up cv system, such that the resulting estimator 
$\hat{\bm u}_n$ achieves the Gaussian Holevo bound \eqref{eq:linHB} in the limit of large $n$. 

Thanks to the parameter localisation and LAN, the asymptotic (rescaled) risk of the corresponding 'global' estimator  $\hat{\bm \theta}_n = \tilde{\bm \theta}_n + \hat{\bm u}_n/\sqrt{n}$ 
satisfies
$$
\lim_{n\to\infty}
nR(\hat{\bm{\theta}}_n, {\bm{\theta}}) = \mathcal{H}^W(\mathcal{G}).
$$
Finally we note that the expressions of the Holevo bound \eqref{eq.Holevo} in i.i.d. model $|\psi_{\bm 
\theta}\rangle$ with loss function $L$, and the corresponding Gaussian shift model $|C{\bm u}\rangle$ with weight matrix $W$ coincide: $\mathcal{H}^W({\bm \theta})=\mathcal{H}^W (\mathcal{G})$. Indeed, since $\rho_{\bm \theta}=|\psi_{\bm \theta}\rangle\langle \psi_{\bm \theta}|$ is a pure state, the minimisation in \eqref{eq.Holevo} can be 
restricted to operators ${\bm X} = (X_1, \dots, X_m)$ such that 
$P X_i P = P^{\perp} X_i P^{\perp}=0$ where $P= \rho_{\bm \theta}, P^{\perp} = \mathbf{1} -P$. In this case the two Holevo bounds coincide after making the identification 
$B_{j,k} = \sqrt{2}{\rm Re} \langle k|X_j|0\rangle, 
B_{j+d-1,k} = \sqrt{2}{\rm Im} \langle k|X_j|0\rangle$.

\subsection{Achieving the Holevo bound with displaced-null measurements} \label{subsec:HOdispl}

In this section we show how displaced-null measurements offer an alternative strategy to the one presented in the previous section, for optimal estimation  in a general finite dimensional pure statistical model $\ket{\psi_{\bm{\theta}}}$ with ${\bm \theta} \in \Theta \subset \mathbb{R}^m$. As before, we assume that the risk function $L:\Theta \times \Theta \rightarrow \mathbb{R}_+$ has a continuous quadratic local approximation given by the matrix valued function $W({\bm \theta})$.

The first steps are the same as in the estimation procedure in section \ref{sec:LANHolevo}: we use $\tilde{n}=n^{1-\epsilon}$ samples to produce a preliminary estimator $\tilde{\bm{\theta}}_n$
and we write ${\bm \theta} =\tilde{\bm \theta}_n + {\bm u}/\sqrt{n}$ where ${\bm u}$ is the local parameter such that $\|{\bm u}\| \leq n^{\epsilon}$ with high probability. We choose an ONB ${\cal B}=\{\ket{0}, \dots, \ket{d-1}\}$ such that 
$\ket{0}: = \ket{\psi_{\bm{\tilde{\theta}}_n}}$ and apply Lemma \ref{lem:linear} to approximate the local model as in equation \eqref{eq:linmod}. We couple each system with an ancillary qudit in state $|0^\prime\rangle$. By QLAN, the joint model is approximated by the Gaussian shift model consisting of coherent states $\ket{C{\bm u}} \otimes \ket{0}$ of a $2(d-1)$-modes cv system.

As detailed in section \ref{sec:Holevo.Gaussian.shift}, the Holevo bound for the Gaussian shift can be attained by measuring a certain set of canonical coordinates $\tilde{Q}:= (\widetilde{Q}_1,\dots, \widetilde{Q}_m)$ of the doubled-up systems. In turn, this provides an asymptotically optimal measurement for the i.i.d. qudit model as explained in section  \ref{sec:LANHolevo}. Instead of measuring these quadratures, here we adopt the displaced-null measurements philosophy used in section \ref{sec:qubits}, which achieves the same asymptotic risk. This means that one measures the commuting set of displaced number operators
$
\tilde{N}^j_{{\bm \Delta}_n}=D(-{{\bm \Delta}_n})\widetilde{N}^j D({\bm \Delta}_n)
$
where $\widetilde{N}^j=\widetilde{a}^*_j\widetilde{a}_j $ is the  number operator corresponding to the mode $(\widetilde{Q}_j,\widetilde{P}_j)$ and
$$D({\bm \Delta}_n)=\exp\left(-i\Delta_n \sum_{k=1}^{m} \tilde{P}_k\right), \quad \Delta_n=\sqrt{n} \delta_n=n^{3\epsilon}.$$ 
We note that
\[
\tilde{N}^j_{{\bm \Delta}_n}=(\widetilde{a}_j-n^{3\epsilon}{\bf 1})^*(\widetilde{a}_j-n^{3\epsilon}{\bf 1}) = n^{6\epsilon} {\bf 1} -\sqrt{2}\widetilde{Q}_jn^{3\epsilon} +\widetilde{N}^j,
\]
so for large $n$, measuring $\tilde{N}^j_{{\bm \Delta}_n}$ is equivalent to measuring $\widetilde{Q}_j$. We recall that by measuring 
${\bf Z}^\star := T\widetilde{{\bf Q}}$ we obtain an optimal unbiased estimator of ${\bf u}$, where  $T$ is the invertible matrix defined at the end of section \ref{sec:Holevo.Gaussian.shift}. Therefore, using the above equation we can construct  
an (asymptotically) optimal estimator given by the outcomes of the following set of commuting operators 
\[
\sum_{k=1}^m T_{jk} \left (\frac{n^{3\epsilon}}{\sqrt{2}} {\bf 1}-\frac{n^{-3\epsilon}}{\sqrt{2}}\tilde{N}^k_{{\bm \Delta}_n} \right ) 
\approx  Z_i
\]
We are now ready to translate the above cv measurement into its corresponding projective qudit measurement using the correspondence between displaced number operators measurements and rotated bases, described in section \ref{sec.QLAN.qudits}.

Using the general CLT map \eqref{eq:QCLTcorr}, we identify 
vectors $\{ \ket{\tilde{1}}, \dots, \ket{\widetilde{m}}\}$ in the orthogonal complement of $|\tilde{0}\rangle=\ket{0}\otimes \ket{0^\prime}$ such that their corresponding limit quadratures are
$X(|\tilde{k}\rangle) = \tilde{Q}_k$, for $k=1,\dots, m$. By virtue of the CLT the vectors $|\tilde{0}\rangle, \ket{\tilde{1}}, \dots, \ket{\widetilde{m}}$ are normalised and orthogonal to each other, so we can complete the set to an ONB 
$\tilde{\mathcal{B}} :=\{|\tilde{0}\rangle, \dots |\widetilde{d^2-1}\rangle\}$ of $\mathbb{C}^d \otimes\mathbb{C}^d$ where the remaining vectors are chosen arbitrarily. Now let $\tilde{\mathcal{B}}_n$ be the rotated basis 
\[
|v^{\delta_n}_j\rangle = 
U(\delta_n)|\tilde{j}\rangle = 
\exp\left(-i\delta_n \sum_{k=1}^{m} \sigma(i\tilde{k})\right) |\tilde{j}\rangle 
\]
for $\delta_n=n^{-1/2+3\epsilon}$ and 
$\sigma(i\tilde{k}) := -i|\tilde{0}\langle \tilde{k}| + i|\tilde{k}\rangle\langle \tilde{0}|.$ Note that $\tilde{\mathcal{B}}_n$ is a small rotation of the basis $\mathcal{B}$ which contains the reference state 
    $|\tilde{0}\rangle = |0\rangle\otimes \ket{0^\prime}$, so the corresponding measurement is a of the displaced-null type.

We measure each of the qudits in the basis $\tilde{\mathcal{B}}_n$ and obtain i.i.d. outcomes $X_1,\dots,X_{n}$ taking values in $\{0,\dots, d^2-1\}$, and let $p_{{\bm u}}^{(n)}$ be their distribution:
%
\[
p^{(n)}_{{\bm u}}(j) =
|\langle 
\psi_{{\bm u}/\sqrt{n}}\otimes 0^\prime| v^{\delta_n}_j\rangle|^2,\quad  j=0,\dots, d^2-1.
\]
 
The following Theorem is one of the main results of the paper and shows that the Holevo bound can be attained by using displaced-null measurements.

\begin{theorem} \label{thm:dnmgeneral}
Assume we are given $n$ samples of the qudit state $|\psi_{\bm{\theta}}\rangle$ where $\bm{\theta}\in\Theta\subset \mathbb{R}^m$ is unknown. We further assume that assume that $\Theta$ is bounded and $\epsilon<1/10$. Using $\tilde{n} =n^{1-\epsilon}$ samples, we compute a preliminary estimator $\tilde{\bm{\theta}}_n$, and we measure the rest of the systems in the ONB $\tilde{{\cal B}}_n$, as defined above.  Let 
$$
\hat{\bm \theta}_n:= \tilde{\bm \theta}_n+\hat{\bm u}_n/\sqrt{n}
$$
be the estimator with
\[
\hat{u}_n^j=
\sum_{k=1}^{m} T_{jk} \left (\frac{n^{3\epsilon}}{\sqrt{2}} -\frac{n^{1-3\epsilon}}{\sqrt{2}} \hat{p}_n(k) \right ), \quad j=1,\dots,m
\]
where $\hat{p}_n(j)$ is the empirical estimator of $p^{(n)}_{{\bm u}}(j)$, i.e.
\[
\hat{p}_n(j)= \frac{| \{ i: X_i=j ,~i=1, \dots ,n\}|}{n},
\]
for $j=1,\dots,m$.

Then $\hat{\bm \theta}_n$ is asymptotically optimal in the sense that for every ${\bm \theta} \in \Theta$
$$
\lim_{n\to \infty} nR_n(\hat{\bm \theta}_n,{\bm \theta}) = {\cal H}^{W({\bm \theta})}({\bm \theta})
$$
Moreover, 
$\sqrt{n}(\hat{\bm \theta}_n-{\bm \theta})$ converges in law to a centered normal random variable with covariance given by $TT^{T}/2$.

\end{theorem}

The proof of Theorem \ref{thm:dnmgeneral} can be found in Appendix \ref{sec:lawsrb}.

Our measurement has been obtained by 
modifying the optimal linear measurement for the limiting Gaussian shift to displaced counting one, and translating this to a qudit and ancilla measurement with repect to a displaced-null basis.  Interestingly, this resulting measurement is closely connected to the optimal measurement described in \cite{Ma02}. The connection is discussed in Appendix \ref{app:matsu}.

\subsection{Estimating a completely unknown pure state with respect to the Bures distance}
\label{sec:completely.unknown.qudit}

In this section we consider the problem of estimating a completely unknown pure qudit state, when the loss function (figure of merit) is defined as the squared Bures distance
$$
d^2_b (|\psi\rangle\langle \psi| , |\phi\rangle\langle \phi|) = 2(1- |\langle \psi|\phi\rangle|).
$$
In this particular case, we will show that one can asymptotically achieve the Holevo bound using diplaced-null measurement without the need of using any ancillary system.

We parametrise a neighbourhood of the preliminary estimator $\ket{0}:=|\tilde{\psi}_n\rangle$ as
$$
|\psi_{{\bm u}/\sqrt{n}}\rangle = 
\exp\left(-i\sum_{k=1}^{d-1} ( u_1^k \sigma^k_{y} -u_2^k \sigma^k_{x} )/\sqrt{n}
\right) |0\rangle
$$
where ${\bm u} = (u_1^1, u_2^1, \dots, u_{1}^{d-1}, u_2^{d-1})\in \mathbb{R}^{2(d-1)}$ satisfies $\|{\bm u}\|\leq n^{\epsilon}$ with high probability.

For small deviations from $|0\rangle$ the Bures distance has the quadratic approximation
$$
d^2_b 
\left(
|\psi_{\frac{\bm u}{\sqrt{n}}}
\rangle\langle 
\psi_{\frac{\bm u}{\sqrt{n}}}|, 
|\psi_{\frac{{\bm u}^\prime}{\sqrt{n}}}
\rangle\langle 
\psi_{\frac{{\bm u}^\prime}{\sqrt{n}}}|
\right) 
= \frac{1}{n}\|{\bm u} - {\bm u}^\prime\|^2 + o(n^{-1+2\epsilon})
$$
which determines the optimal measurement and error rate in the asymptotic regime.

The Gaussian approximation consists in the model $\ket{{\bm u}_1+i {\bm u}_2}$ and the optimal measurement with respect to the identity cost matrix would be to measure the $Q_k$'s and $P_k$'s. In order to estimate ${\bm u}$, instead of usign an acilla, we split the ensemble of $n$ qudits in two equal sub-ensembles and perform separate 
`displaced-null' measurements on each of them in the following bases which are obtained by rotating $\{|0\rangle,\dots |d-1\rangle\}$ by (small) angles of size $\delta_n=n^{-1/2+3\epsilon}$
\begin{eqnarray}
|v^{\delta_n}_j\rangle &= &
U_1(\delta_n)|j\rangle = 
\exp\left(-i\delta_n \sum_{k=1}^{d-1} \sigma^k_y\right) |j\rangle 
\label{eq:U_1}\\
|w^{\delta_n}_j\rangle &= &
U_2(\delta_n)|j\rangle = 
\exp\left(i\delta_n \sum_{k=1}^{d-1} \sigma^k_x\right) |j\rangle.
\label{eq:U_2}
\end{eqnarray}
Therefore in the asymptotic picture, the proposed measurements are effectively joint measurements of 
$\{Q_i, i=1,\dots d-1\}$ and respectively $\{P_i, i=1,\dots d-1\}$ which are known to be optimal measurements for the local parameter ${\bm u}$ in the Gaussian shift model when performed on two separate copies of  
$|({\bm u}_1  +i {\bm u}_2)/\sqrt{2}\rangle$ obtained from the original state by using a beamsplitter.

Let $X_1,\dots,X_{n/2}$ and $Y_1,\dots,Y_{n/2}$ be the independent outcomes of the two types of measurements, taking values in $\{0,\dots, d-1\}$, and let $p_{{\bm u}}^{(n)}$ and $q_{{\bm u}}^{(n)}$ be their respective distributions 
\begin{equation} \label{eq:laws}
p^{(n)}_{{\bm u}}(j) =
|\langle 
\psi_{{\bm u}/\sqrt{n}}| v^{\delta_n}_j\rangle|^2,\quad
q^{(n)}_{{\bm u}}(j) =
|\langle 
\psi_{{\bm u}/\sqrt{n}}| w^{\delta_n}_j\rangle|^2.
\end{equation}

\begin{proposition} \label{prop:optlqd}
Assume $\epsilon<1/10$ and let 
$$
|\hat{\psi}\rangle:= |\psi_{\hat{\bm u}/\sqrt{n}} \rangle
$$
be the state estimator with local parameter  
$\hat{\bm u}_n$ defined as 
\begin{eqnarray*} 
\hat{u}^j_1&=&
\frac{n^{3\epsilon}}{2} -\frac{n^{1-3\epsilon}}{2} \hat{p}_n(j),\\ 
\hat{u}^j_2&=&
\frac{n^{3\epsilon}}{2} -\frac{n^{1-3\epsilon}}{2} \hat{q}_n(j), \quad j=1,\dots,{d-1},
\end{eqnarray*}
where $\hat{p}_n$, $\hat{q}_n$ are the empirical estimator of $p^{(n)}_{{\bm u}}$ and $q^{(n)}_{{\bm u}}$, respectively, i.e.
\begin{eqnarray*}
\hat{p}_n(j)&=& \frac{| \{ i: X_i=j ,~i=1, \dots ,n/2\}|}{n/2},\\
\hat{q}_n(j)&=& \frac{| \{ i: Y_i=j ,~i=1, \dots ,n/2\}|}{n/2},
\end{eqnarray*}
for $j=1,\dots,d-1$.

Then under $\mathbb{P}_{\bm u}$, $\sqrt{n}(\hat{\bm{u}}_n-{\bm u})$ is asymptotically distributed as a centered Gaussian random vector with covariance $\bm{1}/2$ and $|\hat{\psi}_n\rangle $ is asymptotically optimal in the sense that it achieves the Holevo bound:
$$
\lim_{n\to \infty} n\mathbb{E} _{\ket{\psi}}[
d_b^2(|\psi\rangle\langle \psi|, 
|\hat{\psi}_n\rangle\langle \hat{\psi}_n| ) ] = d-1.
$$

\end{proposition}
The proof of Proposition \ref{prop:optlqd} can be found in see Appendix \ref{sec:multidimproof}.

\subsection{Achieving the QCRB with displaced-null measurements}
\label{sec:achievingQCRB-displacednull}

We now consider quantum statistical models for which the QCRB is (asymptotically) achievable. In contrast to models discussed in sections \ref{subsec:HOdispl} and \ref{sec:completely.unknown.qudit}, in this case all parameter components can be estimated simultaneously at maximum precision. We will provide a class a displaced-null measurements which achieve the QCRB asymptotically.

Let us consider the statistical model $\{\ket{\psi_{\bm{\theta}}}\}$, $\bm{\theta}\in \Theta \subset \mathbb{R}^m$ with $m \leq 2 (d-1)$ and assume that the parameter is identifiable and that the QCRB is achievable for all $\theta\in \Theta$. This is equivalent to condition \eqref{eq:pureach} for all ${\bm \theta}\in \Theta$. The QFI is given by 
$$
F({\bm \theta})_{ij} = 4\langle 
\partial_i \psi_{\bm \theta} |
\partial_j \psi_{\bm \theta}
\rangle - 4 \langle \psi_{\bm \theta} | \partial_j \psi_{\bm \theta}\rangle 
\langle \partial_i \psi_{\bm \theta} |
\psi_{\bm \theta} \rangle,
$$
for $i,j=1,\dots, m$. Let $\ket{0}: = \ket{\psi_{\tilde{\bm \theta}_n}}$ be the preliminary estimator. We write ${\bm \theta} = \tilde{\bm \theta}_n+{\bm u}/\sqrt{n}$ with ${\bm u}$ the local parameter satisfying $\|{\bm u}\|\leq n^{\epsilon}$ with high probability. We assume that the phase of $|\psi_\theta\rangle$ has been chosen such that $\langle \dot{\psi}_i |0\rangle =0$ for all $i$, and denote $\dot{\psi}_i:=  \partial_i \psi_{\tilde{\bm \theta}_n}$.

We now describe a class of measurements that will be shown to achieve the QCRB asymptotically. We choose an orthonormal basis 
$\mathcal{B} := \{|0\rangle, |1\rangle , \dots, |d-1\rangle\}$ whose first vector is $|0\rangle$ and the other vectors satisfy 
\begin{equation}
\label{eq:real.inner.product}
c_{ki}:=\langle k |\dot{\psi}_i\rangle \in \mathbb{R}, 
\qquad i=1,\dots, m, \quad k=1,\dots , d-1.
\end{equation}
This condition is similar to equation (7) in 
\cite{NullQFI2}, but unlike this reference we do not impose additional conditions for the case when 
$\langle k |\dot{\psi}_i\rangle =0$ for all $i=1,\dots, m$. If we assume that the parameter $\bm{\theta}$ is identifiable, then the matrix $C=(c_{ki})$ needs to have rank $m$.

We will further rotate $\mathcal{B}$ with a unitary $U= \exp(-i\delta_n G)$ where $\delta_n=n^{-1/2+3\epsilon}$ and 
$$
G= \sum_{k=1}^{d-1} g_k\sigma^k_y, \quad \sigma_y^k=-i\ket{0}\bra{k} + i \ket{k} \bra{0}
$$
where $g_k\neq 0$ are arbitrary real coefficients. We obtain the ONB $\{|v_0^{\delta_n}\rangle, \dots |v_{d-1}^{\delta_n}\rangle \}$ with
$$
|v_k^{\delta_n}\rangle = U|k\rangle, \qquad k=0,\dots d-1.
$$
We measure all the systems in the basis $\tilde{\mathcal{B}}$ and obtain i.i.d. outcomes 
$X_1, \dots,  X_n\in \{0,\dots , d-1\}$ and denote by $\hat{p}_n$ the corresponding empirical frequency. We denote by $T=(T_{ij})$ the $m\times (d-1)$ matrix defined as
\[
T=(C^T C)^{-1}C^T.
\]
\begin{proposition}
\label{th:QCRB-achievability-null}
Assume that $\Theta$ be bounded and $\epsilon<1/10$. Let $\hat{\bm \theta}_n = \tilde{\bm \theta}_n+\hat{\bm u}_n/\sqrt{n}$ be the estimator determined by 
$$
\hat{\bm u}^j_n = \sum_{k=1}^{d-1} T_{jk}\left (\frac{g_kn^{3\epsilon}}{2} -\frac{n^{1-3\epsilon}}{2g_k} \hat{p}_n(k)\right ).
$$
Then $\hat{\bm \theta}_n$ achieves the QCRB, i.e. 
$$
\lim_{n\to\infty} 
n 
\mathbb{E}_{\bm \theta} [(\hat{\bm \theta}_n -{\bm \theta}) 
(\hat{\bm \theta}_n -{\bm \theta})^T ]
= F({\bm \theta})^{-1}.
$$
\end{proposition}

The proof of Proposition \ref{th:QCRB-achievability-null} can be found in Appendix \ref{app:proofQCRBac}.

We now give a QLAN interpretation of the above construction. The fact that $c_{ki}$ are real implies that the linearisation of the model around the preliminary estimation is given by
$$
|\tilde{\psi}_{{\bm u}/\sqrt{n}}\rangle = 
\exp
\left(-i\sum_{j=1}^m u_j S_j/\sqrt{n}
\right)|0\rangle
$$
with 
$$
S_j = \sum_{k=1}^{d-1} c_{kj}\sigma^k_y, 
\qquad 
c_{kj} = \langle k |\dot{\psi}_j\rangle. 
$$
By QLAN, the corresponding Gaussian model consists of coherent states $|C{\bm u}\rangle$ of a 
$(d-1)$-modes cv system where $C: \mathbb{R}^m\to \mathbb{C}^{d-1}$ is given by the \emph{real} coefficients $c_{kj} =\langle k|\dot{\psi}_j\rangle$. This means that each of the $(d-1)$ modes is in a coherent state whose displacement is along the $Q$ axis, so 
$ \langle C{\bm u} | P_k|C{\bm u}\rangle =0$ for all $k$, while 
$$ 
q_k := \langle C{\bm u} | Q_k|C{\bm u}\rangle = 
\sqrt{2}
\sum_{j=1}^m c_{kj} u_j.
$$
As we mentioned in Section \ref{sec:Holevo.Gaussian.shift}, the QCRB is achievable for the limit model too and the simultaneous measurement of all $Q_k$ is optimal. This is asymptotically obtained by the counting in the rotated basis.

\section{Conclusions and outlook}
In this paper we showed that the framework of displaced-null 
measurements provides a general scheme for optimal estimation of unknown parameters ${\bm \theta}\in \mathbb{R}^m$ of pure states models $|\psi_{\bm \theta}\rangle\in \mathbb{C}^d$. In particular, displaced-null measurements achieve the quantum Cram\'{e}r-Rao bound (QCRB) for models in which the bound is achievable, and the Holevo bound for general qudit models.

Our method is related to previous works \cite{NullQFI1,NullQFI2,NullQFI3} that deal with the achivebility of the QCRB for pure state models $|\psi_{\bm \theta}\rangle$. These works exhibit a class of parameter-dependent orthonormal bases $\mathcal{B}(\tilde{\bm \theta})$ whose associated classical Fisher information $I_{\tilde{\bm \theta}}({\bm \theta})$ converges to the quantum Fisher information $F({\bm \theta})$ of $|\psi_{\bm \theta}\rangle$ as 
$\tilde{\bm \theta}$ approaches the true unknown state parameter ${\bm \theta}$. The measurement basis $\mathcal{B}(\tilde{\bm \theta})$ has the 
special feature that it contains the state $|\psi_{\tilde{\bm \theta}}\rangle$ as one of its elements,  so that at $\tilde{\bm \theta}={\bm \theta}$ the measurement has only one outcome, while for $\tilde{\bm \theta}\approx{\bm \theta}$ the occurrence of other outcomes can be interpreted as signaling the deviation from the reference value $\tilde{\bm \theta}$. With this in mind we called such measurements, null measurements. 

However, the references \cite{NullQFI1,NullQFI2,NullQFI3} do not provide an explicit operational implementation of a strategy that achieves the QCRB. The naive solution would be to choose the reference parameter as a preliminary estimator $\tilde{\bm \theta}_n$ obtained by measuring a sub-sample of $\tilde{n}\ll n$ systems, and to apply the approximate null measurement $\mathcal{B}_{\tilde{\bm \theta}_n}$ to the rest of the systems. Surprisingly, it turned out that this adaptive strategy fails to achieve the QCRB, and indeed does not even reach the standard $n^{-1}$ scaling of precision, when the preliminary estimator satisfies certain natural assumptions. This is due to the fact that $\tilde{\bm \theta}_n $ lies in the interior of a confidence interval of ${\bm \theta}$ and the measurement cannot distinguish positive and negative deviations from the 
reference since probabilities depend on the square of the deviations. This is an important finding which shows the pitfalls of drawing statistical conclusions based solely on Fisher information arguments. 

To avoid this issue, we proposed to displace the preliminary estimator by a small amount ${\bm \delta}_n$ which is however sufficiently large to ensure that the new reference parameter $\tilde{\bm \theta}_n+ {\bm \delta}_n$ is outside the confidence interval of ${\bm \theta}$. Building on this idea we showed the achievability of the QCRB in the setting of \cite{NullQFI1,NullQFI2,NullQFI3}. Furthermore, for general pure state models and locally quadratic loss functions, we devised displaced-null measurements which achieve the Holevo bound asymptotically for arbitrary qudit models.  

The theory of quantum local asymptotic normality (QLAN) has played an important role in our investigations. The QLAN machinery translates the multi-copy estimation problem into one about estimating the mean of a multi-mode coherent state. In the latter case, counting measurements are paradigmatic example of null-measurements, while appropriately displacing the number operators provides the basis for displaced-null measurements. Using the QLAN correspondence, this translates into a simple prescription for rotating a basis containing the preliminary estimator $|\psi_{\tilde{\bm \theta}_n}
\rangle$ into that of the displaced-null measurement. Interestingly, the obtained measurement turned out to be closely related to the parameter-dependent measurements proposed by Matsumoto in \cite{Ma02}, and our approach offers an alternative asymptotic perspective on this work.

An exciting area of applications for displaced-null measurements is that of optimal estimation of dynamical parameters of open systems \cite{GW01,Guta2011,Molmer14,Guta_2015,GutaCB15,Guta_2017,Ilias22,Fallani22}. 
Recent works \cite{Godley2023,DayouCounting} have shown out that quantum post-processing by means of coherent absorbers allows for optimal estimation of such parameters. In particular \cite{DayouCounting} pointed out that a basic measurement such as photon counting constitutes a null-measurement, thus opening the route for devising optimal measurements for multidimensional estimation of Markov dynamics. An asymptotic analysis of displaced null measurements in this context will be the subject of a forthcoming publication \cite{GGG}.

Another area of  future interest is 
to extend the method to models consisting of mixed states. While this will probably not work in general, the ideas presented here may be useful for models consisting of states with a high degree of purity which is the relevant setup in many quantum technology applications. Another important extension is towards refining the methodology for optimal estimation in the finite sample rather than asymptotic regime. Finally, we would like to better understand how displaced-null measurements can be used in the context of quantum metrology and interferometry \cite{SCHNABEL17,BachorRalph}.

\vspace{4mm}
\noindent
{\bf Acknowledgements:} This work was supported by the EPSRC grant EP/T022140/1. We acknoweledge fruitful discussions with Dayou Yang, Rafa\l{} Demkowicz-Dobrza\'{n}ski, Janek Ko\l{}ody\'{n}ski and Richard Gill.

\bibliographystyle{apsrev4-1}
\bibliography{biblio.bib}

\begin{thebibliography}{111}%
\makeatletter
\providecommand \@ifxundefined [1]{%
 \@ifx{#1\undefined}
}%
\providecommand \@ifnum [1]{%
 \ifnum #1\expandafter \@firstoftwo
 \else \expandafter \@secondoftwo
 \fi
}%
\providecommand \@ifx [1]{%
 \ifx #1\expandafter \@firstoftwo
 \else \expandafter \@secondoftwo
 \fi
}%
\providecommand \natexlab [1]{#1}%
\providecommand \enquote  [1]{``#1''}%
\providecommand \bibnamefont  [1]{#1}%
\providecommand \bibfnamefont [1]{#1}%
\providecommand \citenamefont [1]{#1}%
\providecommand \href@noop [0]{\@secondoftwo}%
\providecommand \href [0]{\begingroup \@sanitize@url \@href}%
\providecommand \@href[1]{\@@startlink{#1}\@@href}%
\providecommand \@@href[1]{\endgroup#1\@@endlink}%
\providecommand \@sanitize@url [0]{\catcode `\\12\catcode `\$12\catcode `\&12\catcode `\#12\catcode `\^12\catcode `\_12\catcode `\%12\relax}%
\providecommand \@@startlink[1]{}%
\providecommand \@@endlink[0]{}%
\providecommand \url  [0]{\begingroup\@sanitize@url \@url }%
\providecommand \@url [1]{\endgroup\@href {#1}{\urlprefix }}%
\providecommand \urlprefix  [0]{URL }%
\providecommand \Eprint [0]{\href }%
\providecommand \doibase [0]{http://dx.doi.org/}%
\providecommand \selectlanguage [0]{\@gobble}%
\providecommand \bibinfo  [0]{\@secondoftwo}%
\providecommand \bibfield  [0]{\@secondoftwo}%
\providecommand \translation [1]{[#1]}%
\providecommand \BibitemOpen [0]{}%
\providecommand \bibitemStop [0]{}%
\providecommand \bibitemNoStop [0]{.\EOS\space}%
\providecommand \EOS [0]{\spacefactor3000\relax}%
\providecommand \BibitemShut  [1]{\csname bibitem#1\endcsname}%
\let\auto@bib@innerbib\@empty
\bibitem [{\citenamefont {Hayashi}(2005)}]{Hayashi2005}%
  \BibitemOpen
  \bibfield  {author} {\bibinfo {author} {\bibfnamefont {M.}~\bibnamefont {Hayashi}},\ }\href {\doibase 10.1142/5630} {\emph {\bibinfo {title} {Asymptotic theory of quantum statistical inference: Selected papers}}}\ (\bibinfo  {publisher} {World Scientific Publishing Co.},\ \bibinfo {year} {2005})\BibitemShut {NoStop}%
\bibitem [{\citenamefont {Paris}(2008)}]{Paris2008}%
  \BibitemOpen
  \bibfield  {author} {\bibinfo {author} {\bibfnamefont {M.~G.}\ \bibnamefont {Paris}},\ }\href {\doibase 10.1142/S0219749909004839} {\bibfield  {journal} {\bibinfo  {journal} {International Journal of Quantum Information}\ }\textbf {\bibinfo {volume} {7}},\ \bibinfo {pages} {125} (\bibinfo {year} {2008})}\BibitemShut {NoStop}%
\bibitem [{\citenamefont {Tóth}\ and\ \citenamefont {Apellaniz}(2014)}]{TothReview}%
  \BibitemOpen
  \bibfield  {author} {\bibinfo {author} {\bibfnamefont {G.}~\bibnamefont {Tóth}}\ and\ \bibinfo {author} {\bibfnamefont {I.}~\bibnamefont {Apellaniz}},\ }\href {\doibase 10.1088/1751-8113/47/42/424006} {\bibfield  {journal} {\bibinfo  {journal} {Journal of Physics A: Mathematical and Theoretical}\ }\textbf {\bibinfo {volume} {47}},\ \bibinfo {pages} {424006} (\bibinfo {year} {2014})}\BibitemShut {NoStop}%
\bibitem [{\citenamefont {Demkowicz-Dobrzański}\ \emph {et~al.}(2020)\citenamefont {Demkowicz-Dobrzański}, \citenamefont {Gorecki},\ and\ \citenamefont {Gu\c{t}\u{a}}}]{RafalReview}%
  \BibitemOpen
  \bibfield  {author} {\bibinfo {author} {\bibfnamefont {R.}~\bibnamefont {Demkowicz-Dobrzański}}, \bibinfo {author} {\bibfnamefont {W.}~\bibnamefont {Gorecki}}, \ and\ \bibinfo {author} {\bibfnamefont {M.}~\bibnamefont {Gu\c{t}\u{a}}},\ }\href {\doibase 10.1088/1751-8121/ab8ef3} {\bibfield  {journal} {\bibinfo  {journal} {Journal of Physics A: Mathematical and Theoretical}\ }\textbf {\bibinfo {volume} {53}},\ \bibinfo {pages} {363001} (\bibinfo {year} {2020})}\BibitemShut {NoStop}%
\bibitem [{\citenamefont {Banaszek}\ \emph {et~al.}(2013)\citenamefont {Banaszek}, \citenamefont {Cramer},\ and\ \citenamefont {Gross}}]{Tomo2}%
  \BibitemOpen
  \bibfield  {author} {\bibinfo {author} {\bibfnamefont {K.}~\bibnamefont {Banaszek}}, \bibinfo {author} {\bibfnamefont {M.}~\bibnamefont {Cramer}}, \ and\ \bibinfo {author} {\bibfnamefont {D.}~\bibnamefont {Gross}},\ }\href {\doibase 10.1088/1367-2630/15/12/125020} {\bibfield  {journal} {\bibinfo  {journal} {New Journal of Physics}\ }\textbf {\bibinfo {volume} {15}},\ \bibinfo {pages} {125020} (\bibinfo {year} {2013})}\BibitemShut {NoStop}%
\bibitem [{\citenamefont {Albarelli}\ \emph {et~al.}(2020)\citenamefont {Albarelli}, \citenamefont {Barbieri}, \citenamefont {Genoni},\ and\ \citenamefont {Gianani}}]{Albarelli2020}%
  \BibitemOpen
  \bibfield  {author} {\bibinfo {author} {\bibfnamefont {F.}~\bibnamefont {Albarelli}}, \bibinfo {author} {\bibfnamefont {M.}~\bibnamefont {Barbieri}}, \bibinfo {author} {\bibfnamefont {M.}~\bibnamefont {Genoni}}, \ and\ \bibinfo {author} {\bibfnamefont {I.}~\bibnamefont {Gianani}},\ }\href {\doibase 10.1016/j.physleta.2020.126311} {\bibfield  {journal} {\bibinfo  {journal} {Physics Letters A}\ }\textbf {\bibinfo {volume} {384}},\ \bibinfo {pages} {126311} (\bibinfo {year} {2020})}\BibitemShut {NoStop}%
\bibitem [{\citenamefont {Sidhu}\ and\ \citenamefont {Kok}(2020)}]{Sidhu_2020}%
  \BibitemOpen
  \bibfield  {author} {\bibinfo {author} {\bibfnamefont {J.~S.}\ \bibnamefont {Sidhu}}\ and\ \bibinfo {author} {\bibfnamefont {P.}~\bibnamefont {Kok}},\ }\href {\doibase 10.1116/1.5119961} {\bibfield  {journal} {\bibinfo  {journal} {{AVS} Quantum Science}\ }\textbf {\bibinfo {volume} {2}},\ \bibinfo {pages} {014701} (\bibinfo {year} {2020})}\BibitemShut {NoStop}%
\bibitem [{\citenamefont {Gross}\ \emph {et~al.}(2010)\citenamefont {Gross}, \citenamefont {Liu}, \citenamefont {Flammia}, \citenamefont {Becker},\ and\ \citenamefont {Eisert}}]{Gross10}%
  \BibitemOpen
  \bibfield  {author} {\bibinfo {author} {\bibfnamefont {D.}~\bibnamefont {Gross}}, \bibinfo {author} {\bibfnamefont {Y.-K.}\ \bibnamefont {Liu}}, \bibinfo {author} {\bibfnamefont {S.~T.}\ \bibnamefont {Flammia}}, \bibinfo {author} {\bibfnamefont {S.}~\bibnamefont {Becker}}, \ and\ \bibinfo {author} {\bibfnamefont {J.}~\bibnamefont {Eisert}},\ }\href {\doibase 10.1103/PhysRevLett.105.150401} {\bibfield  {journal} {\bibinfo  {journal} {Phys. Rev. Lett.}\ }\textbf {\bibinfo {volume} {105}},\ \bibinfo {pages} {150401} (\bibinfo {year} {2010})}\BibitemShut {NoStop}%
\bibitem [{\citenamefont {Cramer}\ \emph {et~al.}(2010)\citenamefont {Cramer}, \citenamefont {Plenio}, \citenamefont {Flammia}, \citenamefont {Somma}, \citenamefont {Gross}, \citenamefont {Bartlett}, \citenamefont {Landon-Cardinal}, \citenamefont {Poulin},\ and\ \citenamefont {Liu}}]{Cramer2010}%
  \BibitemOpen
  \bibfield  {author} {\bibinfo {author} {\bibfnamefont {M.}~\bibnamefont {Cramer}}, \bibinfo {author} {\bibfnamefont {M.~B.}\ \bibnamefont {Plenio}}, \bibinfo {author} {\bibfnamefont {S.~T.}\ \bibnamefont {Flammia}}, \bibinfo {author} {\bibfnamefont {R.}~\bibnamefont {Somma}}, \bibinfo {author} {\bibfnamefont {D.}~\bibnamefont {Gross}}, \bibinfo {author} {\bibfnamefont {S.~D.}\ \bibnamefont {Bartlett}}, \bibinfo {author} {\bibfnamefont {O.}~\bibnamefont {Landon-Cardinal}}, \bibinfo {author} {\bibfnamefont {D.}~\bibnamefont {Poulin}}, \ and\ \bibinfo {author} {\bibfnamefont {Y.-K.}\ \bibnamefont {Liu}},\ }\href {\doibase 10.1038/ncomms1147} {\bibfield  {journal} {\bibinfo  {journal} {Nature Communications}\ }\textbf {\bibinfo {volume} {1}},\ \bibinfo {pages} {149} (\bibinfo {year} {2010})}\BibitemShut {NoStop}%
\bibitem [{\citenamefont {O'Donnell}\ and\ \citenamefont {Wright}(2016)}]{Tomo1}%
  \BibitemOpen
  \bibfield  {author} {\bibinfo {author} {\bibfnamefont {R.}~\bibnamefont {O'Donnell}}\ and\ \bibinfo {author} {\bibfnamefont {J.}~\bibnamefont {Wright}},\ }in\ \href {\doibase 10.1145/2897518.2897544} {\emph {\bibinfo {booktitle} {Proceedings of the Forty-Eighth Annual ACM Symposium on Theory of Computing}}},\ \bibinfo {series and number} {STOC '16}\ (\bibinfo {year} {2016})\ p.\ \bibinfo {pages} {899}\BibitemShut {NoStop}%
\bibitem [{\citenamefont {Haah}\ \emph {et~al.}(2016)\citenamefont {Haah}, \citenamefont {Harrow}, \citenamefont {Ji}, \citenamefont {Wu},\ and\ \citenamefont {Yu}}]{Haah16}%
  \BibitemOpen
  \bibfield  {author} {\bibinfo {author} {\bibfnamefont {J.}~\bibnamefont {Haah}}, \bibinfo {author} {\bibfnamefont {A.~W.}\ \bibnamefont {Harrow}}, \bibinfo {author} {\bibfnamefont {Z.}~\bibnamefont {Ji}}, \bibinfo {author} {\bibfnamefont {X.}~\bibnamefont {Wu}}, \ and\ \bibinfo {author} {\bibfnamefont {N.}~\bibnamefont {Yu}},\ }in\ \href {\doibase 10.1145/2897518.2897585} {\emph {\bibinfo {booktitle} {Proceedings of the Forty-Eighth Annual ACM Symposium on Theory of Computing}}},\ \bibinfo {series and number} {STOC '16}\ (\bibinfo {year} {2016})\ p.\ \bibinfo {pages} {913}\BibitemShut {NoStop}%
\bibitem [{\citenamefont {Lanyon}\ \emph {et~al.}(2017)\citenamefont {Lanyon}, \citenamefont {Maier}, \citenamefont {Holz{\"a}pfel}, \citenamefont {Baumgratz}, \citenamefont {Hempel}, \citenamefont {Jurcevic}, \citenamefont {Dhand}, \citenamefont {Buyskikh}, \citenamefont {Daley}, \citenamefont {Cramer}, \citenamefont {Plenio}, \citenamefont {Blatt},\ and\ \citenamefont {Roos}}]{Lanyon2017}%
  \BibitemOpen
  \bibfield  {author} {\bibinfo {author} {\bibfnamefont {B.}~\bibnamefont {Lanyon}}, \bibinfo {author} {\bibfnamefont {C.}~\bibnamefont {Maier}}, \bibinfo {author} {\bibfnamefont {M.}~\bibnamefont {Holz{\"a}pfel}}, \bibinfo {author} {\bibfnamefont {T.}~\bibnamefont {Baumgratz}}, \bibinfo {author} {\bibfnamefont {C.}~\bibnamefont {Hempel}}, \bibinfo {author} {\bibfnamefont {P.}~\bibnamefont {Jurcevic}}, \bibinfo {author} {\bibfnamefont {I.}~\bibnamefont {Dhand}}, \bibinfo {author} {\bibfnamefont {A.}~\bibnamefont {Buyskikh}}, \bibinfo {author} {\bibfnamefont {A.}~\bibnamefont {Daley}}, \bibinfo {author} {\bibfnamefont {M.}~\bibnamefont {Cramer}}, \bibinfo {author} {\bibfnamefont {M.~B.}\ \bibnamefont {Plenio}}, \bibinfo {author} {\bibfnamefont {R.}~\bibnamefont {Blatt}}, \ and\ \bibinfo {author} {\bibfnamefont {C.~F.}\ \bibnamefont {Roos}},\ }\href {\doibase 10.1038/nphys4244} {\bibfield  {journal} {\bibinfo  {journal} {Nature Physics}\ }\textbf {\bibinfo {volume} {13}},\ \bibinfo {pages} {1158–1162}
  (\bibinfo {year} {2017})}\BibitemShut {NoStop}%
\bibitem [{\citenamefont {Gu\c{t}\u{a}}\ and\ \citenamefont {Kiukas}(2017)}]{Guta_2017}%
  \BibitemOpen
  \bibfield  {author} {\bibinfo {author} {\bibfnamefont {M.}~\bibnamefont {Gu\c{t}\u{a}}}\ and\ \bibinfo {author} {\bibfnamefont {J.}~\bibnamefont {Kiukas}},\ }\href {\doibase 10.1063/1.4982958} {\bibfield  {journal} {\bibinfo  {journal} {Journal of Mathematical Physics}\ }\textbf {\bibinfo {volume} {58}},\ \bibinfo {pages} {052201} (\bibinfo {year} {2017})}\BibitemShut {NoStop}%
\bibitem [{\citenamefont {Yang}\ \emph {et~al.}(2019)\citenamefont {Yang}, \citenamefont {Chiribella},\ and\ \citenamefont {Hayashi}}]{Yang19}%
  \BibitemOpen
  \bibfield  {author} {\bibinfo {author} {\bibfnamefont {Y.}~\bibnamefont {Yang}}, \bibinfo {author} {\bibfnamefont {G.}~\bibnamefont {Chiribella}}, \ and\ \bibinfo {author} {\bibfnamefont {M.}~\bibnamefont {Hayashi}},\ }\href {\doibase 10.1103/PhysRevLett.105.150401} {\bibfield  {journal} {\bibinfo  {journal} {Commun. Math. Phys.}\ }\textbf {\bibinfo {volume} {368}},\ \bibinfo {pages} {223–293} (\bibinfo {year} {2019})}\BibitemShut {NoStop}%
\bibitem [{\citenamefont {Lahiry}\ and\ \citenamefont {Nussbaum}(2023)}]{lahiry2021}%
  \BibitemOpen
  \bibfield  {author} {\bibinfo {author} {\bibfnamefont {S.}~\bibnamefont {Lahiry}}\ and\ \bibinfo {author} {\bibfnamefont {M.}~\bibnamefont {Nussbaum}},\ }\href@noop {} {\enquote {\bibinfo {title} {Minimax estimation of low-rank quantum states and their linear functionals},}\ } (\bibinfo {year} {2023}),\ \Eprint {http://arxiv.org/abs/2111.03279} {arXiv:2111.03279 [math.ST]} \BibitemShut {NoStop}%
\bibitem [{\citenamefont {Yuen}(2023)}]{Yuen23}%
  \BibitemOpen
  \bibfield  {author} {\bibinfo {author} {\bibfnamefont {H.}~\bibnamefont {Yuen}},\ }\href {\doibase 10.22331/q-2023-01-03-890} {\bibfield  {journal} {\bibinfo  {journal} {Quantum}\ }\textbf {\bibinfo {volume} {7}},\ \bibinfo {pages} {890} (\bibinfo {year} {2023})}\BibitemShut {NoStop}%
\bibitem [{\citenamefont {Szczykulska}\ \emph {et~al.}(2016)\citenamefont {Szczykulska}, \citenamefont {Baumgratz},\ and\ \citenamefont {Datta}}]{Szczykulska16}%
  \BibitemOpen
  \bibfield  {author} {\bibinfo {author} {\bibfnamefont {M.}~\bibnamefont {Szczykulska}}, \bibinfo {author} {\bibfnamefont {T.}~\bibnamefont {Baumgratz}}, \ and\ \bibinfo {author} {\bibfnamefont {A.}~\bibnamefont {Datta}},\ }\href {\doibase 10.1080/23746149.2016.1230476} {\bibfield  {journal} {\bibinfo  {journal} {Advances in Physics: X}\ }\textbf {\bibinfo {volume} {1}},\ \bibinfo {pages} {621} (\bibinfo {year} {2016})}\BibitemShut {NoStop}%
\bibitem [{\citenamefont {Nichols}\ \emph {et~al.}(2018)\citenamefont {Nichols}, \citenamefont {Liuzzo-Scorpo}, \citenamefont {Knott},\ and\ \citenamefont {Adesso}}]{Nichols2018}%
  \BibitemOpen
  \bibfield  {author} {\bibinfo {author} {\bibfnamefont {R.}~\bibnamefont {Nichols}}, \bibinfo {author} {\bibfnamefont {P.}~\bibnamefont {Liuzzo-Scorpo}}, \bibinfo {author} {\bibfnamefont {P.~A.}\ \bibnamefont {Knott}}, \ and\ \bibinfo {author} {\bibfnamefont {G.}~\bibnamefont {Adesso}},\ }\href {\doibase 10.1103/PHYSREVA.98.012114} {\bibfield  {journal} {\bibinfo  {journal} {Physical Review A}\ }\textbf {\bibinfo {volume} {98}},\ \bibinfo {pages} {012114} (\bibinfo {year} {2018})}\BibitemShut {NoStop}%
\bibitem [{\citenamefont {Albarelli}\ \emph {et~al.}(2019)\citenamefont {Albarelli}, \citenamefont {Friel},\ and\ \citenamefont {Datta}}]{Albarelli19}%
  \BibitemOpen
  \bibfield  {author} {\bibinfo {author} {\bibfnamefont {F.}~\bibnamefont {Albarelli}}, \bibinfo {author} {\bibfnamefont {J.~F.}\ \bibnamefont {Friel}}, \ and\ \bibinfo {author} {\bibfnamefont {A.}~\bibnamefont {Datta}},\ }\href {\doibase 10.1103/PhysRevLett.123.200503} {\bibfield  {journal} {\bibinfo  {journal} {Phys. Rev. Lett.}\ }\textbf {\bibinfo {volume} {123}},\ \bibinfo {pages} {200503} (\bibinfo {year} {2019})}\BibitemShut {NoStop}%
\bibitem [{\citenamefont {Liu}\ \emph {et~al.}(2020)\citenamefont {Liu}, \citenamefont {Yuan}, \citenamefont {Lu},\ and\ \citenamefont {Wang}}]{Liu20}%
  \BibitemOpen
  \bibfield  {author} {\bibinfo {author} {\bibfnamefont {J.}~\bibnamefont {Liu}}, \bibinfo {author} {\bibfnamefont {H.}~\bibnamefont {Yuan}}, \bibinfo {author} {\bibfnamefont {X.-M.}\ \bibnamefont {Lu}}, \ and\ \bibinfo {author} {\bibfnamefont {X.}~\bibnamefont {Wang}},\ }\href {\doibase 10.1088/1751-8121/ab5d4d} {\bibfield  {journal} {\bibinfo  {journal} {J. Phys. A: Math. Theor.}\ }\textbf {\bibinfo {volume} {53}},\ \bibinfo {pages} {023001} (\bibinfo {year} {2020})}\BibitemShut {NoStop}%
\bibitem [{\citenamefont {Mosonyi}\ and\ \citenamefont {Petz}(2003)}]{Mosonyi2003}%
  \BibitemOpen
  \bibfield  {author} {\bibinfo {author} {\bibfnamefont {M.}~\bibnamefont {Mosonyi}}\ and\ \bibinfo {author} {\bibfnamefont {D.}~\bibnamefont {Petz}},\ }\href {\doibase 10.1007/s11005-004-4072-2} {\bibfield  {journal} {\bibinfo  {journal} {Letters in Mathematical Physics}\ }\textbf {\bibinfo {volume} {68}},\ \bibinfo {pages} {19} (\bibinfo {year} {2003})}\BibitemShut {NoStop}%
\bibitem [{\citenamefont {Jenčová}\ and\ \citenamefont {Petz}(2006)}]{Jencova2006}%
  \BibitemOpen
  \bibfield  {author} {\bibinfo {author} {\bibfnamefont {A.}~\bibnamefont {Jenčová}}\ and\ \bibinfo {author} {\bibfnamefont {D.}~\bibnamefont {Petz}},\ }\href {\doibase 10.1007/S00220-005-1510-7} {\bibfield  {journal} {\bibinfo  {journal} {Communications in Mathematical Physics}\ }\textbf {\bibinfo {volume} {263}},\ \bibinfo {pages} {259} (\bibinfo {year} {2006})}\BibitemShut {NoStop}%
\bibitem [{\citenamefont {Gu\c{t}\u{a}}\ and\ \citenamefont {Kahn}(2006)}]{LAN1}%
  \BibitemOpen
  \bibfield  {author} {\bibinfo {author} {\bibfnamefont {M.}~\bibnamefont {Gu\c{t}\u{a}}}\ and\ \bibinfo {author} {\bibfnamefont {J.}~\bibnamefont {Kahn}},\ }\href {\doibase 10.1103/PhysRevA.73.052108} {\bibfield  {journal} {\bibinfo  {journal} {Physical Review A}\ }\textbf {\bibinfo {volume} {73}},\ \bibinfo {pages} {052108} (\bibinfo {year} {2006})}\BibitemShut {NoStop}%
\bibitem [{\citenamefont {Gu\c{t}\u{a}}\ and\ \citenamefont {Jencova}(2007)}]{LAN2}%
  \BibitemOpen
  \bibfield  {author} {\bibinfo {author} {\bibfnamefont {M.}~\bibnamefont {Gu\c{t}\u{a}}}\ and\ \bibinfo {author} {\bibfnamefont {A.}~\bibnamefont {Jencova}},\ }\href {\doibase 10.1007/s00220-007-0340-1} {\bibfield  {journal} {\bibinfo  {journal} {Communications in Mathematical Physics}\ }\textbf {\bibinfo {volume} {276}},\ \bibinfo {pages} {341} (\bibinfo {year} {2007})}\BibitemShut {NoStop}%
\bibitem [{\citenamefont {Kahn}\ and\ \citenamefont {Gu\c{t}\u{a}}(2009)}]{LAN3}%
  \BibitemOpen
  \bibfield  {author} {\bibinfo {author} {\bibfnamefont {J.}~\bibnamefont {Kahn}}\ and\ \bibinfo {author} {\bibfnamefont {M.}~\bibnamefont {Gu\c{t}\u{a}}},\ }\href {\doibase 10.1007/s00220-009-0787-3} {\bibfield  {journal} {\bibinfo  {journal} {Communications in Mathematical Physics}\ }\textbf {\bibinfo {volume} {289}},\ \bibinfo {pages} {597} (\bibinfo {year} {2009})}\BibitemShut {NoStop}%
\bibitem [{\citenamefont {Gill}\ and\ \citenamefont {Gu\c{t}\u{a}}(2013)}]{LAN4}%
  \BibitemOpen
  \bibfield  {author} {\bibinfo {author} {\bibfnamefont {R.~D.}\ \bibnamefont {Gill}}\ and\ \bibinfo {author} {\bibfnamefont {M.}~\bibnamefont {Gu\c{t}\u{a}}},\ }\href {\doibase 10.1214/12-IMSCOLL909} {\bibfield  {journal} {\bibinfo  {journal} {IMS Collections}\ }\textbf {\bibinfo {volume} {9}},\ \bibinfo {pages} {105} (\bibinfo {year} {2013})}\BibitemShut {NoStop}%
\bibitem [{\citenamefont {Butucea}\ \emph {et~al.}(2016)\citenamefont {Butucea}, \citenamefont {Gu\c{t}\u{a}},\ and\ \citenamefont {Nussbaum}}]{LAN6}%
  \BibitemOpen
  \bibfield  {author} {\bibinfo {author} {\bibfnamefont {C.}~\bibnamefont {Butucea}}, \bibinfo {author} {\bibfnamefont {M.}~\bibnamefont {Gu\c{t}\u{a}}}, \ and\ \bibinfo {author} {\bibfnamefont {M.}~\bibnamefont {Nussbaum}},\ }\href {\doibase https://doi.org/10.1214/17-AOS1672} {\bibfield  {journal} {\bibinfo  {journal} {Annals Statist.}\ }\textbf {\bibinfo {volume} {46}},\ \bibinfo {pages} {3676} (\bibinfo {year} {2016})}\BibitemShut {NoStop}%
\bibitem [{\citenamefont {Yamagata}\ \emph {et~al.}(2013)\citenamefont {Yamagata}, \citenamefont {Fujiwara},\ and\ \citenamefont {Gill}}]{Yamagata13}%
  \BibitemOpen
  \bibfield  {author} {\bibinfo {author} {\bibfnamefont {K.}~\bibnamefont {Yamagata}}, \bibinfo {author} {\bibfnamefont {A.}~\bibnamefont {Fujiwara}}, \ and\ \bibinfo {author} {\bibfnamefont {R.~D.}\ \bibnamefont {Gill}},\ }\href {\doibase 10.1214/13-AOS1147} {\bibfield  {journal} {\bibinfo  {journal} {The Annals of Statistics}\ }\textbf {\bibinfo {volume} {41}},\ \bibinfo {pages} {2197} (\bibinfo {year} {2013})}\BibitemShut {NoStop}%
\bibitem [{\citenamefont {Fujiwara}\ and\ \citenamefont {Yamagata}(2020)}]{Fujiwara20}%
  \BibitemOpen
  \bibfield  {author} {\bibinfo {author} {\bibfnamefont {A.}~\bibnamefont {Fujiwara}}\ and\ \bibinfo {author} {\bibfnamefont {K.}~\bibnamefont {Yamagata}},\ }\href {\doibase 10.3150/19-BEJ1185} {\bibfield  {journal} {\bibinfo  {journal} {Bernoulli}\ }\textbf {\bibinfo {volume} {26}},\ \bibinfo {pages} {2105 } (\bibinfo {year} {2020})}\BibitemShut {NoStop}%
\bibitem [{\citenamefont {Fujiwara}\ and\ \citenamefont {Yamagata}(2023)}]{Fujiwara22}%
  \BibitemOpen
  \bibfield  {author} {\bibinfo {author} {\bibfnamefont {A.}~\bibnamefont {Fujiwara}}\ and\ \bibinfo {author} {\bibfnamefont {K.}~\bibnamefont {Yamagata}},\ }\href {\doibase 10.1214/23-aos2285} {\bibfield  {journal} {\bibinfo  {journal} {Ann. Statist.}\ }\textbf {\bibinfo {volume} {51}},\ \bibinfo {pages} {1159} (\bibinfo {year} {2023})}\BibitemShut {NoStop}%
\bibitem [{\citenamefont {Aaronson}(2018)}]{Aaronson2018}%
  \BibitemOpen
  \bibfield  {author} {\bibinfo {author} {\bibfnamefont {S.}~\bibnamefont {Aaronson}},\ }in\ \href {\doibase 10.1145/3188745.3188802} {\emph {\bibinfo {booktitle} {Proceedings of the 50th Annual ACM SIGACT Symposium on Theory of Computing}}},\ \bibinfo {series and number} {STOC 2018}\ (\bibinfo {year} {2018})\ p.\ \bibinfo {pages} {325}\BibitemShut {NoStop}%
\bibitem [{\citenamefont {Huang}\ \emph {et~al.}(2020)\citenamefont {Huang}, \citenamefont {Kueng},\ and\ \citenamefont {Preskill}}]{Huang2020}%
  \BibitemOpen
  \bibfield  {author} {\bibinfo {author} {\bibfnamefont {H.~Y.}\ \bibnamefont {Huang}}, \bibinfo {author} {\bibfnamefont {R.}~\bibnamefont {Kueng}}, \ and\ \bibinfo {author} {\bibfnamefont {J.}~\bibnamefont {Preskill}},\ }\href {\doibase 10.1038/s41567-020-0932-7} {\bibfield  {journal} {\bibinfo  {journal} {Nature Physics}\ }\textbf {\bibinfo {volume} {16}},\ \bibinfo {pages} {1050} (\bibinfo {year} {2020})}\BibitemShut {NoStop}%
\bibitem [{\citenamefont {Personick}(1971)}]{Personick1971}%
  \BibitemOpen
  \bibfield  {author} {\bibinfo {author} {\bibfnamefont {S.~D.}\ \bibnamefont {Personick}},\ }\href {\doibase 10.1109/TIT.1971.1054643} {\bibfield  {journal} {\bibinfo  {journal} {IEEE Transactions on Information Theory}\ }\textbf {\bibinfo {volume} {17}},\ \bibinfo {pages} {240} (\bibinfo {year} {1971})}\BibitemShut {NoStop}%
\bibitem [{\citenamefont {Tsang}(2020)}]{Tsang20}%
  \BibitemOpen
  \bibfield  {author} {\bibinfo {author} {\bibfnamefont {M.}~\bibnamefont {Tsang}},\ }\href {\doibase 10.1103/PhysRevA.102.062217} {\bibfield  {journal} {\bibinfo  {journal} {Phys. Rev. A}\ }\textbf {\bibinfo {volume} {102}},\ \bibinfo {pages} {062217} (\bibinfo {year} {2020})}\BibitemShut {NoStop}%
\bibitem [{\citenamefont {Rubio}\ and\ \citenamefont {Dunningham}(2020)}]{Rubio2020}%
  \BibitemOpen
  \bibfield  {author} {\bibinfo {author} {\bibfnamefont {J.}~\bibnamefont {Rubio}}\ and\ \bibinfo {author} {\bibfnamefont {J.}~\bibnamefont {Dunningham}},\ }\href {\doibase 10.1103/PHYSREVA.101.032114} {\bibfield  {journal} {\bibinfo  {journal} {Physical Review A}\ }\textbf {\bibinfo {volume} {101}},\ \bibinfo {pages} {032114} (\bibinfo {year} {2020})}\BibitemShut {NoStop}%
\bibitem [{\citenamefont {Rubio}\ \emph {et~al.}(2021)\citenamefont {Rubio}, \citenamefont {Anders},\ and\ \citenamefont {Correa}}]{Rubio2021}%
  \BibitemOpen
  \bibfield  {author} {\bibinfo {author} {\bibfnamefont {J.}~\bibnamefont {Rubio}}, \bibinfo {author} {\bibfnamefont {J.}~\bibnamefont {Anders}}, \ and\ \bibinfo {author} {\bibfnamefont {L.~A.}\ \bibnamefont {Correa}},\ }\href {\doibase 10.1103/PHYSREVLETT.127.190402} {\bibfield  {journal} {\bibinfo  {journal} {Physical Review Letters}\ }\textbf {\bibinfo {volume} {127}},\ \bibinfo {pages} {190402} (\bibinfo {year} {2021})}\BibitemShut {NoStop}%
\bibitem [{\citenamefont {Giovannetti}\ \emph {et~al.}(2005)\citenamefont {Giovannetti}, \citenamefont {Lloyd},\ and\ \citenamefont {Maccone}}]{Metrology1}%
  \BibitemOpen
  \bibfield  {author} {\bibinfo {author} {\bibfnamefont {V.}~\bibnamefont {Giovannetti}}, \bibinfo {author} {\bibfnamefont {S.}~\bibnamefont {Lloyd}}, \ and\ \bibinfo {author} {\bibfnamefont {L.}~\bibnamefont {Maccone}},\ }\href {\doibase 10.1103/PhysRevLett.96.010401} {\bibfield  {journal} {\bibinfo  {journal} {Physical Review Letters}\ }\textbf {\bibinfo {volume} {96}},\ \bibinfo {pages} {010401} (\bibinfo {year} {2005})}\BibitemShut {NoStop}%
\bibitem [{\citenamefont {Fujiwara}\ and\ \citenamefont {Imai}(2008)}]{Fujiwara2008}%
  \BibitemOpen
  \bibfield  {author} {\bibinfo {author} {\bibfnamefont {A.}~\bibnamefont {Fujiwara}}\ and\ \bibinfo {author} {\bibfnamefont {H.}~\bibnamefont {Imai}},\ }\href {\doibase 10.1088/1751-8113/41/25/255304} {\bibfield  {journal} {\bibinfo  {journal} {Journal of Physics A: Mathematical and Theoretical}\ }\textbf {\bibinfo {volume} {41}},\ \bibinfo {pages} {255304} (\bibinfo {year} {2008})}\BibitemShut {NoStop}%
\bibitem [{\citenamefont {Giovannetti}\ \emph {et~al.}(2011)\citenamefont {Giovannetti}, \citenamefont {Lloyd},\ and\ \citenamefont {Maccone}}]{Giovannetti2011}%
  \BibitemOpen
  \bibfield  {author} {\bibinfo {author} {\bibfnamefont {V.}~\bibnamefont {Giovannetti}}, \bibinfo {author} {\bibfnamefont {S.}~\bibnamefont {Lloyd}}, \ and\ \bibinfo {author} {\bibfnamefont {L.}~\bibnamefont {Maccone}},\ }\href {\doibase 10.1038/nphoton.2011.35} {\bibfield  {journal} {\bibinfo  {journal} {Nature Photonics}\ }\textbf {\bibinfo {volume} {5}},\ \bibinfo {pages} {222} (\bibinfo {year} {2011})}\BibitemShut {NoStop}%
\bibitem [{\citenamefont {Demkowicz-Dobrzański}\ \emph {et~al.}(2012)\citenamefont {Demkowicz-Dobrzański}, \citenamefont {Kołodyński},\ and\ \citenamefont {Gu\c{t}\u{a}}}]{Metrology2}%
  \BibitemOpen
  \bibfield  {author} {\bibinfo {author} {\bibfnamefont {R.}~\bibnamefont {Demkowicz-Dobrzański}}, \bibinfo {author} {\bibfnamefont {J.}~\bibnamefont {Kołodyński}}, \ and\ \bibinfo {author} {\bibfnamefont {M.}~\bibnamefont {Gu\c{t}\u{a}}},\ }\href {\doibase 10.1038/ncomms2067} {\bibfield  {journal} {\bibinfo  {journal} {Nature Communications}\ }\textbf {\bibinfo {volume} {3}},\ \bibinfo {pages} {1063} (\bibinfo {year} {2012})}\BibitemShut {NoStop}%
\bibitem [{\citenamefont {Girolami}\ \emph {et~al.}(2014)\citenamefont {Girolami}, \citenamefont {Souza}, \citenamefont {Giovannetti}, \citenamefont {Tufarelli}, \citenamefont {Filgueiras}, \citenamefont {Sarthour}, \citenamefont {Soares-Pinto}, \citenamefont {Oliveira},\ and\ \citenamefont {Adesso}}]{Girolami14}%
  \BibitemOpen
  \bibfield  {author} {\bibinfo {author} {\bibfnamefont {D.}~\bibnamefont {Girolami}}, \bibinfo {author} {\bibfnamefont {A.~M.}\ \bibnamefont {Souza}}, \bibinfo {author} {\bibfnamefont {V.}~\bibnamefont {Giovannetti}}, \bibinfo {author} {\bibfnamefont {T.}~\bibnamefont {Tufarelli}}, \bibinfo {author} {\bibfnamefont {J.~G.}\ \bibnamefont {Filgueiras}}, \bibinfo {author} {\bibfnamefont {R.~S.}\ \bibnamefont {Sarthour}}, \bibinfo {author} {\bibfnamefont {D.~O.}\ \bibnamefont {Soares-Pinto}}, \bibinfo {author} {\bibfnamefont {I.~S.}\ \bibnamefont {Oliveira}}, \ and\ \bibinfo {author} {\bibfnamefont {G.}~\bibnamefont {Adesso}},\ }\href {\doibase 10.1103/PhysRevLett.112.210401} {\bibfield  {journal} {\bibinfo  {journal} {Phys. Rev. Lett.}\ }\textbf {\bibinfo {volume} {112}},\ \bibinfo {pages} {210401} (\bibinfo {year} {2014})}\BibitemShut {NoStop}%
\bibitem [{\citenamefont {Smirne}\ \emph {et~al.}(2016)\citenamefont {Smirne}, \citenamefont {Ko\l{}ody\ifmmode~\acute{n}\else \'{n}\fi{}ski}, \citenamefont {Huelga},\ and\ \citenamefont {Demkowicz-Dobrzański}}]{Smirne16}%
  \BibitemOpen
  \bibfield  {author} {\bibinfo {author} {\bibfnamefont {A.}~\bibnamefont {Smirne}}, \bibinfo {author} {\bibfnamefont {J.}~\bibnamefont {Ko\l{}ody\ifmmode~\acute{n}\else \'{n}\fi{}ski}}, \bibinfo {author} {\bibfnamefont {S.~F.}\ \bibnamefont {Huelga}}, \ and\ \bibinfo {author} {\bibfnamefont {R.}~\bibnamefont {Demkowicz-Dobrzański}},\ }\href {\doibase 10.1103/PhysRevLett.116.120801} {\bibfield  {journal} {\bibinfo  {journal} {Phys. Rev. Lett.}\ }\textbf {\bibinfo {volume} {116}},\ \bibinfo {pages} {120801} (\bibinfo {year} {2016})}\BibitemShut {NoStop}%
\bibitem [{\citenamefont {Seveso}\ \emph {et~al.}(2017)\citenamefont {Seveso}, \citenamefont {Rossi},\ and\ \citenamefont {Paris}}]{Seveso2017}%
  \BibitemOpen
  \bibfield  {author} {\bibinfo {author} {\bibfnamefont {L.}~\bibnamefont {Seveso}}, \bibinfo {author} {\bibfnamefont {M.~A.~C.}\ \bibnamefont {Rossi}}, \ and\ \bibinfo {author} {\bibfnamefont {M.~G.~A.}\ \bibnamefont {Paris}},\ }\href {\doibase 10.1103/PhysRevA.95.012111} {\bibfield  {journal} {\bibinfo  {journal} {Phys. Rev. A}\ }\textbf {\bibinfo {volume} {95}},\ \bibinfo {pages} {012111} (\bibinfo {year} {2017})}\BibitemShut {NoStop}%
\bibitem [{\citenamefont {Haase}\ \emph {et~al.}(2018)\citenamefont {Haase}, \citenamefont {Smirne}, \citenamefont {Kołodyński}, \citenamefont {Demkowicz-Dobrzański},\ and\ \citenamefont {Huelga}}]{Haase18}%
  \BibitemOpen
  \bibfield  {author} {\bibinfo {author} {\bibfnamefont {J.~F.}\ \bibnamefont {Haase}}, \bibinfo {author} {\bibfnamefont {A.}~\bibnamefont {Smirne}}, \bibinfo {author} {\bibfnamefont {J.}~\bibnamefont {Kołodyński}}, \bibinfo {author} {\bibfnamefont {R.}~\bibnamefont {Demkowicz-Dobrzański}}, \ and\ \bibinfo {author} {\bibfnamefont {S.~F.}\ \bibnamefont {Huelga}},\ }\href {\doibase 10.1088/1367-2630/aab67f} {\bibfield  {journal} {\bibinfo  {journal} {New Journal of Physics}\ }\textbf {\bibinfo {volume} {20}},\ \bibinfo {pages} {053009} (\bibinfo {year} {2018})}\BibitemShut {NoStop}%
\bibitem [{\citenamefont {Pang}\ and\ \citenamefont {Brun}(2014)}]{Metrology4}%
  \BibitemOpen
  \bibfield  {author} {\bibinfo {author} {\bibfnamefont {S.}~\bibnamefont {Pang}}\ and\ \bibinfo {author} {\bibfnamefont {T.~A.}\ \bibnamefont {Brun}},\ }\href {\doibase 10.1103/PHYSREVA.90.022117} {\bibfield  {journal} {\bibinfo  {journal} {Physical Review A - Atomic, Molecular, and Optical Physics}\ }\textbf {\bibinfo {volume} {90}},\ \bibinfo {pages} {022117} (\bibinfo {year} {2014})}\BibitemShut {NoStop}%
\bibitem [{\citenamefont {Rossi}\ \emph {et~al.}(2020)\citenamefont {Rossi}, \citenamefont {Albarelli}, \citenamefont {Tamascelli},\ and\ \citenamefont {Genoni}}]{Rossi20}%
  \BibitemOpen
  \bibfield  {author} {\bibinfo {author} {\bibfnamefont {M.~A.~C.}\ \bibnamefont {Rossi}}, \bibinfo {author} {\bibfnamefont {F.}~\bibnamefont {Albarelli}}, \bibinfo {author} {\bibfnamefont {D.}~\bibnamefont {Tamascelli}}, \ and\ \bibinfo {author} {\bibfnamefont {M.~G.}\ \bibnamefont {Genoni}},\ }\href {\doibase 10.1103/PhysRevLett.125.200505} {\bibfield  {journal} {\bibinfo  {journal} {Phys. Rev. Lett.}\ }\textbf {\bibinfo {volume} {125}},\ \bibinfo {pages} {200505} (\bibinfo {year} {2020})}\BibitemShut {NoStop}%
\bibitem [{\citenamefont {Zhou}\ and\ \citenamefont {Jiang}(2021)}]{Sisi21}%
  \BibitemOpen
  \bibfield  {author} {\bibinfo {author} {\bibfnamefont {S.}~\bibnamefont {Zhou}}\ and\ \bibinfo {author} {\bibfnamefont {L.}~\bibnamefont {Jiang}},\ }\href {\doibase 10.1103/PRXQuantum.2.010343} {\bibfield  {journal} {\bibinfo  {journal} {PRX Quantum}\ }\textbf {\bibinfo {volume} {2}},\ \bibinfo {pages} {010343} (\bibinfo {year} {2021})}\BibitemShut {NoStop}%
\bibitem [{\citenamefont {Gorecki}\ \emph {et~al.}(2020)\citenamefont {Gorecki}, \citenamefont {Zhou}, \citenamefont {Jiang},\ and\ \citenamefont {Demkowicz-Dobrzański}}]{Gorecki2019}%
  \BibitemOpen
  \bibfield  {author} {\bibinfo {author} {\bibfnamefont {W.}~\bibnamefont {Gorecki}}, \bibinfo {author} {\bibfnamefont {S.}~\bibnamefont {Zhou}}, \bibinfo {author} {\bibfnamefont {L.}~\bibnamefont {Jiang}}, \ and\ \bibinfo {author} {\bibfnamefont {R.}~\bibnamefont {Demkowicz-Dobrzański}},\ }\href {\doibase 10.22331/q-2020-07-02-288} {\bibfield  {journal} {\bibinfo  {journal} {Quantum}\ }\textbf {\bibinfo {volume} {4}},\ \bibinfo {pages} {288} (\bibinfo {year} {2020})}\BibitemShut {NoStop}%
\bibitem [{\citenamefont {Zhou}\ \emph {et~al.}(2018{\natexlab{a}})\citenamefont {Zhou}, \citenamefont {Zhang}, \citenamefont {Preskill},\ and\ \citenamefont {Jiang}}]{Zhou2018_2}%
  \BibitemOpen
  \bibfield  {author} {\bibinfo {author} {\bibfnamefont {S.}~\bibnamefont {Zhou}}, \bibinfo {author} {\bibfnamefont {M.}~\bibnamefont {Zhang}}, \bibinfo {author} {\bibfnamefont {J.}~\bibnamefont {Preskill}}, \ and\ \bibinfo {author} {\bibfnamefont {L.}~\bibnamefont {Jiang}},\ }\href {\doibase 10.1038/s41467-017-02510-3} {\bibfield  {journal} {\bibinfo  {journal} {Nature Communications}\ }\textbf {\bibinfo {volume} {9}},\ \bibinfo {pages} {78} (\bibinfo {year} {2018}{\natexlab{a}})}\BibitemShut {NoStop}%
\bibitem [{\citenamefont {Yuan}\ and\ \citenamefont {Fung}(2015)}]{Yuan2015}%
  \BibitemOpen
  \bibfield  {author} {\bibinfo {author} {\bibfnamefont {H.}~\bibnamefont {Yuan}}\ and\ \bibinfo {author} {\bibfnamefont {C.~H.~F.}\ \bibnamefont {Fung}},\ }\href {\doibase 10.1103/PhysRevLett.115.110401} {\bibfield  {journal} {\bibinfo  {journal} {Physical Review Letters}\ }\textbf {\bibinfo {volume} {115}},\ \bibinfo {pages} {110401} (\bibinfo {year} {2015})}\BibitemShut {NoStop}%
\bibitem [{\citenamefont {Yuan}(2016)}]{Yuan2016}%
  \BibitemOpen
  \bibfield  {author} {\bibinfo {author} {\bibfnamefont {H.}~\bibnamefont {Yuan}},\ }\href {\doibase 10.1103/PhysRevLett.117.160801} {\bibfield  {journal} {\bibinfo  {journal} {Phys. Rev. Lett.}\ }\textbf {\bibinfo {volume} {117}},\ \bibinfo {pages} {160801} (\bibinfo {year} {2016})}\BibitemShut {NoStop}%
\bibitem [{\citenamefont {Correa}\ \emph {et~al.}(2015)\citenamefont {Correa}, \citenamefont {Mehboudi}, \citenamefont {Adesso},\ and\ \citenamefont {Sanpera}}]{Correa2015}%
  \BibitemOpen
  \bibfield  {author} {\bibinfo {author} {\bibfnamefont {L.~A.}\ \bibnamefont {Correa}}, \bibinfo {author} {\bibfnamefont {M.}~\bibnamefont {Mehboudi}}, \bibinfo {author} {\bibfnamefont {G.}~\bibnamefont {Adesso}}, \ and\ \bibinfo {author} {\bibfnamefont {A.}~\bibnamefont {Sanpera}},\ }\href {\doibase 10.1103/PHYSREVLETT.114.220405} {\bibfield  {journal} {\bibinfo  {journal} {Physical Review Letters}\ }\textbf {\bibinfo {volume} {114}},\ \bibinfo {pages} {220405} (\bibinfo {year} {2015})}\BibitemShut {NoStop}%
\bibitem [{\citenamefont {Mehboudi}\ \emph {et~al.}(2019)\citenamefont {Mehboudi}, \citenamefont {Sanpera},\ and\ \citenamefont {Correa}}]{Mehboudi2019}%
  \BibitemOpen
  \bibfield  {author} {\bibinfo {author} {\bibfnamefont {M.}~\bibnamefont {Mehboudi}}, \bibinfo {author} {\bibfnamefont {A.}~\bibnamefont {Sanpera}}, \ and\ \bibinfo {author} {\bibfnamefont {L.~A.}\ \bibnamefont {Correa}},\ }\href {\doibase 10.1088/1751-8121/AB2828} {\bibfield  {journal} {\bibinfo  {journal} {Journal of Physics A: Mathematical and Theoretical}\ }\textbf {\bibinfo {volume} {52}},\ \bibinfo {pages} {303001} (\bibinfo {year} {2019})}\BibitemShut {NoStop}%
\bibitem [{\citenamefont {Christensen}\ and\ \citenamefont {Meyer}(2022)}]{GW4}%
  \BibitemOpen
  \bibfield  {author} {\bibinfo {author} {\bibfnamefont {N.}~\bibnamefont {Christensen}}\ and\ \bibinfo {author} {\bibfnamefont {R.}~\bibnamefont {Meyer}},\ }\href {\doibase 10.1103/REVMODPHYS.94.025001} {\bibfield  {journal} {\bibinfo  {journal} {Reviews of Modern Physics}\ }\textbf {\bibinfo {volume} {94}} (\bibinfo {year} {2022}),\ 10.1103/REVMODPHYS.94.025001}\BibitemShut {NoStop}%
\bibitem [{\citenamefont {Abadie}\ \emph {et~al.}(2011)\citenamefont {Abadie}, \citenamefont {Abbott}, \citenamefont {Abbott}, \citenamefont {Abbott}, \citenamefont {Abernathy} \emph {et~al.}}]{GW5}%
  \BibitemOpen
  \bibfield  {author} {\bibinfo {author} {\bibfnamefont {J.}~\bibnamefont {Abadie}}, \bibinfo {author} {\bibfnamefont {B.~P.}\ \bibnamefont {Abbott}}, \bibinfo {author} {\bibfnamefont {R.}~\bibnamefont {Abbott}}, \bibinfo {author} {\bibfnamefont {T.~D.}\ \bibnamefont {Abbott}}, \bibinfo {author} {\bibfnamefont {M.}~\bibnamefont {Abernathy}},  \emph {et~al.},\ }\href {\doibase 10.1038/nphys2083} {\bibfield  {journal} {\bibinfo  {journal} {Nature Physics}\ }\textbf {\bibinfo {volume} {7}},\ \bibinfo {pages} {962} (\bibinfo {year} {2011})}\BibitemShut {NoStop}%
\bibitem [{\citenamefont {Jones}\ \emph {et~al.}(2009)\citenamefont {Jones}, \citenamefont {Karlen}, \citenamefont {Fitzsimons}, \citenamefont {Ardavan}, \citenamefont {Benjamin}, \citenamefont {Briggs},\ and\ \citenamefont {Morton}}]{Jones2009}%
  \BibitemOpen
  \bibfield  {author} {\bibinfo {author} {\bibfnamefont {J.~A.}\ \bibnamefont {Jones}}, \bibinfo {author} {\bibfnamefont {S.~D.}\ \bibnamefont {Karlen}}, \bibinfo {author} {\bibfnamefont {J.}~\bibnamefont {Fitzsimons}}, \bibinfo {author} {\bibfnamefont {A.}~\bibnamefont {Ardavan}}, \bibinfo {author} {\bibfnamefont {S.~C.}\ \bibnamefont {Benjamin}}, \bibinfo {author} {\bibfnamefont {G.~A.~D.}\ \bibnamefont {Briggs}}, \ and\ \bibinfo {author} {\bibfnamefont {J.~J.}\ \bibnamefont {Morton}},\ }\href {\doibase 10.1126/SCIENCE.1170730} {\bibfield  {journal} {\bibinfo  {journal} {Science}\ }\textbf {\bibinfo {volume} {324}},\ \bibinfo {pages} {1166} (\bibinfo {year} {2009})}\BibitemShut {NoStop}%
\bibitem [{\citenamefont {Amorós-Binefa}\ and\ \citenamefont {Kołodyński}(2021)}]{Jan2021}%
  \BibitemOpen
  \bibfield  {author} {\bibinfo {author} {\bibfnamefont {J.}~\bibnamefont {Amorós-Binefa}}\ and\ \bibinfo {author} {\bibfnamefont {J.}~\bibnamefont {Kołodyński}},\ }\href {\doibase 10.1088/1367-2630/ac3b71} {\ \textbf {\bibinfo {volume} {23}},\ \bibinfo {pages} {123030} (\bibinfo {year} {2021})}\BibitemShut {NoStop}%
\bibitem [{\citenamefont {Brask}\ \emph {et~al.}(2015)\citenamefont {Brask}, \citenamefont {Chaves},\ and\ \citenamefont {Kolodynski}}]{Brask2015}%
  \BibitemOpen
  \bibfield  {author} {\bibinfo {author} {\bibfnamefont {J.~B.}\ \bibnamefont {Brask}}, \bibinfo {author} {\bibfnamefont {R.}~\bibnamefont {Chaves}}, \ and\ \bibinfo {author} {\bibfnamefont {J.}~\bibnamefont {Kolodynski}},\ }\href {\doibase 10.1103/PHYSREVX.5.031010} {\bibfield  {journal} {\bibinfo  {journal} {Physical Review X}\ }\textbf {\bibinfo {volume} {5}},\ \bibinfo {pages} {031010} (\bibinfo {year} {2015})}\BibitemShut {NoStop}%
\bibitem [{\citenamefont {Albarelli}\ \emph {et~al.}(2017)\citenamefont {Albarelli}, \citenamefont {Rossi}, \citenamefont {Paris},\ and\ \citenamefont {Genoni}}]{ARPG17}%
  \BibitemOpen
  \bibfield  {author} {\bibinfo {author} {\bibfnamefont {F.}~\bibnamefont {Albarelli}}, \bibinfo {author} {\bibfnamefont {M.~A.}\ \bibnamefont {Rossi}}, \bibinfo {author} {\bibfnamefont {M.~G.}\ \bibnamefont {Paris}}, \ and\ \bibinfo {author} {\bibfnamefont {M.~G.}\ \bibnamefont {Genoni}},\ }\href {\doibase 10.1088/1367-2630/aa9840} {\bibfield  {journal} {\bibinfo  {journal} {New J. Phys.}\ }\textbf {\bibinfo {volume} {19}},\ \bibinfo {pages} {123011} (\bibinfo {year} {2017})}\BibitemShut {NoStop}%
\bibitem [{\citenamefont {Degen}\ \emph {et~al.}(2017)\citenamefont {Degen}, \citenamefont {Reinhard},\ and\ \citenamefont {Cappellaro}}]{Degen2017}%
  \BibitemOpen
  \bibfield  {author} {\bibinfo {author} {\bibfnamefont {C.~L.}\ \bibnamefont {Degen}}, \bibinfo {author} {\bibfnamefont {F.}~\bibnamefont {Reinhard}}, \ and\ \bibinfo {author} {\bibfnamefont {P.}~\bibnamefont {Cappellaro}},\ }\href {\doibase 10.1103/REVMODPHYS.89.035002/FIGURES/13/MEDIUM} {\bibfield  {journal} {\bibinfo  {journal} {Reviews of Modern Physics}\ }\textbf {\bibinfo {volume} {89}},\ \bibinfo {pages} {035002} (\bibinfo {year} {2017})}\BibitemShut {NoStop}%
\bibitem [{\citenamefont {Marciniak}\ \emph {et~al.}(2022)\citenamefont {Marciniak}, \citenamefont {Feldker}, \citenamefont {Pogorelov}, \citenamefont {Kaubruegger}, \citenamefont {Vasilyev}, \citenamefont {van Bijnen}, \citenamefont {Schindler}, \citenamefont {Zoller}, \citenamefont {Blatt},\ and\ \citenamefont {Monz}}]{Marciniak2022}%
  \BibitemOpen
  \bibfield  {author} {\bibinfo {author} {\bibfnamefont {C.~D.}\ \bibnamefont {Marciniak}}, \bibinfo {author} {\bibfnamefont {T.}~\bibnamefont {Feldker}}, \bibinfo {author} {\bibfnamefont {I.}~\bibnamefont {Pogorelov}}, \bibinfo {author} {\bibfnamefont {R.}~\bibnamefont {Kaubruegger}}, \bibinfo {author} {\bibfnamefont {D.~V.}\ \bibnamefont {Vasilyev}}, \bibinfo {author} {\bibfnamefont {R.}~\bibnamefont {van Bijnen}}, \bibinfo {author} {\bibfnamefont {P.}~\bibnamefont {Schindler}}, \bibinfo {author} {\bibfnamefont {P.}~\bibnamefont {Zoller}}, \bibinfo {author} {\bibfnamefont {R.}~\bibnamefont {Blatt}}, \ and\ \bibinfo {author} {\bibfnamefont {T.}~\bibnamefont {Monz}},\ }\href {\doibase 10.1038/s41586-022-04435-4} {\bibfield  {journal} {\bibinfo  {journal} {Nature}\ }\textbf {\bibinfo {volume} {603}},\ \bibinfo {pages} {604} (\bibinfo {year} {2022})}\BibitemShut {NoStop}%
\bibitem [{\citenamefont {Zwick}\ and\ \citenamefont {Álvarez}(2023)}]{Zwick2023}%
  \BibitemOpen
  \bibfield  {author} {\bibinfo {author} {\bibfnamefont {A.}~\bibnamefont {Zwick}}\ and\ \bibinfo {author} {\bibfnamefont {G.~A.}\ \bibnamefont {Álvarez}},\ }\href {\doibase https://doi.org/10.1016/j.jmro.2023.100113} {\bibfield  {journal} {\bibinfo  {journal} {Journal of Magnetic Resonance Open}\ }\textbf {\bibinfo {volume} {16-17}},\ \bibinfo {pages} {100113} (\bibinfo {year} {2023})}\BibitemShut {NoStop}%
\bibitem [{\citenamefont {Tsang}\ \emph {et~al.}(2016)\citenamefont {Tsang}, \citenamefont {Nair},\ and\ \citenamefont {Lu}}]{Tsang16}%
  \BibitemOpen
  \bibfield  {author} {\bibinfo {author} {\bibfnamefont {M.}~\bibnamefont {Tsang}}, \bibinfo {author} {\bibfnamefont {R.}~\bibnamefont {Nair}}, \ and\ \bibinfo {author} {\bibfnamefont {X.-M.}\ \bibnamefont {Lu}},\ }\href {\doibase 10.1103/PhysRevX.6.031033} {\bibfield  {journal} {\bibinfo  {journal} {Phys. Rev. X}\ }\textbf {\bibinfo {volume} {6}},\ \bibinfo {pages} {031033} (\bibinfo {year} {2016})}\BibitemShut {NoStop}%
\bibitem [{\citenamefont {Tsang}(2021)}]{Tsang21}%
  \BibitemOpen
  \bibfield  {author} {\bibinfo {author} {\bibfnamefont {M.}~\bibnamefont {Tsang}},\ }\href {\doibase 10.1103/PhysRevA.104.052411} {\bibfield  {journal} {\bibinfo  {journal} {Phys. Rev. A}\ }\textbf {\bibinfo {volume} {104}},\ \bibinfo {pages} {052411} (\bibinfo {year} {2021})}\BibitemShut {NoStop}%
\bibitem [{\citenamefont {Lupo}\ \emph {et~al.}(2020)\citenamefont {Lupo}, \citenamefont {Huang},\ and\ \citenamefont {Kok}}]{Lupo20}%
  \BibitemOpen
  \bibfield  {author} {\bibinfo {author} {\bibfnamefont {C.}~\bibnamefont {Lupo}}, \bibinfo {author} {\bibfnamefont {Z.}~\bibnamefont {Huang}}, \ and\ \bibinfo {author} {\bibfnamefont {P.}~\bibnamefont {Kok}},\ }\href {\doibase 10.1103/PhysRevLett.124.080503} {\bibfield  {journal} {\bibinfo  {journal} {Phys. Rev. Lett.}\ }\textbf {\bibinfo {volume} {124}},\ \bibinfo {pages} {080503} (\bibinfo {year} {2020})}\BibitemShut {NoStop}%
\bibitem [{\citenamefont {Fiderer}\ \emph {et~al.}(2021)\citenamefont {Fiderer}, \citenamefont {Tufarelli}, \citenamefont {Piano},\ and\ \citenamefont {Adesso}}]{Fiderer21}%
  \BibitemOpen
  \bibfield  {author} {\bibinfo {author} {\bibfnamefont {L.~J.}\ \bibnamefont {Fiderer}}, \bibinfo {author} {\bibfnamefont {T.}~\bibnamefont {Tufarelli}}, \bibinfo {author} {\bibfnamefont {S.}~\bibnamefont {Piano}}, \ and\ \bibinfo {author} {\bibfnamefont {G.}~\bibnamefont {Adesso}},\ }\href {\doibase 10.1103/PRXQuantum.2.020308} {\bibfield  {journal} {\bibinfo  {journal} {PRX Quantum}\ }\textbf {\bibinfo {volume} {2}},\ \bibinfo {pages} {020308} (\bibinfo {year} {2021})}\BibitemShut {NoStop}%
\bibitem [{\citenamefont {Oh}\ \emph {et~al.}(2021)\citenamefont {Oh}, \citenamefont {Zhou}, \citenamefont {Wong},\ and\ \citenamefont {Jiang}}]{Oh21}%
  \BibitemOpen
  \bibfield  {author} {\bibinfo {author} {\bibfnamefont {C.}~\bibnamefont {Oh}}, \bibinfo {author} {\bibfnamefont {S.}~\bibnamefont {Zhou}}, \bibinfo {author} {\bibfnamefont {Y.}~\bibnamefont {Wong}}, \ and\ \bibinfo {author} {\bibfnamefont {L.}~\bibnamefont {Jiang}},\ }\href {\doibase 10.1103/PhysRevLett.126.120502} {\bibfield  {journal} {\bibinfo  {journal} {Phys. Rev. Lett.}\ }\textbf {\bibinfo {volume} {126}},\ \bibinfo {pages} {120502} (\bibinfo {year} {2021})}\BibitemShut {NoStop}%
\bibitem [{\citenamefont {Tsang}(2019)}]{Tsang19}%
  \BibitemOpen
  \bibfield  {author} {\bibinfo {author} {\bibfnamefont {M.}~\bibnamefont {Tsang}},\ }\href {\doibase 10.1103/PhysRevResearch.1.033006} {\bibfield  {journal} {\bibinfo  {journal} {Phys. Rev. Res.}\ }\textbf {\bibinfo {volume} {1}},\ \bibinfo {pages} {033006} (\bibinfo {year} {2019})}\BibitemShut {NoStop}%
\bibitem [{\citenamefont {Tsang}\ \emph {et~al.}(2020)\citenamefont {Tsang}, \citenamefont {Albarelli},\ and\ \citenamefont {Datta}}]{Tsang2020}%
  \BibitemOpen
  \bibfield  {author} {\bibinfo {author} {\bibfnamefont {M.}~\bibnamefont {Tsang}}, \bibinfo {author} {\bibfnamefont {F.}~\bibnamefont {Albarelli}}, \ and\ \bibinfo {author} {\bibfnamefont {A.}~\bibnamefont {Datta}},\ }\href {\doibase 10.1103/PHYSREVX.10.031023} {\bibfield  {journal} {\bibinfo  {journal} {Physical Review X}\ }\textbf {\bibinfo {volume} {10}},\ \bibinfo {pages} {031023} (\bibinfo {year} {2020})}\BibitemShut {NoStop}%
\bibitem [{\citenamefont {Gambetta}\ and\ \citenamefont {Wiseman}(2001)}]{GW01}%
  \BibitemOpen
  \bibfield  {author} {\bibinfo {author} {\bibfnamefont {J.}~\bibnamefont {Gambetta}}\ and\ \bibinfo {author} {\bibfnamefont {H.~M.}\ \bibnamefont {Wiseman}},\ }\href@noop {} {\bibfield  {journal} {\bibinfo  {journal} {Phys. Rev. A}\ }\textbf {\bibinfo {volume} {64}},\ \bibinfo {pages} {042105} (\bibinfo {year} {2001})}\BibitemShut {NoStop}%
\bibitem [{\citenamefont {Gu\c{t}\u{a}}(2011)}]{Guta2011}%
  \BibitemOpen
  \bibfield  {author} {\bibinfo {author} {\bibfnamefont {M.}~\bibnamefont {Gu\c{t}\u{a}}},\ }\href {\doibase https://doi.org/10.1103/PhysRevA.83.062324} {\bibfield  {journal} {\bibinfo  {journal} {Physical Review A}\ }\textbf {\bibinfo {volume} {83}},\ \bibinfo {pages} {062324} (\bibinfo {year} {2011})}\BibitemShut {NoStop}%
\bibitem [{\citenamefont {Gammelmark}\ and\ \citenamefont {M\o{}lmer}(2014)}]{Molmer14}%
  \BibitemOpen
  \bibfield  {author} {\bibinfo {author} {\bibfnamefont {S.}~\bibnamefont {Gammelmark}}\ and\ \bibinfo {author} {\bibfnamefont {K.}~\bibnamefont {M\o{}lmer}},\ }\href {\doibase 10.1103/PhysRevLett.112.170401} {\bibfield  {journal} {\bibinfo  {journal} {Phys. Rev. Lett.}\ }\textbf {\bibinfo {volume} {112}},\ \bibinfo {pages} {170401} (\bibinfo {year} {2014})}\BibitemShut {NoStop}%
\bibitem [{\citenamefont {Gu\c{t}\u{a}}\ and\ \citenamefont {Kiukas}(2015)}]{Guta_2015}%
  \BibitemOpen
  \bibfield  {author} {\bibinfo {author} {\bibfnamefont {M.}~\bibnamefont {Gu\c{t}\u{a}}}\ and\ \bibinfo {author} {\bibfnamefont {J.}~\bibnamefont {Kiukas}},\ }\href {\doibase 10.1007/S00220-014-2253-0} {\bibfield  {journal} {\bibinfo  {journal} {Communications in Mathematical Physics}\ }\textbf {\bibinfo {volume} {335}},\ \bibinfo {pages} {1397} (\bibinfo {year} {2015})}\BibitemShut {NoStop}%
\bibitem [{\citenamefont {Catana}\ \emph {et~al.}(2015)\citenamefont {Catana}, \citenamefont {Bouten},\ and\ \citenamefont {Gu\c{t}\u{a}}}]{GutaCB15}%
  \BibitemOpen
  \bibfield  {author} {\bibinfo {author} {\bibfnamefont {C.}~\bibnamefont {Catana}}, \bibinfo {author} {\bibfnamefont {L.}~\bibnamefont {Bouten}}, \ and\ \bibinfo {author} {\bibfnamefont {M.}~\bibnamefont {Gu\c{t}\u{a}}},\ }\href {\doibase 10.1088/1751-8113/48/36/365301} {\bibfield  {journal} {\bibinfo  {journal} {Journal of Physics A: Mathematical and Theoretical}\ }\textbf {\bibinfo {volume} {48}},\ \bibinfo {pages} {365301} (\bibinfo {year} {2015})}\BibitemShut {NoStop}%
\bibitem [{\citenamefont {Ilias}\ \emph {et~al.}(2022)\citenamefont {Ilias}, \citenamefont {Yang}, \citenamefont {Huelga},\ and\ \citenamefont {Plenio}}]{Ilias22}%
  \BibitemOpen
  \bibfield  {author} {\bibinfo {author} {\bibfnamefont {T.}~\bibnamefont {Ilias}}, \bibinfo {author} {\bibfnamefont {D.}~\bibnamefont {Yang}}, \bibinfo {author} {\bibfnamefont {S.~F.}\ \bibnamefont {Huelga}}, \ and\ \bibinfo {author} {\bibfnamefont {M.~B.}\ \bibnamefont {Plenio}},\ }\href {\doibase 10.1103/PRXQuantum.3.010354} {\bibfield  {journal} {\bibinfo  {journal} {PRX Quantum}\ }\textbf {\bibinfo {volume} {3}},\ \bibinfo {pages} {010354} (\bibinfo {year} {2022})}\BibitemShut {NoStop}%
\bibitem [{\citenamefont {Fallani}\ \emph {et~al.}(2022)\citenamefont {Fallani}, \citenamefont {Rossi}, \citenamefont {Tamascelli},\ and\ \citenamefont {Genoni}}]{Fallani22}%
  \BibitemOpen
  \bibfield  {author} {\bibinfo {author} {\bibfnamefont {A.}~\bibnamefont {Fallani}}, \bibinfo {author} {\bibfnamefont {M.~A.~C.}\ \bibnamefont {Rossi}}, \bibinfo {author} {\bibfnamefont {D.}~\bibnamefont {Tamascelli}}, \ and\ \bibinfo {author} {\bibfnamefont {M.~G.}\ \bibnamefont {Genoni}},\ }\href {\doibase 10.1103/PRXQuantum.3.020310} {\bibfield  {journal} {\bibinfo  {journal} {PRX Quantum}\ }\textbf {\bibinfo {volume} {3}},\ \bibinfo {pages} {020310} (\bibinfo {year} {2022})}\BibitemShut {NoStop}%
\bibitem [{\citenamefont {Tsang}\ \emph {et~al.}(2011)\citenamefont {Tsang}, \citenamefont {Wiseman},\ and\ \citenamefont {Caves}}]{TWC11}%
  \BibitemOpen
  \bibfield  {author} {\bibinfo {author} {\bibfnamefont {M.}~\bibnamefont {Tsang}}, \bibinfo {author} {\bibfnamefont {H.~M.}\ \bibnamefont {Wiseman}}, \ and\ \bibinfo {author} {\bibfnamefont {C.}~\bibnamefont {Caves}},\ }\href {\doibase https://doi.org/10.1103/PhysRevLett.106.090401} {\bibfield  {journal} {\bibinfo  {journal} {Phys. Rev. Lett.}\ }\textbf {\bibinfo {volume} {106}},\ \bibinfo {pages} {090401} (\bibinfo {year} {2011})}\BibitemShut {NoStop}%
\bibitem [{\citenamefont {Berry}\ \emph {et~al.}(2015)\citenamefont {Berry}, \citenamefont {Tsang}, \citenamefont {Hall},\ and\ \citenamefont {Wiseman}}]{Berry2015}%
  \BibitemOpen
  \bibfield  {author} {\bibinfo {author} {\bibfnamefont {D.~W.}\ \bibnamefont {Berry}}, \bibinfo {author} {\bibfnamefont {M.}~\bibnamefont {Tsang}}, \bibinfo {author} {\bibfnamefont {M.~J.}\ \bibnamefont {Hall}}, \ and\ \bibinfo {author} {\bibfnamefont {H.~M.}\ \bibnamefont {Wiseman}},\ }\href {\doibase 10.1103/PhysRevX.5.031018} {\bibfield  {journal} {\bibinfo  {journal} {Physical Review X}\ }\textbf {\bibinfo {volume} {5}} (\bibinfo {year} {2015}),\ 10.1103/PhysRevX.5.031018}\BibitemShut {NoStop}%
\bibitem [{\citenamefont {Ng}\ \emph {et~al.}(2016)\citenamefont {Ng}, \citenamefont {Ang}, \citenamefont {Wheatley}, \citenamefont {Yonezawa}, \citenamefont {Furusawa}, \citenamefont {Huntington},\ and\ \citenamefont {Tsang}}]{Ng16}%
  \BibitemOpen
  \bibfield  {author} {\bibinfo {author} {\bibfnamefont {S.}~\bibnamefont {Ng}}, \bibinfo {author} {\bibfnamefont {S.~Z.}\ \bibnamefont {Ang}}, \bibinfo {author} {\bibfnamefont {T.~A.}\ \bibnamefont {Wheatley}}, \bibinfo {author} {\bibfnamefont {H.}~\bibnamefont {Yonezawa}}, \bibinfo {author} {\bibfnamefont {A.}~\bibnamefont {Furusawa}}, \bibinfo {author} {\bibfnamefont {E.~H.}\ \bibnamefont {Huntington}}, \ and\ \bibinfo {author} {\bibfnamefont {M.}~\bibnamefont {Tsang}},\ }\href {\doibase 10.1103/PhysRevA.93.042121} {\bibfield  {journal} {\bibinfo  {journal} {Phys. Rev. A}\ }\textbf {\bibinfo {volume} {93}},\ \bibinfo {pages} {042121} (\bibinfo {year} {2016})}\BibitemShut {NoStop}%
\bibitem [{\citenamefont {Norris}\ \emph {et~al.}(2016)\citenamefont {Norris}, \citenamefont {Paz-Silva},\ and\ \citenamefont {Viola}}]{Norris16}%
  \BibitemOpen
  \bibfield  {author} {\bibinfo {author} {\bibfnamefont {L.~M.}\ \bibnamefont {Norris}}, \bibinfo {author} {\bibfnamefont {G.~A.}\ \bibnamefont {Paz-Silva}}, \ and\ \bibinfo {author} {\bibfnamefont {L.}~\bibnamefont {Viola}},\ }\href {\doibase 10.1103/PhysRevLett.116.150503} {\bibfield  {journal} {\bibinfo  {journal} {Phys. Rev. Lett.}\ }\textbf {\bibinfo {volume} {116}},\ \bibinfo {pages} {150503} (\bibinfo {year} {2016})}\BibitemShut {NoStop}%
\bibitem [{\citenamefont {Shi}\ and\ \citenamefont {Zhuang}(2023)}]{Shi23}%
  \BibitemOpen
  \bibfield  {author} {\bibinfo {author} {\bibfnamefont {H.}~\bibnamefont {Shi}}\ and\ \bibinfo {author} {\bibfnamefont {Q.}~\bibnamefont {Zhuang}},\ }\href@noop {} {\bibfield  {journal} {\bibinfo  {journal} {npj Quantum Inf}\ }\textbf {\bibinfo {volume} {9}},\ \bibinfo {pages} {27} (\bibinfo {year} {2023})}\BibitemShut {NoStop}%
\bibitem [{\citenamefont {Sung}\ \emph {et~al.}(2019)\citenamefont {Sung}, \citenamefont {Beaudoin}, \citenamefont {Norris}, \citenamefont {Yan}, \citenamefont {Kim}, \citenamefont {Qiu}, \citenamefont {von L{\"u}pke}, \citenamefont {Yoder}, \citenamefont {Orlando}, \citenamefont {Gustavsson} \emph {et~al.}}]{Sung19}%
  \BibitemOpen
  \bibfield  {author} {\bibinfo {author} {\bibfnamefont {Y.}~\bibnamefont {Sung}}, \bibinfo {author} {\bibfnamefont {F.}~\bibnamefont {Beaudoin}}, \bibinfo {author} {\bibfnamefont {L.~M.}\ \bibnamefont {Norris}}, \bibinfo {author} {\bibfnamefont {F.}~\bibnamefont {Yan}}, \bibinfo {author} {\bibfnamefont {D.~K.}\ \bibnamefont {Kim}}, \bibinfo {author} {\bibfnamefont {J.~Y.}\ \bibnamefont {Qiu}}, \bibinfo {author} {\bibfnamefont {U.}~\bibnamefont {von L{\"u}pke}}, \bibinfo {author} {\bibfnamefont {J.~L.}\ \bibnamefont {Yoder}}, \bibinfo {author} {\bibfnamefont {T.~P.}\ \bibnamefont {Orlando}}, \bibinfo {author} {\bibfnamefont {S.}~\bibnamefont {Gustavsson}},  \emph {et~al.},\ }\href {\doibase 10.1038/s41467-019-11699-4} {\bibfield  {journal} {\bibinfo  {journal} {Nature communications}\ }\textbf {\bibinfo {volume} {10}},\ \bibinfo {pages} {3715} (\bibinfo {year} {2019})}\BibitemShut {NoStop}%
\bibitem [{\citenamefont {Tsang}(2023)}]{Tsang23}%
  \BibitemOpen
  \bibfield  {author} {\bibinfo {author} {\bibfnamefont {M.}~\bibnamefont {Tsang}},\ }\href {\doibase 10.1103/PhysRevA.107.012611} {\bibfield  {journal} {\bibinfo  {journal} {Phys. Rev. A}\ }\textbf {\bibinfo {volume} {107}},\ \bibinfo {pages} {012611} (\bibinfo {year} {2023})}\BibitemShut {NoStop}%
\bibitem [{\citenamefont {Holevo}(2011)}]{Holevo2011}%
  \BibitemOpen
  \bibfield  {author} {\bibinfo {author} {\bibfnamefont {A.}~\bibnamefont {Holevo}},\ }\href {\doibase 10.1007/978-88-7642-378-9} {\emph {\bibinfo {title} {Probabilistic and Statistical Aspects of Quantum Theory}}}\ (\bibinfo  {publisher} {Edizioni della Normale},\ \bibinfo {year} {2011})\BibitemShut {NoStop}%
\bibitem [{\citenamefont {Yuen}\ and\ \citenamefont {Lax}(1973)}]{YuenLax76}%
  \BibitemOpen
  \bibfield  {author} {\bibinfo {author} {\bibfnamefont {H.}~\bibnamefont {Yuen}}\ and\ \bibinfo {author} {\bibfnamefont {M.}~\bibnamefont {Lax}},\ }\href {\doibase 10.1109/TIT.1973.1055103} {\bibfield  {journal} {\bibinfo  {journal} {IEEE Trans. Inf. Theory}\ }\textbf {\bibinfo {volume} {19}},\ \bibinfo {pages} {740} (\bibinfo {year} {1973})}\BibitemShut {NoStop}%
\bibitem [{\citenamefont {Belavkin}(1976)}]{Belavkin76}%
  \BibitemOpen
  \bibfield  {author} {\bibinfo {author} {\bibfnamefont {V.}~\bibnamefont {Belavkin}},\ }\href {\doibase 10.1007/BF01032091} {\bibfield  {journal} {\bibinfo  {journal} {Theoretical and Mathematical Physics}\ }\textbf {\bibinfo {volume} {26}},\ \bibinfo {pages} {213} (\bibinfo {year} {1976})}\BibitemShut {NoStop}%
\bibitem [{\citenamefont {Helstrom}(1976)}]{Helstrom1976}%
  \BibitemOpen
  \bibfield  {author} {\bibinfo {author} {\bibfnamefont {C.~W.}\ \bibnamefont {Helstrom}},\ }\href@noop {} {\emph {\bibinfo {title} {Quantum detection and estimation theory}}}\ (\bibinfo  {publisher} {Academic Press},\ \bibinfo {year} {1976})\BibitemShut {NoStop}%
\bibitem [{\citenamefont {Braunstein}\ and\ \citenamefont {Caves}(1994)}]{QCR1}%
  \BibitemOpen
  \bibfield  {author} {\bibinfo {author} {\bibfnamefont {S.~L.}\ \bibnamefont {Braunstein}}\ and\ \bibinfo {author} {\bibfnamefont {C.~M.}\ \bibnamefont {Caves}},\ }\href {\doibase 10.1103/PhysRevLett.72.3439} {\bibfield  {journal} {\bibinfo  {journal} {Phys. Rev. Lett.}\ }\textbf {\bibinfo {volume} {72}},\ \bibinfo {pages} {3439} (\bibinfo {year} {1994})}\BibitemShut {NoStop}%
\bibitem [{\citenamefont {Nagaoka}(2005)}]{Nagaoka}%
  \BibitemOpen
  \bibfield  {author} {\bibinfo {author} {\bibfnamefont {H.}~\bibnamefont {Nagaoka}},\ }in\ \href@noop {} {\emph {\bibinfo {booktitle} {Asymptotic Theory of Quantum Statistical Inference: Selected Papers}}}\ (\bibinfo  {publisher} {World Scientific},\ \bibinfo {year} {2005})\ p.\ \bibinfo {pages} {100–112},\ \bibinfo {note} {originally published as IEICE Technical Report, 89, 228, IT 89–42, 9–14 (1989)}\BibitemShut {NoStop}%
\bibitem [{\citenamefont {Gill}\ and\ \citenamefont {Massar}(2000)}]{GillMassar}%
  \BibitemOpen
  \bibfield  {author} {\bibinfo {author} {\bibfnamefont {R.~D.}\ \bibnamefont {Gill}}\ and\ \bibinfo {author} {\bibfnamefont {S.}~\bibnamefont {Massar}},\ }\href {\doibase 10.1103/PhysRevA.61.042312} {\bibfield  {journal} {\bibinfo  {journal} {Phys. Rev. A}\ }\textbf {\bibinfo {volume} {61}},\ \bibinfo {pages} {042312} (\bibinfo {year} {2000})}\BibitemShut {NoStop}%
\bibitem [{\citenamefont {G\'orecki}\ \emph {et~al.}(2020)\citenamefont {G\'orecki}, \citenamefont {Demkowicz-Dobrzański}, \citenamefont {Wiseman},\ and\ \citenamefont {Berry}}]{RDDWiseman}%
  \BibitemOpen
  \bibfield  {author} {\bibinfo {author} {\bibfnamefont {W.}~\bibnamefont {G\'orecki}}, \bibinfo {author} {\bibfnamefont {R.}~\bibnamefont {Demkowicz-Dobrzański}}, \bibinfo {author} {\bibfnamefont {H.~M.}\ \bibnamefont {Wiseman}}, \ and\ \bibinfo {author} {\bibfnamefont {D.~W.}\ \bibnamefont {Berry}},\ }\href {\doibase 10.1103/PhysRevLett.124.030501} {\bibfield  {journal} {\bibinfo  {journal} {Phys. Rev. Lett.}\ }\textbf {\bibinfo {volume} {124}},\ \bibinfo {pages} {030501} (\bibinfo {year} {2020})}\BibitemShut {NoStop}%
\bibitem [{\citenamefont {Pezze}\ and\ \citenamefont {Smerzi}(2014)}]{NullQFI1}%
  \BibitemOpen
  \bibfield  {author} {\bibinfo {author} {\bibfnamefont {L.}~\bibnamefont {Pezze}}\ and\ \bibinfo {author} {\bibfnamefont {A.}~\bibnamefont {Smerzi}},\ }\href {\doibase 10.3254/978-1-61499-488-0-691} {\bibfield  {journal} {\bibinfo  {journal} {Atom Interferometry}\ }\textbf {\bibinfo {volume} {188}},\ \bibinfo {pages} {691} (\bibinfo {year} {2014})}\BibitemShut {NoStop}%
\bibitem [{\citenamefont {Pezz\`e}\ \emph {et~al.}(2017)\citenamefont {Pezz\`e}, \citenamefont {Ciampini}, \citenamefont {Spagnolo}, \citenamefont {Humphreys}, \citenamefont {Datta}, \citenamefont {Walmsley}, \citenamefont {Barbieri}, \citenamefont {Sciarrino},\ and\ \citenamefont {Smerzi}}]{NullQFI2}%
  \BibitemOpen
  \bibfield  {author} {\bibinfo {author} {\bibfnamefont {L.}~\bibnamefont {Pezz\`e}}, \bibinfo {author} {\bibfnamefont {M.~A.}\ \bibnamefont {Ciampini}}, \bibinfo {author} {\bibfnamefont {N.}~\bibnamefont {Spagnolo}}, \bibinfo {author} {\bibfnamefont {P.~C.}\ \bibnamefont {Humphreys}}, \bibinfo {author} {\bibfnamefont {A.}~\bibnamefont {Datta}}, \bibinfo {author} {\bibfnamefont {I.~A.}\ \bibnamefont {Walmsley}}, \bibinfo {author} {\bibfnamefont {M.}~\bibnamefont {Barbieri}}, \bibinfo {author} {\bibfnamefont {F.}~\bibnamefont {Sciarrino}}, \ and\ \bibinfo {author} {\bibfnamefont {A.}~\bibnamefont {Smerzi}},\ }\href {\doibase 10.1103/PhysRevLett.119.130504} {\bibfield  {journal} {\bibinfo  {journal} {Phys. Rev. Lett.}\ }\textbf {\bibinfo {volume} {119}},\ \bibinfo {pages} {130504} (\bibinfo {year} {2017})}\BibitemShut {NoStop}%
\bibitem [{\citenamefont {Liu}\ \emph {et~al.}(2019)\citenamefont {Liu}, \citenamefont {Yuan}, \citenamefont {Lu},\ and\ \citenamefont {Wang}}]{NullQFI3}%
  \BibitemOpen
  \bibfield  {author} {\bibinfo {author} {\bibfnamefont {J.}~\bibnamefont {Liu}}, \bibinfo {author} {\bibfnamefont {H.}~\bibnamefont {Yuan}}, \bibinfo {author} {\bibfnamefont {X.-M.}\ \bibnamefont {Lu}}, \ and\ \bibinfo {author} {\bibfnamefont {X.}~\bibnamefont {Wang}},\ }\href {\doibase 10.1088/1751-8121/AB5D4D} {\bibfield  {journal} {\bibinfo  {journal} {Journal of Physics A: Mathematical and Theoretical}\ }\textbf {\bibinfo {volume} {53}},\ \bibinfo {pages} {023001} (\bibinfo {year} {2019})}\BibitemShut {NoStop}%
\bibitem [{\citenamefont {Tsang}()}]{Tsang_blog}%
  \BibitemOpen
  \bibfield  {author} {\bibinfo {author} {\bibfnamefont {M.}~\bibnamefont {Tsang}},\ }\href@noop {} {\enquote {\bibinfo {title} {Caveats of the {C}r\'amer-{R}ao bound},}\ }\bibinfo {note} {{h}ttps://blog.nus.edu.sg/mankei/caveats-of-the-cramer-rao-bound/}\BibitemShut {NoStop}%
\bibitem [{\citenamefont {Yang}\ \emph {et~al.}(2023)\citenamefont {Yang}, \citenamefont {Huelga},\ and\ \citenamefont {Plenio}}]{DayouCounting}%
  \BibitemOpen
  \bibfield  {author} {\bibinfo {author} {\bibfnamefont {D.}~\bibnamefont {Yang}}, \bibinfo {author} {\bibfnamefont {S.~F.}\ \bibnamefont {Huelga}}, \ and\ \bibinfo {author} {\bibfnamefont {M.~B.}\ \bibnamefont {Plenio}},\ }\href {\doibase 10.1103/PhysRevX.13.031012} {\bibfield  {journal} {\bibinfo  {journal} {Phys. Rev. X}\ }\textbf {\bibinfo {volume} {13}},\ \bibinfo {pages} {031012} (\bibinfo {year} {2023})}\BibitemShut {NoStop}%
\bibitem [{\citenamefont {Godley}\ \emph {et~al.}()\citenamefont {Godley}, \citenamefont {Girotti},\ and\ \citenamefont {Gu\c{t}\u{a}}}]{GGG}%
  \BibitemOpen
  \bibfield  {author} {\bibinfo {author} {\bibfnamefont {A.}~\bibnamefont {Godley}}, \bibinfo {author} {\bibfnamefont {F.}~\bibnamefont {Girotti}}, \ and\ \bibinfo {author} {\bibfnamefont {M.}~\bibnamefont {Gu\c{t}\u{a}}},\ }\href@noop {} {\enquote {\bibinfo {title} {Asymptotic estimation theory for quantum {M}arkov chains},}\ }\BibitemShut {NoStop}%
\bibitem [{\citenamefont {Lehmann}\ and\ \citenamefont {Casella}(1998)}]{Lehmann1998}%
  \BibitemOpen
  \bibfield  {author} {\bibinfo {author} {\bibfnamefont {E.~L.}\ \bibnamefont {Lehmann}}\ and\ \bibinfo {author} {\bibfnamefont {G.}~\bibnamefont {Casella}},\ }\href {\doibase 10.2307/1270597} {\bibfield  {journal} {\bibinfo  {journal} {Design}\ }\textbf {\bibinfo {volume} {41}} (\bibinfo {year} {1998}),\ 10.2307/1270597}\BibitemShut {NoStop}%
\bibitem [{\citenamefont {van~der Vaart}(1998)}]{Vaart1998}%
  \BibitemOpen
  \bibfield  {author} {\bibinfo {author} {\bibfnamefont {A.~W.}\ \bibnamefont {van~der Vaart}},\ }\href {\doibase 10.1017/CBO9780511802256} {\emph {\bibinfo {title} {Asymptotic Statistics}}}\ (\bibinfo  {publisher} {Cambridge University Press},\ \bibinfo {year} {1998})\BibitemShut {NoStop}%
\bibitem [{\citenamefont {Zhou}\ \emph {et~al.}(2018{\natexlab{b}})\citenamefont {Zhou}, \citenamefont {Zou},\ and\ \citenamefont {Jiang}}]{Zhou2018}%
  \BibitemOpen
  \bibfield  {author} {\bibinfo {author} {\bibfnamefont {S.}~\bibnamefont {Zhou}}, \bibinfo {author} {\bibfnamefont {C.-L.}\ \bibnamefont {Zou}}, \ and\ \bibinfo {author} {\bibfnamefont {L.}~\bibnamefont {Jiang}},\ }\href {\doibase 10.1088/2058-9565/ab71f8} {\bibfield  {journal} {\bibinfo  {journal} {Quantum Science and Technology}\ }\textbf {\bibinfo {volume} {5}},\ \bibinfo {pages} {025005} (\bibinfo {year} {2018}{\natexlab{b}})}\BibitemShut {NoStop}%
\bibitem [{\citenamefont {Gammelmark}\ \emph {et~al.}(2014)\citenamefont {Gammelmark} \emph {et~al.}}]{Molmer}%
  \BibitemOpen
  \bibfield  {author} {\bibinfo {author} {\bibfnamefont {S.}~\bibnamefont {Gammelmark}} \emph {et~al.},\ }\href {\doibase 10.1103/PHYSREVA.89.043839} {\bibfield  {journal} {\bibinfo  {journal} {Physical Review A - Atomic, Molecular, and Optical Physics}\ }\textbf {\bibinfo {volume} {89}},\ \bibinfo {pages} {043839} (\bibinfo {year} {2014})}\BibitemShut {NoStop}%
\bibitem [{\citenamefont {Leonhardt}(1997)}]{Leonhardt}%
  \BibitemOpen
  \bibfield  {author} {\bibinfo {author} {\bibfnamefont {U.}~\bibnamefont {Leonhardt}},\ }\href@noop {} {\emph {\bibinfo {title} {Measuring the Quantum State of Light}}}\ (\bibinfo  {publisher} {Cambridge University Press},\ \bibinfo {year} {1997})\BibitemShut {NoStop}%
\bibitem [{\citenamefont {Braunstein}\ \emph {et~al.}(1996)\citenamefont {Braunstein}, \citenamefont {Caves},\ and\ \citenamefont {Milburn}}]{QCR2}%
  \BibitemOpen
  \bibfield  {author} {\bibinfo {author} {\bibfnamefont {S.~L.}\ \bibnamefont {Braunstein}}, \bibinfo {author} {\bibfnamefont {C.~M.}\ \bibnamefont {Caves}}, \ and\ \bibinfo {author} {\bibfnamefont {G.~J.}\ \bibnamefont {Milburn}},\ }\href {\doibase 10.1006/APHY.1996.0040} {\bibfield  {journal} {\bibinfo  {journal} {Annals of Physics}\ }\textbf {\bibinfo {volume} {247}},\ \bibinfo {pages} {135} (\bibinfo {year} {1996})}\BibitemShut {NoStop}%
\bibitem [{\citenamefont {Gu\c{t}\u{a}}\ \emph {et~al.}(2008)\citenamefont {Gu\c{t}\u{a}}, \citenamefont {Janssens},\ and\ \citenamefont {Kahn}}]{LAN5}%
  \BibitemOpen
  \bibfield  {author} {\bibinfo {author} {\bibfnamefont {M.}~\bibnamefont {Gu\c{t}\u{a}}}, \bibinfo {author} {\bibfnamefont {B.}~\bibnamefont {Janssens}}, \ and\ \bibinfo {author} {\bibfnamefont {J.}~\bibnamefont {Kahn}},\ }\href {\doibase 10.1007/s00220-007-0357-5} {\bibfield  {journal} {\bibinfo  {journal} {Commun. Math. Phys.}\ }\textbf {\bibinfo {volume} {277}},\ \bibinfo {pages} {127} (\bibinfo {year} {2008})}\BibitemShut {NoStop}%
\bibitem [{\citenamefont {Petz}(2008)}]{Petz2008}%
  \BibitemOpen
  \bibfield  {author} {\bibinfo {author} {\bibfnamefont {D.}~\bibnamefont {Petz}},\ }\href {\doibase 10.1007/978-3-540-74636-2} {\bibfield  {journal} {\bibinfo  {journal} {Theoretical and Mathematical Physics}\ } (\bibinfo {year} {2008}),\ 10.1007/978-3-540-74636-2}\BibitemShut {NoStop}%
\bibitem [{\citenamefont {Ragy}\ \emph {et~al.}(2016)\citenamefont {Ragy}, \citenamefont {Jarzyna},\ and\ \citenamefont {Demkowicz-Dobrzański}}]{Ragy16}%
  \BibitemOpen
  \bibfield  {author} {\bibinfo {author} {\bibfnamefont {S.}~\bibnamefont {Ragy}}, \bibinfo {author} {\bibfnamefont {M.}~\bibnamefont {Jarzyna}}, \ and\ \bibinfo {author} {\bibfnamefont {R.}~\bibnamefont {Demkowicz-Dobrzański}},\ }\href {\doibase 10.1103/PhysRevA.94.052108} {\bibfield  {journal} {\bibinfo  {journal} {Phys. Rev. A}\ }\textbf {\bibinfo {volume} {94}},\ \bibinfo {pages} {052108} (\bibinfo {year} {2016})}\BibitemShut {NoStop}%
\bibitem [{\citenamefont {Matsumoto}(2002)}]{Ma02}%
  \BibitemOpen
  \bibfield  {author} {\bibinfo {author} {\bibfnamefont {K.}~\bibnamefont {Matsumoto}},\ }\href {\doibase 10.1088/0305-4470/35/13/307} {\bibfield  {journal} {\bibinfo  {journal} {Journal of Physics A: Mathematical and General}\ }\textbf {\bibinfo {volume} {35}},\ \bibinfo {pages} {3111} (\bibinfo {year} {2002})}\BibitemShut {NoStop}%
\bibitem [{\citenamefont {Godley}\ and\ \citenamefont {Gu\c{t}\u{a}}(2023)}]{Godley2023}%
  \BibitemOpen
  \bibfield  {author} {\bibinfo {author} {\bibfnamefont {A.}~\bibnamefont {Godley}}\ and\ \bibinfo {author} {\bibfnamefont {M.}~\bibnamefont {Gu\c{t}\u{a}}},\ }\href {\doibase 10.22331/q-2023-04-06-973} {\bibfield  {journal} {\bibinfo  {journal} {Quantum}\ }\textbf {\bibinfo {volume} {7}},\ \bibinfo {pages} {973} (\bibinfo {year} {2023})}\BibitemShut {NoStop}%
\bibitem [{\citenamefont {Schnabel}(2017)}]{SCHNABEL17}%
  \BibitemOpen
  \bibfield  {author} {\bibinfo {author} {\bibfnamefont {R.}~\bibnamefont {Schnabel}},\ }\href {\doibase https://doi.org/10.1016/j.physrep.2017.04.001} {\bibfield  {journal} {\bibinfo  {journal} {Physics Reports}\ }\textbf {\bibinfo {volume} {684}},\ \bibinfo {pages} {1} (\bibinfo {year} {2017})}\BibitemShut {NoStop}%
\bibitem [{\citenamefont {Hans-A.~Bachor}(2019)}]{BachorRalph}%
  \BibitemOpen
  \bibfield  {author} {\bibinfo {author} {\bibfnamefont {T.~C.~R.}\ \bibnamefont {Hans-A.~Bachor}},\ }\href {\doibase 10.1002/9783527695805} {\emph {\bibinfo {title} {A Guide to Experiments in Quantum Optics}}}\ (\bibinfo  {publisher} {Wiley‐VCH Verlag GmbH \& Co. KGaA},\ \bibinfo {year} {2019})\BibitemShut {NoStop}%
\bibitem [{\citenamefont {Young}\ and\ \citenamefont {Smith}(2005)}]{young&Smith}%
  \BibitemOpen
  \bibfield  {author} {\bibinfo {author} {\bibfnamefont {G.~A.}\ \bibnamefont {Young}}\ and\ \bibinfo {author} {\bibfnamefont {R.~L.}\ \bibnamefont {Smith}},\ }\href@noop {} {\emph {\bibinfo {title} {Essentials of Statistical Inference}}}\ (\bibinfo  {publisher} {Cambridge University Press},\ \bibinfo {year} {2005})\BibitemShut {NoStop}%
\end{thebibliography}%

\appendix
\section{Proof of the QCRB-saturating POVM} \label{sec:QCRBsat}
In this appendix we present the necessary and sufficient conditions for a quantum Cram\'er-Rao bound saturating POVM, based on the proof by Zhou et al. \cite{Zhou2018}, highlighting the features that allow this POVM to saturate the bound. The classical Fisher information corresponding to the measurement $\mathcal{M}=\{M_i\}$ is given by

\begin{align*}
    I_{\mathcal{M}}(\theta) &= \sum_{i:{\rm Tr}(M_i\rho_\theta)\neq0} \frac{({\rm Tr}(M_i \partial_\theta \rho_\theta))^2} {{\rm Tr}(M_i \rtht)} \\
    &= \sum_{i:{\rm Tr}(M_i\rho_\theta)\neq0} \frac{({\rm Re}[Tr(M_i \mathcal{L}_\theta\rtht)])^2} {{\rm Tr}(M_i \rtht)},
\end{align*}
where we have used the Lyapunov equation, see section \ref{sec:QCRB}, and identified that the resulting term corresponds to the above real component. Clearly 

\begin{align*}
    \sum_{i:{\rm Tr}(M_i\rho_\theta)\neq0} &\frac{({\rm Re}[{\rm Tr}(M_i \mathcal{L}_\theta\rtht)])^2} {{\rm Tr}(M_i \rtht)} \\
    &\leq 
    \sum_{i:{\rm Tr}(M_i\rho_\theta)\neq0} \frac{|{\rm Tr}(M_i \mathcal{L}_\theta \rtht)|^2}{{\rm Tr}(M_i \rtht)},
\end{align*}
where we have equality when ${\rm Im}[{\rm Tr}(M_i \mathcal{L}_\theta \rtht)] = 0$. We now use the Cauchy-Schwarz inequality to cancel the denominator, identifying the terms $M_i^{\half}\rtht^\half$ and $M_i^\half \Ltht \rtht^\half$ within the expression above, finding

\begin{equation}
    \sum_{i:{\rm Tr}(M_i\rho_\theta)\neq0} \frac{|{\rm Tr}(M_i \mathcal{L}_\theta \rtht)|^2}{{\rm Tr}(M_i \rtht)} \leq
    \sum_{i:{\rm Tr}(M_i\rho_\theta)\neq0} {\rm Tr}(M_i \Ltht \rtht \Ltht).
\end{equation}
The equality holds if 
\begin{equation}
    M_i^\half \rtht^\half = \lambda_i M_i^\half \mathcal{L}_\theta \rtht^\half
\end{equation}
for all $i$ such that ${\rm Tr}(M_i\rho_\theta)\neq0$ for some $\lambda_i\in \mathbb{C}$.
Finally, since $M_i$ represents a POVM we have 

\begin{equation}
    \sum_{i:{\rm Tr}(M_i\rho_\theta)\neq0} {\rm Tr}(M_i \Ltht \rtht \Ltht) \leq {\rm Tr}(\mathcal{L}_\theta^2 \rtht) \equiv F(\theta).
\end{equation}
where $F(\theta)$ is the quantum Fisher information. Equality is achieved when ${\rm Tr}(M_i \Ltht \rtht \Ltht)=0$ for all $i$ such that ${\rm Tr}(M_i\rtht)=0$. Note: this ensures that all the information on the parameter is contained in measurable outcomes with non-zero probabilities. 

In summary, to achieve the QFI the POVM $\{M_i\}$ needs to satisfy the 3 following conditions:
\begin{enumerate}
    \item If ${\rm Tr}(M_i\rtht)>0$ then ${\rm Im}[{\rm Tr}(M_i \mathcal{L}_\theta \rtht)] = 0$,
    \item If ${\rm Tr}(M_i\rtht)>0$ then  $M_i^\half \rtht^\half = \lambda_i M_i^\half \mathcal{L}_\theta \rtht^\half$ ,$\lambda_i\in\mathbb{C}$,
        \item If ${\rm Tr}(M_i\rtht)=0$ then ${\rm Tr}(M_i \Ltht \rtht \Ltht)=0$.  
\end{enumerate}

\comm{
We now consider a pure state model $\ket{\psi_\theta}$ and apply these conditions in order to construct an optimal measurement basis. Firstly, for pure states we have 
\begin{align}
    \Ltht = &2(|\partht\psi_\theta\rangle\langle\psi_\theta| + |\pstht\rangle\langle\partht\pstht|), \\
    F(\theta) = &4(\|\partht \pstht\|^2-|\langle\partht \pstht | \pstht\rangle|^2).
\end{align}
}

We can combine conditions 1. and 2. into the following condition:
\begin{enumerate}
    \item[4.] If ${\rm Tr}(M_i\rtht)>0$ then $M_i^\half \rtht^\half = \lambda_i M_i^\half \mathcal{L}_\theta \rtht^\half$, $\lambda_i\in\mathbb{R}$ .
\end{enumerate}

\section{Parameter localisation via a two step adaptive procedure}
\label{sec:adaptive.argument}
Here we discuss in more detail the general \emph{parameter localisation} principle to which we refer repeatedly in the paper. The principle is formulated for one-dimensional models, but can be extended staightforwardly to multidimensional ones.

Suppose we are given a large number $n$ of independent, identically prepared systems in the state $\rho_\theta$, depending smoothly on a parameter $\theta$ which lies in an open set $\Theta\subset\mathbb{R}$. To avoid pathological cases we assume that $F(\theta)>f>0$ for all $\theta\in\Theta.$
%
Even though the set $\Theta$ is a priori `large', we can `localise' the parameter and subsequently perform  measurements adapted to the parameter value, by using the following two step procedure. 

Consider a measurement $\mathcal{M}$ such that $\theta$ is \emph{identifiable}, i.e. no two different parameters produce outcomes with identical probability distributions.
In the first step we apply $\mathcal{M}$ to each system belonging to a vanishingly small proportion $\tilde{n}=o(n)$ of the samples, with $\tilde{n}$ growing with $n$. For concreteness we assume that   $\tilde{n}= n^{1-\epsilon} \ll n$, with $\epsilon>0$ a small number, but the arguments hold for a generic choice. Using the data obtained from measuring this sub-ensemble we construct a preliminary rough estimator $\tilde{\theta}_n$ of $\theta$. 

Naturally, one would like the estimator $\tilde{\theta}_n$ to be `pretty good' (given the used sample size), but not necessarily optimal. There are several properties that could embody this requirement; for example one could require that the mean square error (MSE) scales at the standard rate $\tilde{n}^{-1}=n^{-1+\epsilon}$ and $\tilde{\theta}_n$ is asymptotically normal, i.e. it concentrates around $\theta$ at rate $\tilde{n}^{-1/2}$ with 
$$
\sqrt{\tilde{n}} (\tilde{\theta}_n-\theta) \longrightarrow
N(0, V_\theta)
$$ 
where the convergence to the normal distribution  holds in law as $n\to\infty$ and the variance satisfies $V_\theta\geq F(\theta)^{-1}$. Standard estimators such as maximum likelihood typically satisfy this property \cite{Lehmann1998}. In particular, this implies that the (confidence) interval 
$$
I_n = (\tilde{\theta}_n-n^{-1/2 +\epsilon}, ~ \tilde{\theta}_n+n^{-1/2 +\epsilon})
$$ 
contains $\theta$ with probability converging to one exponentially fast. This follows from the fact that the ratio between the size of 
$I_n$  and the standard deviation of $\tilde{\theta}_n-\theta$ diverges as  $|I_n|/n^{(-1+\epsilon)/2} = n^{\epsilon/2}$.

On the other hand, if one adopts a Bayesian viewpoint and assumes the existence of a prior distribution on $\Theta$ with density $\pi(\theta)$, then it is natural to require the asymptotic normality of the posterior density 
$$
p(\theta|\tilde{\theta}_n):= \frac{\pi(\theta) p(\tilde{\theta}_n|\theta)}{p(\tilde{\theta}_n)}
$$ 
where $p(\tilde{\theta}_n|\theta) = p_\theta(\tilde{\theta}_n)$ and $p(\tilde{\theta}_n) =  \int p(\tilde{\theta}_n|\theta) \pi(d\theta)$. Intuitively this follows from the asymptotic normality of $\tilde{\theta}_n$. Indeed if the prior $\pi(\theta)$ and the variance $V_\theta$ are sufficiently regular with respect to $\theta$ and $V_\theta$ is bounded away from zero, then 
$p(\theta| \tilde{\theta}_n)\propto \exp(-\tilde{n}(\tilde{\theta}_n-\theta)^2/2 V_\theta)$ concentrates around $\tilde{\theta}_n$  with approximately normal distribution. For more details on asymptotic normality theory we refer to \cite{Vaart1998,young&Smith}. For our purposes, it will suffice to assume that $\tilde{\theta}_n$ is a `reasonable' estimator in the sense that the posterior distribution is `balanced' with respect to  $\tilde{\theta}_n$
in a sense that is precisely defined in section \ref{sec.null}. In particular, this means that we exclude `dishonest' estimators for which the mass of the posterior distribution lies largely on one side of the estimator. For instance, taking a reasonable estimator $\tilde{\theta}_n$ and adding a constant 
$\delta_n$ such that $\delta_n /\tilde{n}^{-1/2} \to \infty$ for large $n$ would be an example of a dishonest estimator. As we will see later this distinction becomes important since the preliminary estimator enters the definition of the second stage estimator, and the performance of null measurements based on reasonable/dishonest estimators is radically different.

\emph{Adaptive step.} We pass now to the second step of the estimation procedure in which one measures the remaining $n^\prime = n-\tilde{n}$ systems, taking into account the information provided by the first step. We distinguish two measurement strategies, the SLD measurement and the approximative null measurement.

For 
one-dimensional parameters, an \emph{optimal} procedure is to measure the SLD $\mathcal{L}_{\tilde{\theta}_n}$ separately on each system and then construct the (final) estimator 
$\hat{\theta}_n = \tilde{\theta}_n + \bar{X}_n/F(\tilde{\theta}_n)
$ 
where $
\bar{X}_n= \frac{1}{n^\prime} \sum_{i=1}^{n^\prime}X_i
$
is the average of the measurement outcomes. 
Assuming that the preliminary estimator is consistent (i.e. $\tilde{\theta}_n\to \theta$ for large $n$), we obtain that  $\hat{\theta}_n$  achieves the multicopy (asymptotic) version of the QCRB in the sense that 
$$
\lim_{n\to\infty}
n \mathbb{E}_\theta [(\hat{\theta}_n -\theta)^2 ] = F(\theta)^{-1}.
$$
Moreover, $\hat{\theta}_n$ is  asymptotically normal 
$$
\sqrt{n}(\hat{\theta}_n -\theta) \longrightarrow N(0,F(\theta)^{-1})
$$
thus providing us with simple asymptotic confidence intervals.

Let us consider now the case of  null measurements. In section \ref{sec:QCRB} we showed that if we measure in a basis 
$\mathcal{B}(\theta) = 
\{|v_1\rangle, \dots,|v_d\rangle\}$ such that $|v_1\rangle = \ket{\psi_\theta}$ then the classical Fisher information is zero. However, at $\tilde{\theta}\approx \theta$ ($\tilde{\theta} \neq \theta$) the approximate null measurement with respect to a basis $\mathcal{B}(\tilde{\theta})$ has classical Fisher information $I_{\mathcal{B}(\tilde{\theta})}\approx F(\theta)$. As anticipated in section \ref{sec:QCRB}, the adaptive strategy use for the SLD measurement does not work in the case of null measurements when the initial estimator is reasonable. Proving this will be the subject of section \ref{sec.null}.

Finally, let us briefly consider the case of multidimensional parameter models. In this setting, separate measurements may not be optimal in the second step due to non-commutativity of the SLD operators for different parameter components. However, using the information contained in $\tilde{\theta}_n$, we can devise collective measurements procedures which are asymptotically optimal in the sense of achieving the Holevo bound \cite{Holevo2011}. This can be understood by employing the local asymptotic normality (LAN) theory \cite{LAN1, LAN2, LAN3}, which we briefly recall in section \ref{sec:LAN-qudits}.

Note that since step one uses a vanishing proportion of the samples, the asymptotic result remains the same if we assume that $n$ samples are available in the second step. Therefore, in order to simplify notation, in the sequel we will replace $n^\prime=n-\tilde{n}$ by $n$.

\section{Proof of Lemma \ref{lem:reas} on existence of reasonable estimators}
\label{sec:proof.lemma.reasonable.estimator}

The outcomes $X_1, \dots , X_{\tilde{n}}$ have probabilities
$$
p_\theta(0) = \cos^2(\theta- \pi/4), \quad
p_\theta(1) = \sin^2(\theta- \pi/4).
$$
The maximum likelihood estimator is $\tilde{\theta}_n = \pi/4+
f^{-1}(\bar{X}_n)$ 
where 
$\bar{X}_n = 
\frac{1}{\tilde{n}}\sum_{i=1}^{\tilde{n}}X_i
$
and $f(x) =\sin^2(x)$ (which is invertible on $(-\pi/2, 0)$).

For every $k=0, \dots, \tilde{n}$ and $\theta \in [-\pi/8, \pi/8]$, the density of the posterior at time $n$ is given by
\begin{equation} \label{eq:post}
\begin{split}
&\pi(\theta|\tilde{\theta}_n=\pi/4+f^{-1}(k/ \tilde{n}))=\\
&\frac{\sin^{2k}(\theta-\pi/4) \cos^{2( \tilde{n}-k)}(\theta-\pi/4)}{\int_{-\frac{\pi}{8}}^{\frac{\pi}{8}}\sin^{2k}(\zeta-\pi/4) \cos^{2( \tilde{n}-k)}(\zeta-\pi/4)d\zeta}
\end{split}
\end{equation}
and the unconditional distribution of $\tilde{\theta}_n$ is given by
\[\begin{split}
&\mathbb{P}(\tilde{\theta}_n =\pi/4+f^{-1}(k/ \tilde{n}))=\\
&\frac{4}{\pi} \binom{ \tilde{n}}{k}\int_{-\frac{\pi}{8}}^{\frac{\pi}{8}}\sin^{2k}(\theta-\pi/4) \cos^{2( \tilde{n}-k)}(\theta-\pi/4)d\theta.
\end{split}
\]
Consistency of $\tilde{\theta}_n$ and the dominated convergence theorem imply that for every Borel set
\[
\lim_{n \rightarrow +\infty}\mathbb{P}(\tilde{\theta}_n \in A) =\frac{4}{\pi} \int_{-\frac{\pi}{8}}^{\frac{\pi}{8}}\chi_{A}(\theta) d\theta.
\]
This allows us to consider for instance $A_n:=(-\pi/8,\pi/8)$. Moreover, we can rewrite equation \eqref{eq:post} as
\[
\frac{e^{- \tilde{n} H(\sin^{2}(\tilde{\theta}_n-\pi/4),\,\sin^{2}(\theta-\pi/4))}}{\int_{-\frac{\pi}{8}}^{\frac{\pi}{8}}e^{- \tilde{n} H(\sin^{2}(\tilde{\theta}_n-\pi/4),\,\sin^{2}(\zeta-\pi/4))}d\zeta}
\]
where $H(p,q)=-p\log(q)-(1-p)\log(1-q)$. Notice that if $\tilde{\theta}_n \in (-\pi/8,\pi/8)$, then $H(\sin^{2}(\tilde{\theta}_n-\pi/4),\,\sin^{2}(\theta-\pi/4))$ admits a unique minimum in $[-\pi/8, \pi/8]$ at $\theta=\tilde{\theta}_n$ where it vanishes and where the value of the second derivative is equal to $4$. Therefore, for $n$ big enough (uniformly in $\tilde{\theta}_n \in (-\pi/8,\pi/8)$) one has that
\[
\begin{split}
&\int_{\tilde{\theta}_n + \tau_n}^{\tilde{\theta}_n + \sqrt{\tilde{n}}} e^{- \tilde{n} H(\sin^{2}(\tilde{\theta}_n-\pi/4),\,\sin^{2}(\theta-\pi/4))}d\theta \\
&\geq \int_{\tilde{\theta}_n + \tau_n}^{\tilde{\theta}_n + \tilde{n}^{-1/2}} e^{- \frac{5\tilde{n} (\theta-\tilde{\theta}_n)^2}{2}}d\theta\\
&=\frac{1}{\sqrt{\tilde{n}}}\int_{n^{-\epsilon/4}}^{1}e^{- \frac{5\theta^2}{2}}d\theta
\end{split}
\]
and analogously for the integral in the interval $[\tilde{\theta}_n-\sqrt{\tilde{n}},\tilde{\theta}_n-\tau_n]$. Moreover
\begin{equation*}
\begin{split}
&\int_{-\frac{\pi}{8}}^{\frac{\pi}{8}}e^{- \tilde{n} H(\sin^{2}(\tilde{\theta}_n-\pi/4),\,\sin^{2}(\zeta-\pi/4))}d\zeta \\
&=\int_{|\zeta-\tilde{\theta}_n| \geq n^{-1/2+\epsilon}}e^{- \tilde{n} H(\sin^{2}(\tilde{\theta}_n-\pi/4),\,\sin^{2}(\zeta-\pi/4))}d\zeta\\
&+\int_{|\zeta-\tilde{\theta}_n| \leq n^{-1/2+\epsilon}}e^{- \tilde{n} H(\sin^{2}(\tilde{\theta}_n-\pi/4),\,\sin^{2}(\zeta-\pi/4))}d\zeta\\
&\leq e^{-3n^{\epsilon}/2}+ \int_{\tilde{\theta}_n - n^{-1/2+\epsilon}}^{\tilde{\theta}_n + n^{-1/2+\epsilon}} e^{- \frac{3\tilde{n} (\zeta-\tilde{\theta}_n)^2}{2}}d\zeta\\
&= e^{-3n^{\epsilon}/2}+ \frac{1}{\sqrt{\tilde{n}}}\int_{-n^{-\epsilon/2}}^{n^{-\epsilon/2}}e^{- \frac{3\zeta^2}{2}}d\zeta
\end{split}
\end{equation*}
and we are done, since we just proved that for $n$ big enough (uniformly in $\tilde{\theta}_n \in (-\pi/8,\pi/8)$), one has
\[\begin{split}
    &\mathbb{P}(\theta \geq \tilde{\theta}_n \pm \tau_n) \\
    &\geq \frac{\int_{n^{-\epsilon/4}}^{1}e^{- \frac{5\theta^2}{2}}d\theta}{\sqrt{\tilde{n}}e^{-3n^{\epsilon}/2}+ \int_{-n^{-\epsilon/2}}^{n^{-\epsilon/2}}e^{- \frac{3\zeta^2}{2}}d\zeta}\geq c >0
\end{split}\]
for some $c$ independent on $n$ and $\tilde{\theta}_n.$

\qed 
\section{Proof of Theorem \ref{prop.null} on suboptimality of the null measurement with reasonable preliminary estimator}
\label{app:proof.theorem.null.measurement}

Taking into account the two step procedure we write
\begin{eqnarray*}
\mathbb{E}_\theta [(\hat{\theta}_n -\theta)^2] &=& \int_{\mathbb{R}} p(d\hat{\theta}_n|\theta)
(\hat{\theta}_n -\theta)^2\\
&=&
\int_{\mathbb{R}^2}
p(d\tilde{\theta}_n|\theta) p(d\hat{\theta}_n|\theta, \tilde{\theta}_n) (\hat{\theta}_n -\theta)^2
\end{eqnarray*}
where $p(d\tilde{\theta}_n|\theta)$ is the distribution of the preliminary estimator at $\theta$ and $p(d\hat{\theta}_n|\theta, \tilde{\theta}_n)$ is the distribution of the final estimator given $\theta$ and $\tilde{\theta}_n$. Since the final estimator is obtained by measuring at angle $\tilde{\theta}_n$, its distribution depends only on $r=|\tilde{\theta}_n -\theta|$, so $p(d\hat{\theta}_n|\tilde{\theta}_n,\theta)= p_r(d\hat{\theta}_n)$.

The Bayesian risk of the final estimator $\hat{\theta}_n$ is
\begin{eqnarray*}
R_\pi(\hat{\theta}_n) &=&
\mathbb{E} 
[(\hat{\theta}_n-\theta)^2] \\
&=&
\int_{\Theta \times \mathbb{R}^2} \pi(d\theta) p(d\tilde{\theta}_n|\theta)
p(d\hat{\theta}_n|\theta,\tilde{\theta}_n) (\hat{\theta}_n-\theta)^2\\
&=&
\int_{\mathbb{R}} p(d\tilde{\theta}_n)
\int_{\Theta \times \mathbb{R}} \pi(d\theta|\tilde{\theta}_n)p(d\hat{\theta}_n|\theta,\tilde{\theta}_n)(\hat{\theta}_n-\theta)^2
\end{eqnarray*}
We have
\begin{eqnarray*}
&&\int_\Theta\int_{\mathbb{R}} \pi(d\theta|\tilde{\theta}_n)p(d\hat{\theta}_n|\theta,\tilde{\theta}_n)(\hat{\theta}_n-\theta)^2
\\
&&=
\int_{r\geq 0}dr\int_{\mathbb{R}}  p_r(d\hat{\theta}_n)  \times\\
&&
[\pi(\tilde{\theta}_n +r|\tilde{\theta}_n)(\hat{\theta}_n -\tilde{\theta}_n -r)^2+\\
&&\pi(\tilde{\theta}_n -r|\tilde{\theta}_n)(\hat{\theta}_n -\tilde{\theta}_n +r)^2]   
\end{eqnarray*}
By assumption, $\pi(\tilde{\theta}_n \pm r|\tilde{\theta}_n) \geq 
g_{n,\tilde{\theta}_n}(r) $ and since
$$
(\hat{\theta}_n - \tilde{\theta}_n -r)^2+ 
(\hat{\theta}_n - \tilde{\theta}_n +r)^2
\geq
2 r^2
$$
we get that for every $\tilde{\theta}_n \in A_n$
\begin{eqnarray*}
&&\int_{\Theta \times \mathbb{R}} \pi(d\theta|\tilde{\theta}_n)p(d\hat{\theta}_n|\theta,\tilde{\theta}_n)(\hat{\theta}_n-\theta)^2\\
&&\geq 
\int_{|r|\geq n^{-(1/2-\epsilon/4)}} g_{n,\tilde{\theta}_n}(r) 2r^2 dr \int_\mathbb{R} p_r (d\hat{\theta}_n) \\
&&\geq 2Cn^{-1+\epsilon/2},\\
\end{eqnarray*}
where the last inequality follows from point 2. in the definition of the reasonable estimator. Finally, since $\mathbb{P}(\tilde{\theta}_n\in A_n)\geq c$, we get 
$$
R_\pi (\hat{\theta}_n) \geq 2c Cn^{-1+\epsilon/2} 
$$
which implies the result.

\qed

\section{Proof of Theorem \ref{prop.null} for weaker notions of unidentifiability} \label{sec:local}
As we already mentioned, in the proof of Proposition \ref{prop.null} we made use of the fact that for the statistical model defined in equation \eqref{eq:qubit.model}, the law of the measurements in the basis containing $\ket{\psi_{\tilde{\theta}_n}}$ could not distinguish between $\tilde{\theta}_n \pm r$.

In the qubit case, we can still prove Theorem \ref{prop.null} for a wider class of one-parameter models under two additional assumptions. The first one is asking that unidentifiable parameters concentrate around the preliminary estimate at the same speed on both sides; more precisely, let us consider a general (smooth) one parameter model $\ket{\psi_{{\tilde{\theta}_n}+r}}$ for $r \in (-a,a)$; the corresponding probabilities describing the measurement in the $\tilde{\theta}_n$-null-basis are given by
\[p_r(1)=|\langle\psi_{\tilde{\theta}_n}|\psi_{{\tilde{\theta}_n}+r}\rangle|^2=1-p_r(0).\]

In general, there is no reason why $p_r=p_{-r}$, however at $r=0$ the function $p_r(1)$ has a global minimum and we can pick a neighborhood $(-a^\prime,b^\prime)$ such that
\begin{enumerate}
\item $p_r(1)$ is invertible on $(-a^\prime,0]$ and $[0,b^\prime)$,
\item $p_r(1)$ maps both $(-a^\prime,0]$ and $[0,b^\prime)$ onto the same interval.
\end{enumerate}
A priori, the neighborhood depends on $\tilde{\theta}_n$, but if the preliminary estimator takes value in a compact set, we can find a nonempty neighborhood $(-a^\prime,b^\prime)$ that works for every value of $\tilde{\theta}_n$. If we denote by $r^\prime(r)$ the unique value in $(-a^\prime,0]$ such that $p_{r^\prime(r)}(1)=p_r(1)$ for $r \in [0,b^\prime)$, we require
\begin{equation} \label{eq:bilip}
\begin{split}
&r \mapsto r^\prime(r) \text{ and its inverse to be Lipschitz with a }\\
&\text{Lipschitz constant } L \text{ that is uniform in }\tilde{\theta}_n.
\end{split}
\end{equation}

The second assumption consists in replacing the condition in equation \ref{eq:constantmass} with
\begin{equation} \label{eq:constantmassI}
\min \left \{\int_{\tau_n}^{\tau_n^\prime}\,g_{n,\tilde{\theta}_n}(r)dr, \int_{\tau_n/L}^{L\tau_n^\prime}\,g_{n,\tilde{\theta}_n}(r)dr\right \} \geq C
\end{equation}
where $ \tau_n := n^{-(1-\epsilon+\alpha)/2}$, $\tau^\prime_n := n^{-(1-\epsilon-\beta)/2}$ for some $\alpha,\beta$ such that $\tau^\prime_n=o(1)$, and $C>0$ is a constant independent on $n$. The additional requirement is that the posterior measure concentrates around the preliminary estimator. Under assumptions \eqref{eq:bilip} and \eqref{eq:constantmassI}, the proof of Theorem \ref{prop.null} can be adapted quite straightforwardly. This reparametrisation trick, however, cannot be repeated in the qudit case.

\bigskip More generally, if instead of conditions \ref{eq:constantmass} and \ref{eq:constantmassI}, we assume that
\begin{equation} \label{eq:constantmassII}
\int_{\tau_n}^{\tau_n^\prime}\,g_{n,\tilde{\theta}_n}(r)dr \geq C
\end{equation}
where $ \tau_n$, $\tau^\prime_n$ and $C>0$ satisfy the same conditions above and, moreover, we require the preliminary estimator to be enough accurate, i.e. $\epsilon<1/3-(2\alpha+5\beta)/3$, we can prove Theorem \ref{prop.null} without any assumption on the statistical model. We will make use of the fact that the conditional law of the measurements in the $\tilde{\theta}_n$-null basis conditional to the parameter $\theta=\tilde{\theta}_n +r$ does not distinguish between $\theta=\tilde{\theta}_n\pm r$ locally (which is the condition equivalent to have zero Fisher information), i.e.
\[\dot{p}_0(0)=\dot{p}_0(1)=0,\]
where the derivative is taken with respect to $r$.
\begin{proposition}
Consider any one-parameter qudit model $\{\ket{\psi_\theta}\}$ and assume that $\tilde{\theta}_n$ is a reasonably good estimator satisfying condition \eqref{eq:constantmassII}, obtained by measuring a sub-ensemble of size $n^{1-\epsilon}$ with $\epsilon<1/3-(2\alpha+5\beta)/3$.
Let $\hat{\theta}_n$ be an estimator of $\theta$ based on measuring the remaining $n-n^{1-\epsilon}$ sub-ensemble in a basis containing $\ket{\psi_{\tilde{\theta}_n}}$. Then   
$$
\lim_{n\to\infty} 
n R_{\pi} (\hat{\theta}_n) = \infty.
$$   
\end{proposition}
\begin{proof}
First notice that
\[
\begin{split}
    \mathbb{E}[(\hat{\theta}_n-\theta)^2] &\geq \mathbb{P}(|\hat{\theta}_n-\theta|>\tau_n^\prime)\tau_n^{\prime2}\\
    &+\mathbb{E}\left [\chi_{|\hat{\theta}_n-\theta| \leq \tau_n^\prime}(\hat{\theta}_n-\theta)^2 \right ]. 
\end{split}
\]
Moreover we can write
\[\begin{split}
&\int_{|\hat{\theta}_n-\theta| \leq \tau_n^\prime} \pi(d\theta|\tilde{\theta}_n)p(d\hat{\theta}_n|\theta,\tilde{\theta}_n)(\hat{\theta}_n-\theta)^2\geq\\
&\int _{\substack{\tau_n \leq |\theta-\tilde{\theta}_n|\leq \tau_n^\prime\\|\hat{\theta}_n-\theta| \leq \tau_n^\prime}}\pi(d\theta|\tilde{\theta}_n)p(d\hat{\theta}_n|\theta,\tilde{\theta}_n)(\hat{\theta}_n-\theta)^2\geq \\
&\int_{\substack{\tau_n \leq r\leq \tau_n^\prime\\|\hat{\theta}_n-\theta| \leq \tau_n^\prime}} g_{n,\tilde{\theta}_n}(r)(p_r(d\hat{\theta}_n)(\hat{\theta}_n-\tilde{\theta}_n-r)^2+\\
&p_{-r}(d\hat{\theta}_n)(\hat{\theta}_n-\tilde{\theta}_n+r)^2)=\\
&\int_{\substack{\tau_n \leq r\leq \tau_n^\prime\\|\hat{\theta}_n-\theta| \leq \tau_n^\prime}} g_{n,\tilde{\theta}_n}(r)(p_r(d\hat{\theta}_n)+p_{-r}(d\hat{\theta}_n))((\hat{\theta}_n-\tilde{\theta}_n)^2+r^2))+\\
&-2\int_{\substack{\tau_n \leq r\leq \tau_n^\prime\\|\hat{\theta}_n-\theta| \leq \tau_n^\prime}} g_{n,\tilde{\theta}_n}(r)(p_r(d\hat{\theta}_n)-p_{-r}(d\hat{\theta}_n))(\hat{\theta}_n-\tilde{\theta}_n)r\\
\end{split}
\]
The first addend can be lower bounded by
\[\int_{\substack{\tau_n \leq r\leq \tau_n^\prime\\|\hat{\theta}_n-\theta| \leq \tau_n^\prime}} g_{n,\tilde{\theta}_n}(r)(p_r(d\hat{\theta}_n)+p_{-r}(d\hat{\theta}_n)) \tau_n^2.\]
The second addend will be negligible in the final analysis. Indeed, we have that
 \[\begin{split}
     &\sum_{\bm{x} \in \{0,1\}^n}\left  |\Pi_{k=1}^n p_r(x_k) -\Pi_{k=1}^n p_{-r}(x_k) \right |\leq\\
     & \sum_{l=1}^{n}\sum_{\bm{x} \in \{0,1\}^n}p_{-r}(x_1)\cdots p_{-r}(x_{l-1})|p_r(x_l) -p_{-r}(x_l)| \cdot \\
     &\cdot q_{r}(x_{l+1})\cdots q_{r}(x_n)=\\
     &n \sum_{x = 0,1}|p_r(x_l) -p_{-r}(x_l)| \lesssim n \tau_n^{\prime 3},
 \end{split}
 \]
where in the last inequality we used that
\[
p_0(x)-p_{0}(x)=\dot{p}_0(x)+\dot{p}_{0}(x)=\Ddot{p}_0(x)-\Ddot{p}_{0}(x)=0
\]
for $x=0,1$ and the symbol $\lesssim$ means that the left hand side is less or equal than a constant times the right hand side. Therefore, the second term can be upper bounded by a constant times $n\tau_n^{\prime 5} =n^{-3/2+5(\epsilon+\beta)/2}$, which is $o(\tau_n^2)$ if $\epsilon<1/3-(2\alpha+5\beta)/3$.
To sum up, one has that
\[
\begin{split}
    &\mathbb{E}[(\hat{\theta}_n-\theta)^2] \geq  \mathbb{P}(|\hat{\theta}_n-\theta|>\tau_n^\prime)\tau_n^{\prime2}+\\
    &\int_{\substack{\tau_n \leq r\leq \tau_n^\prime\\|\hat{\theta}_n-\theta| \leq \tau_n^\prime}} \pi(d\tilde{\theta}_n)g_{n,\tilde{\theta}_n}(r)(p_r(d\hat{\theta}_n)+p_{-r}(d\hat{\theta}_n)) \tau_n^2 +o(\tau_n^2) \geq\\
    &\int_{\tau_n \leq r\leq \tau_n^\prime} \pi(d\tilde{\theta}_n)g_{n,\tilde{\theta}_n}(r)(p_r(d\hat{\theta}_n)+p_{-r}(d\hat{\theta}_n)) \tau_n^2+o(\tau_n^2) \geq \\
    &cC\tau_n^2 +o(\tau_n^2). 
\end{split}
\]
\end{proof}
We remark that, if all the derivatives of $p_r$ up to the $2s-1$-th order for some $s\geq 1$ vanish at $0$, then we get that
\[\sum_{x=0,1}|p_r(x)-p_{-r}(x)| \lesssim \tau_n^{\prime 2s+1}
\]
and we obtain the statement under the assumption that $\epsilon<(2s-1)/(2s+1)-(2\alpha+(2s+3)\beta)/(2s+1)$. Notice that in general we can pick $\alpha$ and $\beta$ arbitrarily small, hence the restriction on $\epsilon$ effectively becomes $\epsilon<(2s-1)/(2s+1)$, which does not preclude any value in the limit $s \rightarrow +\infty$.

\section{Proof of Proposition \ref{prop:displaced.null} on optimality of displaced null measurements}
\label{app:proof.prop:displaced.null}


Since this measurement setting depends on $n$, we need to look in more detail at the asymptotic behaviour of the estimation problem. 

We start by assuming that $\theta\in I_n$ and at the end of the proof we treat the case $\theta\notin I_n$ by employing the concentration bound in equation \eqref{eq:concentration}. 

Since $\theta\in I_n$, we can write $\theta =\tilde{\theta}_n + u/n^{1/2}$ with local parameter $u$ satisfying $|u|\leq n^{\epsilon}$. Then
$$
\theta-\theta^\prime_n= un^{-1/2} -n^{-1/2+3\epsilon} = O(n^{-1/2+3\epsilon})
$$ 
and 
\begin{eqnarray*}
p^{(n)}_\theta &=&
\sin^2(\theta-\theta^\prime_n) =(\theta-\theta^\prime_n)^2 + O(n^{-2+12\epsilon}) \\
&=&
n^{-1}(
u-n^{3\epsilon})^2+ O(n^{-2+12\epsilon}).
\end{eqnarray*}
From this we get that
\begin{equation}
\label{eq.u.exp}
u =
\frac{n^{3\epsilon}}{2} -\frac{n^{1-3\epsilon}}{2} p^{(n)}_\theta 
+O(n^{-\epsilon}) 
\end{equation}
where the $u^2$ term is negligible compared to $un^{3\epsilon}$ and the remainder is $O(n^{-\epsilon})$ for $\epsilon<1/10$.

The probability $p^{(n)}_\theta$ can be estimated by the empirical frequency \eqref{eq:p.hat}
whose distribution is the binomial ${\rm Bin}(p^{(n)}_\theta, n)$.
Taking into account that $\theta= \tilde{\theta}_n+ u/n^{1/2}$ and using \eqref{eq.u.exp} we define the estimator 
\begin{equation}\label{eq:final.estimator}
\hat{\theta}_n = \tilde{\theta}_n+ 
\frac{n^{-1/2+3\epsilon}}{2} -
\frac{n^{1/2-3\epsilon}}{2}\hat{p}_n 
\end{equation}
with $\hat{p}_n$ as in \eqref{eq:p.hat}.
Now from \eqref{eq.u.exp} we get
$$
\sqrt{n}(\hat{\theta}_n -\theta )=
\frac{n^{1-3\epsilon}}{2} (p^{(n)}_\theta-\hat{p}_n) + O(n^{-\epsilon}).
$$
Conditional to a certain value of $\tilde{\theta}_n$, $\hat{p}_n$ has a binomial distribution with parameters $p^{(n)}_\theta$ and the term $O(n^{-\epsilon})$ is deterministic, hence
$$
n\mathbb{E}[(\hat{\theta}_n -\theta)^2|\tilde{\theta}_n]=\frac{n^{2-6\epsilon}}{4}p_\theta^{(n)}(1-p_\theta^{(n)}) = \frac{1}{4}+o(1). 
$$
In order to study the convergence in law of $\sqrt{n}(\hat{\theta}_n-\theta)$, one can consider the conditional characteristic function of $n^{1-3\epsilon} (p^{(n)}_\theta-\hat{p}_n)/2$ instead (conditional to $\tilde{\theta}_n$, they only differ by a deterministic vanishing quantity):
\begin{eqnarray*}
&&\mathbb{E}_\theta \left[
\exp (ia n^{1-3\epsilon} (\hat{p}_n -p^{(n)}_\theta)/2)|\tilde{\theta}_n\right]\\
&&= 
\mathbb{E}_\theta
 \left[
\exp \left (\left .ia n^{-3\epsilon} \sum_{i=1}^n (X_i-p_\theta^{(n)})/2 \right )\right | \tilde{\theta}_n\right]
\\&&=
\mathbb{E}_\theta
 \left[
\exp (ia n^{-3\epsilon}  (X_1-p_\theta^{(n)})/2|\tilde{\theta}_n\right]^n \\
&&=\left (
 1- \frac{a^2}{8}n^{-6\epsilon}p_\theta^{(n)}(1-p_\theta^{(n)}) + o(n^{-1})
 \right)^n 
\\
 &&
 =\left (
 1- \frac{a^2}{8n} + o(n^{-1})
 \right)^n=
 e^{-a^2/8}+o(1).
\end{eqnarray*}

Notice that for every $a \in \mathbb{R}$
\[\begin{split}
\mathbb{E}_\theta[e^{ia \sqrt{n} (\hat{\theta}_n-\theta)}]&=e^{\frac{-a^2}{8}}\mathbb{P}(\theta \in I_n) \\
&+\int_{\theta \in I_n}p(d\tilde{\theta}_n|\theta)(\mathbb{E}_\theta[e^{ia \sqrt{n} (\hat{\theta}_n-\theta)}|\tilde{\theta}_n]-e^{\frac{a^2}{8}})\\
&+\int_{\theta \notin I_n}p(d\tilde{\theta}_n|\theta)\mathbb{E}_\theta[e^{ia \sqrt{n} (\hat{\theta}_n-\theta)}|\tilde{\theta}_n].\\
\end{split}
\]
Since $\mathbb{P}_\theta(\theta \notin I_n)$ goes to zero, the first term goes to $e^{\frac{-a^2}{8}}$ and the third one vanishes. The second term vanishes because $|(\mathbb{E}_\theta[e^{ia \sqrt{n} (\hat{\theta}_n-\theta)}|\tilde{\theta}_n]-e^{\frac{a^2}{8}})|\chi_{\{\theta \in I_n\}}$ can be upper bounded uniformly in $\tilde{\theta}_n$ by a sequence converging to $0$. Therefore we obtain the convergence of $\sqrt{n} (\hat{\theta}_n-\theta)$ to the normal, in distribution. For the convergece of the rescaled MSE we note that since $\theta,\tilde{\theta}_n$ are bounded, equation \eqref{eq:final.estimator} shows that the square error $n(\hat{\theta}_n -\theta)^2$ does not grow more than $n^2$;  using the fact that $\mathbb{P}_\theta(\theta\notin I_n)$ decays exponentially fast, one can remove the conditioning also in the convergence of the MSE.

\qed

\section{Proof of Lemma \ref{lem:linear}} \label{sec:lemproof}

In the present section we want to show that the statistical models $\ket{\Psi^n_{\bm{u}}}$ and $\ket{\tilde{\Psi}^n_{\bm{u}}}$ (which are the ensemble states corresponding to the models in Equations \eqref{eq:genmod} and \eqref{eq:linmod}) become equivalent in Le Cam distance when the the neighborhood of parameters considered shrinks around $\bm{\tilde{\theta}}_n$. Let us first compute the overlaps between states corresponding to the same local parameter $\bm{u}$ in the single copy scenario: expanding the unitary rotations one obtains
\[\begin{split}
&\langle \tilde{\psi}_{\bm{u}/\sqrt{n}}|\psi_{\bm{\tilde{\theta}}_n+\bm{u}/\sqrt{n}} \rangle= \left \langle 0\left  | \left (\mathbf{1}+i\frac{S(\bm{u})}{\sqrt{n}}-\frac{S(\bm{u})^2}{2n} +o\left ( \frac{1}{n}\right )\right )\right . \right .\cdot \\
&\cdot\left .\left (\mathbf{1}+i\frac{S(\bm{u})}{\sqrt{n}}-i\frac{T(\bm{u})}{\sqrt{n}}-\frac{S(\bm{u})^2}{2n} +o\left ( \frac{1}{n}\right )\right )0\right \rangle \\
&=1-i\langle 0 | T(\bm{u}) 0 \rangle/\sqrt{n} +o(1/n)=1+o(1/n),
\end{split}
\]
where
\begin{align*}
&S(\bm{u})=\sum_{j=1}^m u_j S_j, \\
&T(\bm{u})=\sum_{i,j=1}u_iu_j \sum_{k=1}^{d-1} \partial_{ij} f_k^q(0) \sigma_y^k-\partial_{ij} f_k^p(0) \sigma_x^k.
\end{align*}
Notice that the last equality in the computation of the overlap is true because $T(\bm{u})$ has zero expectation in $\ket{0}$. We remark that the error remains of the order of $o(1/n)$ and is uniform in $\bm{u}$ if $\|\bm{u}\| \leq n^\epsilon$ with $\epsilon <1/6$. Now we can conclude easily using the expressions of the trace norm between two pure states in terms of the overlap of two representative vectors, and noticing that
\[\langle \tilde{\Psi}^n_{\bm{u}}|\Psi^n_{\bm{u}} \rangle = \langle \tilde{\psi}_{\bm{u}/\sqrt{n}}|\psi_{\bm{\tilde{\theta}}_n+\bm{u}/\sqrt{n}} \rangle^n=(1+o(1/n))^n \rightarrow 1\]
uniformly in $\|\bm{u}\| \leq n^\epsilon$.

\qed
\section{Proof of Theorem \ref{thm:dnmgeneral}} \label{sec:lawsrb}
We first assume that ${\bm \theta} \in I_n$ where
$$I_n=\{\bm{\theta} \in \mathbb{R}^m:\|\bm{\theta}-\tilde{\bm{\theta}}_n\| \leq n^{-1/2+\epsilon}\}.$$

Recall that $|0\rangle = |\psi_{\tilde{\bm\theta}_n}\rangle$ is the preliminary estimator, and we denote $\ket{\tilde{0}}:=|0\rangle \otimes |0^\prime\rangle$ the first basis vector of an ONB $\tilde{B}:= \{ |\tilde{0}\rangle, \dots, |
\widetilde{d^2-1}\rangle\}$ in $\mathbb{C}^d \otimes \mathbb{C}^d$ which is chosen such that $|\tilde{1}\rangle, \dots , |\tilde{m}\rangle $ are vectors corresponding to the 
canonical variables $\tilde{Q}_1,\dots , \tilde{Q}_m$ which span the elements of the optimal unbiased set of observables ${\bf Z}^*$. Without loss of generality we can assume that 
$\{ |\tilde{1}\rangle,\dots ,|\widetilde{2d-1}\rangle\} $ form an ONB of the subspace 
$ \mathcal{L}: ={\rm Lin} \{|0\rangle \otimes |i^\prime\rangle, |i\rangle\otimes |0^\prime\rangle : i=1,\dots d-1\} $. The local state (of system and ancilla) can be written as 
\[
\ket{\tilde{\psi}_{\tilde{\bm{\theta}}_n+{\bm u}/\sqrt{n}}} =  e^{-i\sum_{k=1}^{2(d-1)} \left ( \tilde{f}_1^k\left (\frac{{\bm u}}{\sqrt{n}}\right ) \tilde{\sigma}^k_{y} -\tilde{f}_2^k\left (\frac{{\bm u}}{\sqrt{n}}\right ) \tilde{\sigma}^k_{x} \right )
}  |\tilde{0}  \rangle.
\]
where $\tilde{f}_{1,2}^k$ are smooth real valued functions and $\tilde{\sigma}^k_{x,y}$ are the Pauli operators for the vectors in the basis $\tilde{\mathcal{B}}$.

From the definition of the basis $\tilde{{\cal B}}$, the subspace $\mathcal{L}$ and of the matrix $T$ defined at the end of section \ref{sec:Holevo.Gaussian.shift} we have 
$$
(T^{-1})_{kj}=\sqrt{2}\partial_j\tilde{f}_1^k(0) \quad \text{for } j=1,\dots,m.
$$
 In particular, we note that
\begin{equation}
\frac{1}{2}\Tr(W(\tilde{\bm \theta}_n)TT^{ T })={\cal H}^{W(\tilde{\bm \theta}_n)}(\tilde{\bm \theta}_n).
\end{equation}
Expanding the unitary rotation, one has
\begin{align} \label{eq:Taylor}
     \ket{\tilde{\psi}_{\tilde{\bm{\theta}}+{\bm u}/\sqrt{n}}} &=\ket{\tilde{0}} + \frac{1}{\sqrt{2}}\sum_{k=1}^{m} \left ( \sum_{j=1}^m  T^{-1}_{kj} \frac{u_j}{\sqrt{n}} \right )\ket{\tilde{k}} +\nonumber\\
     & i\sum_{k=1}^{m} \left ( \sum_{j=1}^m  \partial_{j}f^k_2(0) \frac{u_j}{\sqrt{n}} \right )\ket{\tilde{k}}+\\
& + O(n^{-1+ 2 \epsilon}).\nonumber
\end{align}
The Taylor expansion for the vectors in the rotated basis is 
\begin{align} \label{eq:TEvk}
     \ket{v_j^{\delta_n}} &= \exp\left(-i\delta_n \left (\sum_{k=1}^{m} \tilde{\sigma}^k_y\right )\right) | \tilde{j} \rangle\\
     &=\begin{cases}\ket{\tilde{j}}-\delta_n \ket{\tilde{0}}+ O(n^{-1+6\epsilon}) & \text{ if } j=1,\dots,m \\
     \ket{\tilde{j}} & \text{ otherwise} \nonumber\end{cases}.
\end{align}
Therefore one obtain the following expression for the outcome probability measure
\[\begin{split}
p^{(n)}_{\bm u}(k)&=
\frac{1}{n}\left ( \sum_{j=1}^m  \frac{T^{-1}_{kj}}{\sqrt{2}} u_j -n^{3\epsilon}\right)^2 +\\
&\frac{1}{n}\left ( \sum_{j=1}^{m}  \partial_jf^k_2(0) u_j \right)^2+ O(n^{-3/2+9\epsilon})
\end{split}
\]
if $k=1,\dots,m$ and $p^{(n)}_{\bm u}(k)=O(n^{-1+2\epsilon})$ otherwise. Using the fact that $\|{\bm u}\|\leq n^{\epsilon}$ one can neglect the quadratic terms in ${\bm u}$ and write
\[
u_j=\sum_{k=1}^m T_{jk} \left (\frac{n^{3\epsilon}}{\sqrt{2}} - \frac{n^{1-3\epsilon}}{\sqrt{2}}p^{(n)}_{\bm u}(k) \right )+ O(n^{-\epsilon}).
\]
Moreover, from explicit computations one can see that for every $j \neq k=1,\dots, m$
\begin{align*}
&\frac{n^{2-6\epsilon}}{2}\mathbb{E}_{\bm \theta}[(\hat{p}_n(j) -p^{(n)}_{\bm{u}}(j))^2|\tilde{\bm{\theta}}_n]=\\
&\frac{n^{1-6\epsilon}}{2}p^{(n)}_{\bm{u}}(j)(1-p^{(n)}_{\bm{u}}(j)) =\frac{1}{2}+o(1),
\end{align*}
and
\begin{align*}
&\frac{n^{2-6\epsilon}}{2}\mathbb{E}_{\bm \theta}[(\hat{p}_n(j) -p^{(n)}_{\bm{u}}(j))(\hat{p}_n(k) -p^{(n)}_{\bm{u}}(k))|\tilde{\bm{\theta}}_n]=\\
&-\frac{n^{1-6\epsilon}}{2}p^{(n)}_{\bm{u}}(j)p^{(n)}_{\bm{u}}(k) = 0 +o(1). 
\end{align*}
Therefore
\[
\begin{split}
&n\mathbb{E}_{\bm \theta} [
L({\bm \theta}, 
\hat{\bm \theta}_n )^2|\tilde{\bm \theta}_n ]\\
&=\mathbb{E}_{\bm{\theta}}[\Tr((\hat{\bm{u}}_n-\bm{u})^TW(\tilde{\bm \theta}_n)(\hat{\bm{u}}_n-\bm{u}))|\tilde{\bm{\theta}}_n]+o(1)\\
&=\frac{n^{2-6\epsilon}}{2}\sum_{j,k=1}^m (T^{ T }W(\tilde{\bm \theta}_n) T)_{kj}\cdot \\
&\cdot \mathbb{E}_{\bm{\theta}}[(\hat{p}_n(j)-p_{\bm{u}}^{(n)}(j))(\hat{p}_n(k)-p_{\bm{u}}^{(n)}(k))|\tilde{\bm{\theta}}_n])\\
&+ o(1)\\
&=\frac{1}{2}\Tr(W(\tilde{\bm \theta}_n)TT^{ T })={\cal H}^{W(\tilde{\bm \theta}_n)}(\tilde{\bm \theta}_n)+o(1).
\end{split}\]

In order to derive the asymptotic normality result for $\sqrt{n}(\hat{{\bm \theta}}-{\bm \theta})$, we first consider the characteristic function of $n^{1-3\epsilon}(\hat{p}_{n}-p^{(n)}_{\bm \theta})/\sqrt{2}$: for every $\bm{a} \in \mathbb{R}^{d-1}$, one has
\begin{eqnarray*}
&&\mathbb{E}_{{\bm \theta}} \left[
\exp (i n^{1-3\epsilon} \bm{a} \cdot (\hat{p}_n -p^{(n)}_{{\bm \theta}})/\sqrt{2})|\tilde{\bm \theta}_n\right]\\
&&= 
\mathbb{E}_{{\bm \theta}}
 \left[
\exp \left (\left .i n^{-3\epsilon} \sum_{i=1}^{n} \frac{\bm{a}}{\sqrt{2}} \cdot(X_i-p^{(n)}_{{\bm \theta}}) \right )\right | \tilde{\bm \theta}_n\right]
\\&&=
\mathbb{E}_{{\bm \theta}}
 \left[
\exp \left( i n^{-3\epsilon}  \frac{\bm{a}}{\sqrt{2}} \cdot(X_1-p^{(n)}_{{\bm \theta}})\right )\Bigg|\tilde{\bm \theta}_n\right]^{n} \\
&&=\left (
 1- \frac{\|\bm{a} \|^2}{4}n^{-6\epsilon}p^{(n)}_{{\bm \theta}}(1-p^{(n)}_{{\bm \theta}}) + o(n^{-1})
 \right)^{n} 
\\
 &&
 =\left (
 1- \frac{\|\bm{a} \|^2}{4n} + o(n^{-1})
 \right)^{n}=
 e^{-\|\bm{a} \|^2/4}+o(1).
\end{eqnarray*}
Indeed, using that $\sqrt{n}(\hat{{\bm \theta}}-{\bm \theta})=n^{1-3\epsilon}T(\hat{p}_{n}-p^{(n)}_{\bm \theta})/\sqrt{2}$, one has that for every ${\bm a } \in \mathbb{R}^{m}$
\[
\mathbb{E}_{\bm \theta}[\exp(i\sqrt{n} {\bm a}\cdot (\hat{{\bm \theta}}-{\bm \theta}))|\tilde{\bm \theta}_n]=e^{-\frac{{\bm a}^T \cdot TT^T \cdot {\bm a}}{4}}+o(1).
\]

We can now remove the conditioning with respect to the preliminary estimate and take the limit for $n \rightarrow +\infty$ (we will only show the computations for the risk, but they are the same in the case of the characteristic function):
\[\begin{split}
&n\mathbb{E}_{\bm \theta} [
L({\bm \theta}, 
\hat{\bm \theta}_n )^2]={\cal H}^{W({\bm \theta})}({\bm \theta})\mathbb{P}_{\bm \theta}({\bm \theta} \in I_n)\\
&+ \int_{\theta \in I_n}p(d\tilde{\theta}_n|\theta)({\cal H}^{W(\tilde{\bm \theta}_n)}(\tilde{\bm \theta}_n)-{\cal H}^{W({\bm \theta})}({\bm \theta}))\\
&+\int_{\theta \in I_n}p(d\tilde{\theta}_n|\theta)(n\mathbb{E}_{\bm \theta} [
d({\bm \theta}, 
\hat{\bm \theta}_n )^2|\tilde{\bm \theta}_n ]-{\cal H}^{W(\tilde{\bm \theta}_n)}(\tilde{\bm \theta}_n))\\
&+\int_{\theta \notin I_n}p(d\tilde{\theta}_n|\theta)n\mathbb{E}_{\bm \theta} [
d({\bm \theta}, 
\hat{\bm \theta}_n )^2|\tilde{\bm \theta}_n ].\\
\end{split}
\]
The first term in the sum tends to ${\cal H}^{W({\bm \theta})}$, while all the other ones tend to $0$ because of the continuity of ${\cal H}^{W({\bm \theta})}$, the fact that $n\mathbb{E}_{\bm \theta} [
L({\bm \theta}, 
\hat{\bm \theta}_n )^2|\tilde{\bm \theta}_n ]-{\cal H}^{W(\tilde{\bm \theta}_n)}(\tilde{\bm \theta}_n)$ is uniformly bounded by a vanishing sequence on $I_n$ and that the last term can be upper bounded by a constant times $n\mathbb{P}_{\bm \theta}({\bm \theta} \notin I_n)$. The same reasoning shows unconditional asymptotic normality.

\qed

\section{Comparison between $\hat{\hat{\bm \theta}}_n$ and $\hat{\bm \theta}_n$} \label{app:matsu}

In this section we elucidate the connection between the measurement strategy that we propose and the optimal measurement for pure statistical models pointed out in \cite{Ma02}. Theorems 1 and 2 in \cite{Ma02} show that for every parameter value ${\bm \theta}$, there exists a measurement basis that allows to attain the Holevo bound at ${\bm \theta}$ \textit{in one shot}: given the optimal estimator ${\bf Z}$ of the limit Gaussian model at ${\bm \theta}$, one needs to consider the corresponding vectors $\ket{z_1}, \dots, \ket{z_m}$ via QCLT (see Eq. \eqref{eq:QCLTcorr}) and pick any orthonormal basis $\{\ket{b_k}\}_{k=0}^{m}$ of ${\rm span}_{\mathbb{R}}\{\ket{\psi_{\bm \theta}},\ket{z_1}, \dots, \ket{z_m}\}$ such that $\langle b_k|\psi_{\bm \theta}\rangle \neq 0$ for every $k$. The measurement in any orthonormal basis containing $\ket{b_k}_{k=0}^{m}$ is optimal and the estimator achieving the Holevo bound is given by
$$
\hat{\hat{\theta}}^i=\frac{\langle b_k|z_i\rangle}{\sqrt{2}\langle b_k|\psi_{\bm \theta}\rangle}+\theta^i$$
if $k$ is observed for $k=0,\dots, m$ and $0$ otherwise.

As in the case of the SLD, such an optimal measurement depends on the true parameter; in order to come up with a concrete estimation strategy, one needs a two step procedure. After producing a preliminary estimate $\tilde{\bm \theta}_n$ of the parameter, one would then choose an orthonormal basis $\{b_k\}_{k=0}^{m}$ of ${\rm span}_{\mathbb{R}}\{\ket{\psi_{\tilde{\bm \theta}_n}},\ket{z_1}, \dots, \ket{z_m}\}$ such that $\langle b_k|\psi_{\tilde{\bm \theta}_n}\rangle \neq 0$ for every $k$ and measure in any orthonormal basis containing $\{b_k\}_{k=0}^{m}$. The final estimator would be given by
\begin{equation} \label{eq:matsuest}
    \hat{\hat{\theta}}^i_n=\sum_{k=0}^{m}\frac{\langle b_k|z_i\rangle}{\sqrt{2}\langle b_k|\psi_{\tilde{\bm \theta}_n}\rangle}\hat{p}_n(k)+\tilde{\theta}_n^i,
\end{equation}
where $\hat{p}_n(k)$ is the empirical probability of observing $k$.

However, Theorem \ref{prop.null} shows that if $\{b_k\}_{k=0}^{m}$ is too close to be a null-basis, such a strategy does not even achieve a standard scaling due to identifiability problems. The basis $\{\ket{v_k^{\delta_n}}\}_{k=0}^{m}$ that we propose satisfies the assumptions above for being optimal at $\tilde{{\bm \theta}}_n$ and ensures an asymptotically optimal estimation precision; moreover, in this case $\hat{\hat{\bm \theta}}_n$ and $\hat{\bm \theta}_n$ are equivalent in the following sense.

\begin{proposition}
Let $\hat{\hat{{\bm \theta}}}_n$ and $\hat{\bm \theta}$ the estimators defined in Eq. \eqref{eq:matsuest} and Theorem \ref{thm:dnmgeneral}, respectively. Then the following holds true:
$$\lim_{n \rightarrow +\infty} n\mathbb{E}_{\bm \theta}[(\hat{\hat{\bm \theta}}_n-\hat{\bm \theta}_n)^{2}]=0.
$$
\end{proposition}
\begin{proof}
First we condition on $\tilde{\bm \theta}_n \in I_n$, where
$$I_n=\{\bm{\theta} \in \mathbb{R}^m:\|\bm{\theta}-\tilde{\bm{\theta}}_n\| \leq n^{-1/2+\epsilon)}\}.$$
Using Eq. \eqref{eq:TEvk} and $\ket{z_i}=\sum_{k=1}^{m}T_{ik}\ket{\tilde{k}}$, one has that
\[
\begin{split}
\langle b_k|z_i\rangle&=T_{ik}-\frac{\delta_n^2}{2}\sum_{j=1}^{m}T_{ij}+O(n^{-3/2+9\epsilon}),\\
\langle b_k|\psi_{\tilde{\bm \theta}}\rangle&=-\delta_n +O(n^{-3/2+9\epsilon}) 
\end{split}\]
for $k=1,\dots, m$ and
\[
\begin{split}
\langle b_0|z_i\rangle&=\delta_n \sum_{j=1}^mT_{ij}+O(n^{-3/2+9\epsilon}),\\
\langle b_0|\psi_{\tilde{\bm \theta}}\rangle&=1+O(n^{-1+6\epsilon}). 
\end{split}\]
Therefore we can write
\[\begin{split}
\hat{\hat{\theta}}^i_n-\tilde{\theta}^i_n&=\sum_{k=1}^{m} \frac{\langle b_k|z_i\rangle}{\sqrt{2}\langle b_k|\psi_{\tilde{\bm \theta}_n}\rangle}\hat{p}_n(k)\\
&=\frac{n^{-1/2 +3\epsilon}}{\sqrt{2}} \sum_{k=1}^{m}T_{ik} \hat{p}_n(0)-\frac{n^{1/2 -3\epsilon}}{\sqrt{2}}\sum_{k=1}^mT_{ik}\hat{p}_n(k)\\
&+\sum_{j=1}^m \frac{n^{-1/2+3\epsilon}}{4}T_{ij}\sum_{k=1}^m\hat{p}_n(k)+R\\
&=\hat{\bm\theta}^i_n-\tilde{\bm \theta}^i_n+R\\
&+\frac{n^{-1/2 +3\epsilon}}{\sqrt{2}} \sum_{k=1}^{m}T_{ik} (\hat{p}_n(0)-1) \quad(I)\\
&+\sum_{j=1}^m \frac{n^{-1/2+3\epsilon}}{4}T_{ij}\sum_{k=1}^m\hat{p}_n(k) \quad(II),
\end{split}
\]
where $R$ is a random variable whose standard deviation is $O(n^{-3/2+9\epsilon})$. Moreover, both $(I)$ and $(II)$ are negligible too: indeed, for every $k=0,\dots, m$ 
$$\mathbb{E}_{\bm \theta}[(\hat{p}_n(k)-p^{(n)}_{\bm\theta}(k))^2|\tilde{\bm \theta}_n] =o(1/n)
$$
and
$$
p^{(n)}_{\bm\theta}(k)=\delta_{0k}+O(n^{-1+6 \epsilon}).
$$
The statement follows removing the conditioning can be shown with the same technique as in the proof of Theorem \ref{thm:dnmgeneral}.
\end{proof}

\section{Proof of Proposition \ref{prop:optlqd}} \label{sec:multidimproof}

We denote by $I_n$ the set of states
$$\{\ket{\psi} :d_b(\ket{\psi}\bra{\psi}, \ket{\tilde{\psi}_n}, \bra{\tilde{\psi}_n}) \leq   n^{(1-\epsilon)/2}\}$$
and we assume that $\ket{\psi}$ belongs to $I_n$ (the converse can be dealt with as in the proof of Theorem \ref{thm:dnmgeneral}). Therefore, we can write
\[
\ket{\psi}= 
\exp\left(-i\sum_{k=1}^{d-1} ( u_1^k \sigma^k_{y} -u_2^k \sigma^k_{x} )/\sqrt{n}
\right) |\tilde{\psi}_n\rangle
\]
for some $\bm{u}=(u_1^1, u_2^1, \dots, u_{1}^{d-1}, u_2^{d-1})\in \mathbb{R}^{2(d-1)}$ that satisfies $\|{\bm u}\| =O(n^{\epsilon})$. Notice that for $j=1,\dots, d$ one has
\[\begin{split}
    p_{\bm{u}}^{(n)}(j)&=|\langle 
\psi_{{\bm u}/\sqrt{n}}| v^{\delta_n}_j\rangle|^2\\
&=\left | \left \langle j \left |\exp \left (i \delta_n\sum_{k=1}  \sigma_y^k \right )\cdot \right . \right . \right .\\
&\cdot \left . \left . \left .\left .\exp\left(-i\sum_{k=1}^{d-1} ( u_1^k \sigma^k_{y} -u_2^k \sigma^k_{x} )/\sqrt{n}
\right)  \right |\tilde{\psi}_n\right .\right \rangle\right |^2\\
&= (u_1^j/\sqrt{n}-\delta_n  )^2 +( u_2^j/\sqrt{n})^2 + O(n^{-2+12\epsilon}),\\
\end{split}\]
where the last equality is obtained expanding the matrix exponential. Analogously one obtains
$$ q_{\bm{u}}^{(n)}(j)=(u_1^j/\sqrt{n}  )^2 +( u_2^j/\sqrt{n}-\delta_n)^2 + O(n^{-2+12\epsilon})$$
for $j=1,\dots, d-1$. This implies that
\[u^j_1=\frac{n^{3\epsilon}}{2} - \frac{n^{1-3\epsilon}}{2}p^{(n)}_{\bm{u}}(j) + O(n^{-\epsilon})
\]
and
\[
u^j_2=\frac{n^{3\epsilon}}{2} - \frac{n^{1-3\epsilon}}{2}q^{(n)}_{\bm{u}}(j) + O(n^{-\epsilon}).
\]
Moreover, for every $j \neq k=1,\dots, d-1$ one has
\begin{align*}
&\frac{n^{2-6\epsilon}}{4}\mathbb{E}_{\ket{\psi}}[(\hat{p}_n(j) -p^{(n)}_{\bm{u}}(j))^2|\ket{\tilde{\psi}_n}]=\\
&\frac{n^{1-6\epsilon}}{2}p^{(n)}_{\bm{u}}(j)(1-p^{(n)}_{\bm{u}}(j)) =\frac{1}{2}+o(1), \\
\end{align*}
and
\begin{align*}
&\frac{n^{2-6\epsilon}}{4}\mathbb{E}_{\ket{\psi}}[(\hat{p}_n(j) -p^{(n)}_{\bm{u}}(j))(\hat{p}_n(k) -p^{(n)}_{\bm{u}}(k))|\ket{\tilde{\psi}_n}]=\\
&-\frac{n^{1-6\epsilon}}{2}p^{(n)}_{\bm{u}}(j)p^{(n)}_{\bm{u}}(k) = 0 +o(1). \\
\end{align*}
Another consequence is that for $\bm{a} \in \mathbb{R}^{d-1}$ one has
\begin{eqnarray*}
&&\mathbb{E}_{\ket{\psi}} \left[
\exp (i n^{1-3\epsilon} \bm{a} \cdot (\hat{p}_n -p^{(n)}_{\bm{u}})/2)|\ket{\tilde{\psi}_n}\right]\\
&&= 
\mathbb{E}_{\ket{\psi}}
 \left[
\exp \left (\left .i n^{-3\epsilon} \sum_{i=1}^{n/2} \bm{a} \cdot(X_i-p^{(n)}_{\bm{u}}) \right )\right | \ket{\tilde{\psi}_n}\right]
\\&&=
\mathbb{E}_{\ket{\psi}}
 \left[
\exp (i n^{-3\epsilon}  \bm{a} \cdot(X_1-p^{(n)}_{\bm{u}})|\ket{\tilde{\psi}_n}\right]^{n/2} \\
&&=\left (
 1- \frac{\|\bm{a} \|^2}{2}n^{-6\epsilon}p^{(n)}_{\bm{u}}(1-p^{(n)}_{\bm{u}}) + o(n^{-1})
 \right)^{n/2} 
\\
 &&
 =\left (
 1- \frac{\|\bm{a} \|^2}{2n} + o(n^{-1})
 \right)^{n/2}=
 e^{-\|\bm{a} \|^2/4}+o(1).
\end{eqnarray*}
The same can be proved for the other batch. Notice that
\[\begin{split}
&n d_b(\ket{\psi}\bra{\psi}, \ket{\hat{\psi}_n}, \bra{\hat{\psi}_n})^2=\|\hat{\bm{u}}-\bm{u}\|^2 + o(1)= \\
& \frac{n^{2-6\epsilon}}{4}\sum_{j=1}^{d-1}(\hat{p}_n(j)-p^{(n)}_{\bm{u}}(j))^2+(\hat{q}_n(j)-q^{(n)}_{\bm{u}}(j))^2 + o(1).\\
\end{split}
\]
Therefore, if $\ket{\psi} \in I_n$
\[
\mathbb{E}_{\ket{\psi}}[ d_b(\ket{\psi}\bra{\psi}, \ket{\hat{\psi}_n}, \bra{\hat{\psi}_n})^2|\ket{\tilde{\psi}_n}]= d-1+o(1)
\]
and for every $\bm{a} \in \mathbb{R}^{2(d-1)}$
\[
\mathbb{E}_{\ket{\psi}}[e^{i\bm{a} \cdot (\hat{\bm{u}}-\bm{u})}|\ket{\tilde{\psi}_n}] = e^{-\|\bm{a} \|^2/4}+o(1).
\]
The rest of the proof is similar to the one of Proposition \ref{prop:displaced.null}.
\section{Proof of Proposition \ref{th:QCRB-achievability-null}} \label{app:proofQCRBac}

In order to avoid confusion and conversely to the main text, in this proof we stress the dependence of $C$ and $B$ on the preliminary estimator $\tilde{\bm{\theta}}_n$.

As usual, we assume that $\bm{\theta} \in I_n$, where
\[
I_n=\{\bm{\theta} \in \mathbb{R}^m:\|\bm{\theta}-\tilde{\bm{\theta}}_n\| \leq n^{-1/2+\epsilon)}\};\]
then one can write $\bm{\theta}=\tilde{\bm{\theta}}_n+\bm{u}/\sqrt{n}$ for some $\bm{u}$ such that $\|u\|\leq n^{\epsilon}$ and the probability law of the $X_j$'s is given by
\[
\begin{split}
&p_{\bm{u}}^{(n)}(k)= \left (\sum_{j=1}^m c_{kj}(\tilde{\bm{\theta}}_n) \theta^j-g_k \delta_n  \right )^2  + O(n^{-3+9\epsilon})\\
&=g^2_k n^{-1+6\epsilon} - 2n^{-1+3\epsilon}\left (\sum_{j=1}^m c_{kj}(\tilde{\bm{\theta}}_n) u^j \right ) g_k +O(n^{-1+2\epsilon})\\\
\end{split}
\]
for $k=1,\dots, d-1$. Equivalently, using that $B(\tilde{\bm{\theta}}_n)C(\tilde{\bm{\theta}}_n)=\bm{1}$, we can write
\[
u^j=\sum_{k=1}^{d-1} b_{jk}(\tilde{\bm{\theta}}_n) \left ( \frac{g_kn^{3\epsilon}}{2} - \frac{n^{1-3\epsilon}}{2g_k} p^{(n)}(k)\right ) + O(n^{-\epsilon}).
\]
Therefore
\[
\begin{split}
&n\mathbb{E}_{\bm{\theta}}[(\hat{\bm{\theta}}_n-\bm{\theta})^T(\hat{\bm{\theta}}_n-\bm{\theta})|\tilde{\bm{\theta}}_n]=\mathbb{E}_{\bm{\theta}}[(\hat{\bm{u}}_n-\bm{u})^T(\hat{\bm{u}}_n-\bm{u})|\tilde{\bm{\theta}}_n]\\
&=\frac{n^{2-6\epsilon}}{4}BG\mathbb{E}_{\bm{\theta}}[(\hat{p}_n-p_{\bm{u}}^{(n)})(\hat{p}_n-p_{\bm{u}}^{(n)})^T|\tilde{\bm{\theta}}_n]GB^T+ o(1),
\end{split}\]
where $G$ is the diagonal matrix with entries given by $(1/g_k)_{k=1}^{d-1}$. Explicit computations show that
\[
\frac{n^{2-6\epsilon}}{4}G\mathbb{E}_{\bm{\theta}}[(\hat{p}_n-p_{\bm{u}}^{(n)})(\hat{p}_n-p_{\bm{u}}^{(n)})^T|\tilde{\bm{\theta}}_n] G= {\bf 1}/4+o(1).
\]
Therefore
$$
n\mathbb{E}_{\bm{\theta}}[(\hat{\bm{\theta}}_n-\bm{\theta})(\hat{\bm{\theta}}_n-\bm{\theta})^T|\tilde{\bm{\theta}}_n] = B(\tilde{\bm{\theta}}_n)B(\tilde{\bm{\theta}}_n)^T/4+o(1).$$
Notice that $B(\tilde{\bm{\theta}}_n)B(\tilde{\bm{\theta}}_n)^T/4=F(\tilde{\bm{\theta}}_n)^{-1}$: indeed, using the explicit expression of $B(\tilde{\bm{\theta}}_n)$
\[
\begin{split}
&B(\tilde{\bm{\theta}}_n)B(\tilde{\bm{\theta}}_n)^T =(C(\tilde{\bm{\theta}}_n)^TC(\tilde{\bm{\theta}}_n))^{-1}=4F(\tilde{\bm{\theta}}_n).
\end{split}\]
The rest of the proof is the same as in the one of Theorem \ref{thm:dnmgeneral} and uses the continuity of $F(\tilde{\bm{\theta}}_n)$.

\end{document}